\documentclass[a4paper,USenglish,cleveref, autoref, numberwithinsect]{lipics-v2021}

\usepackage{xspace}
\usepackage{tikz}
\usetikzlibrary{tikzmark}
\usetikzlibrary{arrows,trees,positioning,calc}
\usetikzlibrary{automata,shapes.arrows,backgrounds,fit}
\usetikzlibrary{decorations.pathmorphing,decorations.pathreplacing}
\usetikzlibrary{shapes}
\tikzstyle{rstate}=[state,ellipse]
\tikzset{>={latex}}

\usetikzlibrary{topaths,calc}

\usepackage{mathdots}
\usepackage[algo2e,boxruled,vlined,linesnumbered]{algorithm2e}
\usepackage{algorithm}
\usepackage{algpseudocode}
\usepackage{stmaryrd}
\usepackage{upgreek}
\usepackage{comment}
\usepackage{bbding}
\usepackage{xcolor}

\usepackage{thmtools}
\usepackage{thm-restate}
 
\usepackage{todonotes}
\usepackage[normalem]{ulem}

\title{Splitting Spanner Atoms:\\ A Tool for Acyclic Core Spanners} 

\titlerunning{Splitting Spanner Atoms}
\author{Dominik D. Freydenberger}{Loughborough University, Loughborough, United Kingdom}{}{0000-0001-5088-0067}{Supported by EPSRC grant EP/T033762/1.}

\author{Sam M. Thompson}{Loughborough University, Loughborough, United Kingdom}{}{https://orcid.org/0000-0002-3476-6739}{}

\authorrunning{D.\,D.\ Freydenberger and S.\,M.\ Thompson}

\Copyright{Dominik D. Freydenberger and Sam M. Thompson}



\relatedversion{Conference version: \cite{frey:split}.}
\acknowledgements{The authors would like to thank Justin Brackemann, and the anonymous reviewers for all their helpful comments and suggestions.}

\nolinenumbers 
\hideLIPIcs  

\EventEditors{Dan Olteanu and Nils Vortmeier}
\EventNoEds{2}
\EventLongTitle{25th International Conference on Database Theory (ICDT 2022)}
\EventShortTitle{ICDT 2022}
\EventAcronym{ICDT}
\EventYear{2022}
\EventDate{March 29--April 1, 2022}
\EventLocation{Edinburgh, UK}
\EventLogo{}
\SeriesVolume{220}
\ArticleNo{6}

\mathchardef\mhyphen="2D

\newcommand{\atom}{\mathsf{atom}}
\newcommand{\decomp}{\Psi}
\newcommand{\formulaSize}[1]{| #1 |}
\newcommand{\cqhead}[1]{\mathsf{Ans}(#1) \leftarrow}
\newcommand{\query}{P}
\newcommand{\conclog}{2\mathsf{FC \mhyphen CQ}}
\newcommand{\kconclog}[1]{#1\mathsf{FC \mhyphen CQ}}

\newcommand{\labelFunction}{\tau}

\newcommand{\brac}{\mathsf{BPat}}
\newcommand{\noconstr}{}
\newcommand{\sercq}{\mathsf{SERCQ}}
\newcommand{\cpfc}{\mathsf{FC \mhyphen CQ}}
\newcommand{\cpfcreg}{\mathsf{FC[REG] \mhyphen CQ}}
\newcommand{\concreg}{2\cpfcreg}

\newcommand{\subs}{\sigma}

\newcommand{\dynfo}{\ifmmode{\mathsf{DynFO}}\else{\textsf{DynFO}}\xspace\fi}
\newcommand{\dynprop}{\ifmmode{\mathsf{DynPROP}}\else{\textsf{DynPROP}}\xspace\fi}
\newcommand{\dyncq}{\ifmmode{\mathsf{DynCQ}}\else{\textsf{DynCQ}}\xspace\fi}
\newcommand{\dynucq}{\ifmmode{\mathsf{DynUCQ}}\else{\textsf{DynUCQ}}\xspace\fi}

\newcommand{\bigO}{\mathcal{O}}

\newcommand{\np}{\mathsf{NP}}
\newcommand{\pspace}{\mathsf{PSPACE}}

\newcommand{\select}{\zeta}
\newcommand{\join}{\bowtie}

\newcommand{\core}{\mathsf{core}}
\newcommand{\synrgx}{\mathsf{RGX_{sync}}}
\newcommand{\rgx}{\mathsf{RGX}}

\newcommand{\spn}[1]{[#1\rangle}
\newcommand{\fun}[1]{\llbracket #1 \rrbracket}

\newcommand{\spanner}[1]{\llbracket #1 \rrbracket}
 
\newcommand{\SVars}[1]{{\mathsf{Vars}\left(#1\right)}} 

\newcommand{\strucbra}[1]{( #1 )}

\newcommand{\bind}[2]{#1\{#2\}}

\newcommand{\strucvar}{\mathfrak{u}}

\newcommand{\fvar}{\mathsf{free}}
\newcommand{\var}{\mathsf{var}}

\newcommand{\splog}{\ifmmode{\mathsf{SpLog}}\else{\textsf{SpLog}}\xspace\fi}
\newcommand{\splogneg}{\ifmmode{\mathsf{SpLog}^{\neg}}\else{\textsf{SpLog}$^{\neg}$}\xspace\fi}
\newcommand{\dpcsplog}{\ifmmode{\mathsf{DPC}}\else{\textsf{DPC}}\xspace\fi}
\newcommand{\pcsplog}{\ifmmode{\mathsf{PC}}\else{\textsf{PC}}\xspace\fi}
\newcommand{\splogrx}{\ifmmode{\mathsf{SpLog_{rx}}}\else{\textsf{SpLog\textsubscript{rx}}}\xspace\fi}
\newcommand{\splognegrx}{\ifmmode{\mathsf{SpLog_{rx}^{\neg}}}\else{\textsf{SpLog$_{\mathsf{rx}}^{\neg}$}}\xspace\fi}
\newcommand{\dpcsplogrx}{\ifmmode{\mathsf{DPC{rx}}}\else{\textsf{DPC\textsubscript{rx}}}\xspace\fi}
\newcommand{\pcsplogrx}{\ifmmode{\mathsf{PC_{rx}}}\else{\textsf{PC\textsubscript{rx}}}\xspace\fi}

\newcommand{\ECrtext}{\textsf{EC\textsuperscript{reg}}}
\newcommand{\EC}{\ifmmode{\mathsf{EC}}\else{\textsf{EC}}\xspace\fi}
\newcommand{\ECr}{\ifmmode{\mathsf{EC^{reg}}}\else{\ECrtext}\xspace\fi}
\newcommand{\fc}{\mathsf{FC}}
\newcommand{\fcreg}{\mathsf{FC}[\mathsf{REG}]}

\newcommand{\regconst}{\mathbin{\dot{\in}}}

\newcommand{\logeq}{\mathbin{\dot{=}}}

\newcommand{\fo}{\mathsf{FO}}
\newcommand{\domain}{\mathsf{Dom}}

\newcommand{\cq}{\ifmmode{\mathsf{CQ}}\else{\textsf{CQ}}\xspace\fi}
\newcommand{\ucq}{\ifmmode{\mathsf{UCQ}}\else{\textsf{UCQ}}\xspace\fi}

\newcommand{\lang}{\mathcal{L}}

\newcommand{\emptyword}{\varepsilon}
\newcommand{\true}{\ifmmode{\mathsf{True}}\else{\textsf{True}}\xspace\fi}
\newcommand{\false}{\ifmmode{\mathsf{False}}\else{\textsf{False}}\xspace\fi}
\newcommand{\union}{\mathrel{\cup}}
\newcommand{\intersect}{\mathrel{\cap}}

\newcommand{\biglor}{\bigvee}
\newcommand{\powerset}[1]{\mathcal{P}(#1)}
\newcommand{\df}{:=}
\newcommand{\subword}{\sqsubseteq}
\newcommand{\ie}{i.\,e.\xspace}
\newcommand{\eg}{e.\,g.\xspace}

\begin{document}
\maketitle
\begin{abstract}
This paper investigates regex $\cq$s with string equalities ($\sercq$s), a subclass of core spanners. 
As shown by Freydenberger, Kimelfeld, and Peterfreund (PODS 2018), these queries are intractable, even if restricted to acyclic queries. 
This previous result defines acyclicity by treating regex formulas as atoms. 
In contrast to this, we propose an alternative definition by converting $\sercq$s into $\cpfc$s -- conjunctive queries in $\fc$, a logic that is based on word equations. 
We introduce a way to decompose word equations of unbounded arity into a conjunction of binary word equations. 
If the result of the decomposition is acyclic, then evaluation and enumeration of results become tractable. 
The main result of this work is an algorithm that decides in polynomial time whether an $\cpfc$ can be decomposed into an acyclic $\cpfc$. 
We also give an efficient conversion from synchronized $\sercq$s to $\cpfc$s with regular constraints. 
As a consequence, tractability results for acyclic relational $\cq$s directly translate to a large class of $\sercq$s. 
\end{abstract}
\section{Introduction}\label{sec:intro}
Document spanners were introduced by Fagin, Kimelfeld, Reiss, and Vansummeren~\cite{fag:spa} as a formalization of AQL, an information extraction query language used in IBM's SystemT. 
Informally, they can be described in two steps. 
First, so-called \emph{extractors} convert an input document, a word over a finite alphabet, into relations of so-called \emph{spans}. 
We assume the extractors to be \emph{regex formulas} (as described in~\cite{fag:spa}), which are regular expressions with capture variables. 
Consider the following example of a regex formula
\[\gamma(x) \df  \Sigma^* \cdot \bind{x}{(\mathtt{EBDT})\lor(\mathtt{ICDT})} \cdot \Sigma^* . \]
Given some input word, $\gamma(x)$ can be used to extract a unary relation of spans such that each span represents a factor of the input word that is either ``$\mathtt{EBDT}$'' or ``$\mathtt{ICDT}$''.

The second step is that the extracted relations are combined using a relational algebra. Classes of spanners can be defined by the choice of relational operators. 
\emph{Regular spanners} allow for union~$\union$, projection~$\pi$, and natural join~$\join$. 
Depending on how they are represented, regular spanners have been shown to be efficient. 
For example, if a regular spanner is given as a so-called vset-automaton, results can be enumerated with constant delay after linear time preprocessing~\cite{florenzano2018constant, amarilli2020constant}. 
However, if a regular spanner is given as a join of regex formulas, evaluation is intractable --
as shown in~\cite{freydenberger2018joining}, evaluation for spanners of the form $\query \df \pi_\emptyset ( \gamma_1 \join \gamma_2 \cdots \join \gamma_n)$ is $\np$-complete, 
even if $\query$ is acyclic.

\emph{Core spanners} extend regular spanners by allowing equality selection~$\select^=$, which checks whether two (potentially different) spans represent the same factor of the input document.
Even when core spanners are restricted to queries of the form $\pi_\emptyset \select^=_{x_1,y_1} \cdots \select^=_{x_m, y_m} \gamma$
for a single regex formula $\gamma$, the evaluation problem is $\np$-complete~\cite{fre:doc}. Therefore, both joins and equalities introduce computational hardness. 

Regex $\cq$s can be understood as the spanner version of relational $\cq$s, which are a central topic in database theory.
In each case, a conjunctive query is a  projection over a join of atoms.
Apart from the setting, the key difference is that while the tables for relational $\cq$s are usually part of the input, 
the tables for regex $\cq$s are defined implicitly through the regex formulas. 
Hence, while one could extract these tables and then perform a standard \cq over the extractions, the number of tuples in the materialized relations may be exponential. 
As a consequence, tractable restrictions on relational queries (such as \emph{acyclic $\cq$s}) do not lead to tractable fragments of regex $\cq$s~\cite{freydenberger2018joining}.

So-called $\sercq$s extend regex $\cq$s by also allowing string equality, thus allowing us to examine both previously discussed sources of intractability.
Consider the following~$\sercq$
\[\query \df \pi_{x,y} \; \select^=_{x,x'} \left( \gamma_{\mathsf{sen}}(z) \join \gamma_{\mathsf{prod}}(x) \join \gamma_{\mathsf{pos}}(y) \join \gamma_{\mathsf{factors}}(x, x', z) \join \gamma_{\mathsf{factor}}(y,z)  \right), \]
where we assume $\gamma_{\mathsf{sen}}$ extracts sentences,
 $\gamma_{\mathsf{prod}}$ extracts product names, $\gamma_{\mathsf{pos}}$ extracts positive sentiments (such as ``enjoyed''), and $\gamma_{\mathsf{factors}}(x,x',z)$ and $\gamma_{\mathsf{factor}}(y,z)$ ensure that $x$ and $x'$ are successive (but not necessarily consecutive) factors of $z$, and $y$ is a factor of $z$ respectively. 
Therefore, $\query$ extracts spans representing products that are mentioned twice within a sentence, along with a positive sentiment that appears in the same sentence.

Syntactic restrictions on conjunctive queries have been incredibly fruitful for finding tractable fragments. A well known result of Yannakakis~\cite{YannakakisAlgorithm} is that for \emph{acyclic} conjunctive queries, evaluation can be solved in polynomial time. Further research on the complexity of acyclic conjunctive queries~\cite{gottlob2001complexity} and the enumeration of results for acyclic conjunctive queries~\cite{bagan2007acyclic} has shown the efficacy of this restriction. On the other hand, for document spanners, such syntactic restrictions are yet to unlock tractable fragments. 

To address this gap, we consider a different approach and represent $\sercq$s as a conjunctive query fragment of the logic $\fcreg$, introduced by Freydenberger and Peterfreund~\cite{frey2019finite}. 
This logic is based on word equations, regular constraints, and first-order logic connectives. 
Consider the following $\fcreg$ conjunctive query 	
\[\varphi \df \cqhead{x, y} (z \logeq z_2 \cdot x \cdot z_3 \cdot x \cdot z_4) \land (z \logeq z_5\cdot y \cdot z_6)  \land  (z \regconst \gamma_{\mathsf{sen}}) \land (x \regconst \gamma_{\mathsf{prod}}) \land (y \regconst \gamma_{\mathsf{pos}}).\] 
If $\gamma_{\mathsf{sen}}$ is a regular expression that accepts sentences, $\gamma_{\mathsf{prod}}$ accepts a product name, and $\gamma_{\mathsf{pos}}$ accepts a positive sentiment, then $\varphi$ is ``equivalent'' to the previously given $\sercq$. 
They are not equivalent in a strict sense -- a key difference being that $\sercq$s reason over spans, whereas $\cpfcreg$s reason over factors of the input words. 
Reasoning over words does bring some advantages: 
For example, $\varphi$ simply uses relations of words (for example, $\gamma_\mathsf{prod}$) encoded as a regular expression, and if we wanted to do something analogous for regex-formulas, we would first have to extract the corresponding relation of spans. 

When dealing with word equations, we run into an issue that we already encountered for regex formulas: Their relations may contain an exponential number of tuples.
This is due to the unbounded arity of word equations. 
However, an $\fc$ atom can be considered shorthand for a concatenation term. 
For example, the word equation $y \logeq x_1 x_2 x_3 x_4$ can be represented as $y \logeq f(f(x_1, x_2), f(x_3, x_4))$ where $f$ denotes binary concatenation. 
This then lends itself to the ``decomposition'' of the word equation into a $\cq$ consisting of smaller word equations.
We can express the above word equation as $(y \logeq z_1 \cdot z_2) \land (z_1 \logeq x_1 \cdot x_2) \land (z_2 \logeq x_3 \cdot x_4)$. 
For~such a \emph{decomposition}, the relations defined by each word equation can be stored in linear space and we can enumerate them with constant delay.
Thus, if the resulting query is acyclic, then the tractability properties of acyclic conjunctive queries directly translate to the $\cpfc$.
 
\subparagraph*{Contributions of this paper}
The goal of this work is to bridge the gap between acyclic relational $\cq$s and information extraction.
To this end, we define $\cpfcreg$s, a conjunctive query fragment of $\fcreg$, and show show that any so-called synchronized $\sercq$ can be converted into an equivalent $\cpfcreg$ in polynomial time~(\cref{Prop:RGXtoPatCQ}).

We define the decomposition of an $\cpfc$ into a $\conclog$, where $\conclog$ denotes the set of $\cpfc$s where the right-hand side of each word equation is of at most length two. 
Our first main result is a polynomial-time algorithm that decides whether a \emph{pattern}\footnote{For the purposes of this introduction, a pattern can be considered a single $\fc$ atom.} can be decomposed into an acyclic $\conclog$~(\cref{polytime}). 

Building on this, we give a polynomial-time algorithm that decomposes an $\cpfc$ into an acyclic $\conclog$, or determines that this is not possible~(\cref{theorem:LVJoinTree}). 
As soon as we have an acyclic $\conclog$, the upper bound results for model checking and enumeration of results follow from previous work on relational acyclic $\cq$s~\cite{gottlob2001complexity, bagan2007acyclic}.

We mainly focus on $\cpfc$s (\ie, no regular constraints) due to the fact that we can add regular constraints for ``free''.
This is because regular constraints are unary predicates, and therefore can be easily incorporated into a join tree.
Thus, our work defines a class of $\cpfcreg$s for which model checking can be solved in polynomial time, and results can be enumerated with polynomial-delay (both in terms of combined complexity). 

Our approach offers a new research direction for tractable document spanners. 
Most of the current literature approaches regular spanners by ``compiling'' the spanner representation (regex formulas that are combined with projection, union, and joins) into a single automaton, where the use of joins can lead to a number of states that is exponential in the size of the original representation.
Instead, we look at decomposing $\fc$ conjunctive queries into small and tractable components.
This allows us to use the wealth of research on relational algebra, while also allowing for the use of the string equality selection operator. 

\subparagraph*{Related Work}
Regarding data complexity, Florenzano, Riveros, Vgarte, Vansummeren, and Vrgoc~\cite{florenzano2018constant} gave a constant-delay algorithm for enumerating the results of deterministic vset-automata, after linear time preprocessing.
Amarilli, Bourhis, Mengal, and Niewerth~\cite{amarilli2020constant} extended this result to  non-deterministic vset-automata.
Regarding combined complexity,
Freydenberger, Kimelfeld, and Peterfreund~\cite{freydenberger2018joining} introduced regex $\cq$s  and proved that their evaluation is $\np$-complete (even for acyclic queries), and that fixing the number of atoms and the number of string equalities in $\sercq$s allows for polynomial-delay enumeration of results.
Freydenberger, Peterfreund, Kimelfeld, and Kr{\"{o}}ll~\cite{pet:com} showed that non-emptiness for a join of two sequential regex formulas is $\np$-hard, under schemaless semantics, even for a single character document. 
Connections between the theory of concatenation and spanners have been considered in~\cite{fre:doc, fre:splog, frey2019finite}, which give many of the lower bound complexity results for core spanners. Schmid and Schweikardt~\cite{schmid2020purely} examined a subclass of core spanners called refl-spanners, which incorporate string equality directly into a regular spanner. 
Peterfreund~\cite{peterfreund2021grammars} considered extraction grammars, and gave an algorithm for unambiguous extraction grammars that enumerates results with constant-delay after quintic preprocessing.
\section{Preliminaries}\label{sec:prelim}
Let $\emptyset$ denote the \emph{empty set}, and for $n \geq 1$ let $[n] \df \{ 1,2,\dots,n \}$. 
Given a set $S$, we use $|S|$ for the \emph{cardinality} of $S$. 
If $S$ is a subset of $T$ then we write $S \subseteq T$ and if $S \neq T$ also holds, then $S \subset T$. 
We write $\powerset{S}$ for the powerset of $S$. 
The difference of two sets $S$ and $T$ is denoted as $S \setminus T$. 
If $\vec{x}$ is a tuple, we write $x \in \vec{x}$ to indicate that $x$ is a component of $\vec{x}$. 
Let $A$ be an alphabet. 
We use $|w|$ to denote the length of some word $w \in A^*$ and $\emptyword$ to denote the \emph{empty word}. 
The number of occurrences of $a \in A$ within $w$ is $|w|_a$. 
We write $u \cdot v$ or just $uv$ for the concatenation of words $u,v \in A^*$. 
If $u = p \cdot v \cdot s$ for $p,s \in A^*$ then $v$ is a \emph{factor} of $u$, denoted $v \sqsubseteq u$. 
If $u \neq v$ also holds, then $v \sqsubset u$. 
Let $\Sigma$ be an alphabet of \emph{terminal symbols} and let $\Xi$ be an infinite alphabet of \emph{variables}. 
We assume that $\Sigma \intersect \Xi = \emptyset$ and $|\Sigma| \geq 2$.

If $T \df (V,E)$ is a tree, then a path between $x_1 \in V$ and $x_n \in V$ is the shortest sequence of edges from $x_1$ to $x_n$. 
If $(\{x_1,x_2\}, \{x_2,x_3\}, \dots, \{x_{n-1}, x_n\})$ is a path, then we say a node $y$ \emph{lies} on this path if $y = x_j$ for some $j \in [n]$. 
We call the number of edges on a path from $x_1$ to $x_n$ the \emph{distance} between $x_1$ and $x_n$.

\subparagraph*{Document Spanners}
Given $w \df w_1 \cdot w_2 \cdots w_n$ where $w_i \in \Sigma$ for all $i \in [n]$, a so-called \emph{span} of $w$ is an interval $\spn{i,j}$ where $1 \leq i \leq j \leq n+1$. A span $\spn{i,j}$ defines a factor $w_{\spn{i,j}} \df w_i \cdot w_{i+1} \cdots w_{j-1}$ of $w$. 
Let $V \subset \Xi$, where $V$ is finite, and let $w \in \Sigma^*$. 
A \emph{$(V, w)$-tuple} is a function $\mu$ that maps each $x \in V$ to a span $\mu(x)$ of $w$. 
A \emph{spanner} $P$, with variables $V$, is a function that maps every $w \in \Sigma^*$ to a set $P(w)$ of $(V,w)$-tuples. 
By $\SVars{P}$, we denote the set of variables of~$P$. 

Like~\cite{fag:spa}, we use \emph{regex formulas} as the primary extractors. 
Regex formulas are an extension of regular expressions with so-called \emph{capture variables}. 
More formally: $\emptyset$, $\emptyword$, and $\mathtt{a}$ where $\mathtt{a} \in \Sigma$ are all regex formulas, and if $\gamma_1$ and $\gamma_2$ are regex formulas then so are $(\gamma_1 \cdot \gamma_2)$, $(\gamma_1 \lor \gamma_2)$, $(\gamma_1)^*$, and $\bind{x}{\gamma_1}$ where $x \in \Xi$. 
We use $\Sigma$ as a shorthand for $\biglor_{a \in \Sigma} a$. We can omit the parentheses when the meaning is clear. 
A variable binding $\bind{x}{\gamma}$ matches the same words as $\gamma$ and assigns the corresponding span of the input word to $x$. 
A regex formula is \emph{functional} if on every match, each variable is assigned exactly one span. We denote the set of functional regex formulas by $\rgx$. 
For $\gamma \in \rgx$, we use $\fun{\gamma}$ to define the corresponding spanner as follows. 
Every match of $\gamma$ on $w$ defines $\mu$, a $(\SVars{\gamma}, w)$-tuple, where for each $x \in \SVars{\gamma}$, we have that $\mu(x)$ is the span assigned to $x$. 
We use $\fun{\gamma}(w)$ to denote the set of all such $(\SVars{\gamma},w)$-tuples.
See~\cite{fag:spa} for more details.

We now define \emph{synchronized $\rgx$-formulas} (this follows the definition by Freydenberger, Kimelfeld, Kr{\"{o}}ll, and Peterfreund in~\cite{pet:com}). 
An expression $\gamma \in \rgx$ is \emph{synchronized} if for all sub-expressions of the form $(\gamma_1 \lor \gamma_2)$, no variable bindings occur in $\gamma_1$ or $\gamma_2$. 
We denote the class of synchronized $\rgx$-formulas by $\synrgx$. 

The motivation for synchronized $\rgx$-formulas is that non-synchronized formulas allow for ``hidden'' disjunctions within the atoms. This  goes (arguably) against the spirit of $\cq$s and (as shown in~\cite{pet:com}) leads to ``un-$\cq$-like'' behavior.

\begin{example}
Consider the regex formula $\gamma \df \Sigma^* \cdot \bind{x}{ \mathtt{a} \lor (\mathtt{b})^*} \cdot \bind{y}{ \Sigma^* } \cdot \Sigma^*$. 
We have that $\fun{\gamma}(w)$ contains those $\mu$ such that $\mu(x)$ is a factor of $w$ which is either an $\mathtt{a}$ or a sequence of~$\mathtt{b}$ symbols, and the span $\mu(y)$ occurs directly after $\mu(x)$. 
Since $\gamma$ is functional, and for every sub-expression of the form $(\gamma_1 \lor \gamma_2)$, we have that $\SVars{\gamma_1} = \SVars{\gamma_2} = \emptyset$, it follows that $\gamma$ is a \emph{synchronized regex formula}.
\end{example}

Essentially, a synchronized regex formula is functional if no variable is redeclared, and no variable is used inside of a Kleene star.

This is extended into a \emph{relational algebra} comprised of $\union$ (union), $\pi$ (projection), $\join$ (natural join), and $\select^=$ (string equality). Let $w \in \Sigma^*$, and let $P_1$ and $P_2$ be spanners. We say $P_1$ and $P_2$ are \emph{compatible} if $\SVars{P_1} = \SVars{P_2}$.  If two spanners $P_1$ an $P_2$ are compatible, then $(P_1 \union P_2)(w) \df P_1(w) \union P_2(w)$. For $Y \subseteq \SVars{P_1}$, the \emph{projection} $\pi_Y P_1(w)$ is defined as the restriction of all $\mu \in P_1(w)$ to the set of variables $Y$, and hence $\SVars{\pi_Y P_1} \df Y$. 

The \emph{natural join}, $P_1 \join P_2$, is obtained by defining $\SVars{P_1 \join P_2} \df \SVars{P_1} \union \SVars{P_2}$, and $(P_1 \join P_2)(w)$ as the set of all $(\SVars{P_1} \union \SVars{P_2},w)$-tuples for which there exists $\mu_1 \in P_1(w)$ and $\mu_2 \in P_2(w)$ such that $\mu_1(x) = \mu_2(x)$ for all $x \in \SVars{P_1} \intersect \SVars{P_2}$. The \emph{string equality operator} $\select^=_{x_1, x_2} P_1$ is defined by $\select^=_{x_1, x_2} P_1(w) \df \{ \mu \in P_1(w) \mid w_{\mu(x_1)} = w_{\mu(x_2)}\}$, where $\SVars{\select^=_{x_1, x_2} P_1} \df \SVars{P_1}$.

Given a class of regex-formulas $C$
and a spanner algebra $\mathsf{O}$, 
we use $C^\mathsf{O}$ to denote the set of spanner representations which can be constructed by repeated combinations of operators from $\mathsf{O}$ with a regex-formula from $C$. We write $\spanner{C^\mathsf{O}}$ to denote the closure of $\spanner{C}$ under $\mathsf{O}$.

The class of \emph{core spanners} (introduced by Fagin, Kimelfeld, Reiss, and Vansummeren~\cite{fag:spa}) is defined as $\spanner{\rgx^\core}$ where $\core \df \{ \pi, \select^=, \union, \join \}$. The class of \emph{regex $\cq$s with string equality} ($\sercq$s) is defined as expressions of the form:
\[ \query \df \pi_Y \left( \select^=_{x_1, y_1} \cdots \select^=_{x_l, y_l} (\gamma_1 \join \cdots \join \gamma_k) \right), \]

where $\gamma_i \in \rgx$ for all $i \in [k]$. We call an $\sercq$ a \emph{synchronized} $\sercq$ if every regex formula is a synchronized $\rgx$-formula.

\begin{example}
\label{example:regexCQ}
Consider $\query \df \select^=_{x_1,x_2} \left( \gamma_1 \join \gamma_2 \right)$ where  $\gamma_1 \df \Sigma^* \cdot \bind{x_1}{\Sigma^+} \cdot \mathtt{a} \cdot \Sigma^*$ and $\gamma_2 \df \Sigma^* \cdot \bind{x_2}{\Sigma^+} \cdot \mathtt{b} \cdot \Sigma^*$. Given $w \in \Sigma^*$, we have that $\fun{P}(w)$ contains those $\mu$ such that the factor $w_{\mu(x_1)}$ is non-empty, and is immediately followed by the symbol $\mathtt{a}$, the factor $w_{\mu(x_2)}$ is immediately followed by the symbol $\mathtt{b}$, and $w_{\mu(x_1)} = w_{\mu(x_2)}$. Since both $\gamma_1$ and $\gamma_2$ are synchronized, $\query$ is a synchronized $\sercq$.
\end{example}

\subparagraph*{Computational Model and Complexity Measures}\label{compModel}
We use the \emph{random access machine} model with uniform cost measures,
where the size of each machine word is logarithmic in the size of the input.
We represent factors of a word $w\in\Sigma^*$ as spans of $w$. This allows us to check whether $u = v$ for $u,v \sqsubseteq w$ in constant time after preprocessing that takes linear time and space~\cite{gus:alg,karkkainen2006linear} (see \cref{lemma:datastructure} for more details). 
The complexity results we state are in terms of \emph{combined complexity}. 
That is, both the query and the word are considered part of the input. 
When considering the enumeration of results for a query executed on a word, we say that we can enumerate results with \emph{polynomial-delay} if there exists an algorithm which returns the first result in polynomial time, the time between two consecutive results is polynomial, and the time between the last result and terminating is polynomial.

\section{Conjunctive Queries for FC}\label{sec:fccqs}
This section introduces $\cpfcreg$s, a conjunctive query fragment of $\fc$ with regular constraints. We give some complexity results regarding $\sercq$s and show an efficient conversion from synchronized $\sercq$s to $\cpfcreg$s.

A pattern is a word $\alpha \in (\Sigma \cup \Xi)^*$, and a \emph{word equation} is a pair $\eta \df (\alpha_L, \alpha_R)$ where $\alpha_L, \alpha_R \in (\Sigma \union \Xi)^*$ are patterns known as the \emph{left} and \emph{right} side respectively. 
We usually write such $\eta$ as $(\alpha_L \logeq \alpha_R)$. 
The length of a word equation, denoted $|(\alpha_L \logeq \alpha_R)|$, is $|\alpha_L| + |\alpha_R|$. 
A~\emph{pattern substitution} is a morphism $\subs \colon (\Sigma \cup \Xi)^* \rightarrow \Sigma^*$ such that $\subs(\mathtt{a}) = \mathtt{a}$  holds for all $\mathtt{a} \in \Sigma$. 
Since $\subs$ is a morphism, we have  $\subs(\alpha_1 \cdot \alpha_2) = \subs(\alpha_1) \cdot \subs(\alpha_2)$ for all $\alpha_1, \alpha_2 \in (\Sigma \cup \Xi)^*$. 

A~pattern substitution $\subs$ is a \emph{solution} to a word equation $(\alpha_L \logeq \alpha_R)$ if and only if $\subs(\alpha_L) = \subs(\alpha_R)$. 
When applying a pattern substitution $\subs$ to a pattern $\alpha$, we assume that its domain $\mathsf{dom}(\subs)$ satisfies $\mathsf{var}(\alpha) \subseteq \mathsf{dom}(\subs)$. 
Freydenberger and Peterfreund~\cite{frey2019finite} introduced $\fc$ as a first-order logic that is based on word equations. In the present paper, we do not consider the full logic $\fc$. Instead, we introduce its conjunctive queries. 

\begin{definition}
An $\cpfc$ is an $\fc$-formula of the form $\varphi(\vec{x}) \df \exists \vec{y} \colon \bigwedge_{i=1}^n \eta_i$, where $\eta_i \df (x_i \logeq \alpha_i)$, $x_i \in \Xi$, and $\alpha_i \in (\Sigma \union \Xi)^*$ for all $i \in [n]$. We use the shorthand $\varphi \df \cqhead{\vec{x}} \bigwedge_{i=1}^n \eta_i$ where $\vec{x}$ is the tuple of free variables. We call $\mathsf{Ans}(\vec{x})$ the \emph{head} of $\varphi$, and $\bigwedge_{i=1}^n \eta_i$ the \emph{body} of $\varphi$.
\end{definition}
We write $\varphi(\vec{x})$ to denote that $\vec{x}$ is the set of free variables of $\varphi$.
The set of all variables used in $\varphi$ is denoted by $\var(\varphi)$. 
We distinguish a variable $\strucvar \in \Xi$, called the \emph{universe variable},
that shall represent the input document~$w$.
The universe variable is not considered a free variable, and  we adopt the convention that $\strucvar\notin \var(\varphi)$ for all $\varphi$ (even if $\strucvar$ occurs in $\varphi$). 
Next, we define the semantics for $\cpfc$s.
\begin{definition}
For $\varphi \in \cpfc$ and a pattern substitution $\subs$ with $\var(\varphi) \union \{\strucvar\} \subseteq \mathsf{dom}(\subs)$, we define $\subs \models \varphi$ as follows: $\subs \models (\alpha_l \logeq \alpha_R)$ if $\subs(\eta_L) = \subs(\eta_R)$ and $\subs(x) \sqsubseteq \subs(\strucvar)$ for all $x \in \var(\alpha_L \logeq \alpha_R)$. 
For $\subs \models \exists x \colon \varphi$ we have that $\subs_{x \mapsto u} \models \varphi$ holds for some $u \sqsubseteq \subs(\strucvar)$, where $\subs_{x \mapsto u}$ is defined as $\subs_{x \mapsto u}(x) \df u$ and $\subs_{x \mapsto u}(y) \df \subs(y)$ for all $y \in (\Sigma \cup \Xi)$ where $y \neq x$. We use the canonical definition for conjunction. 
\end{definition}
Hence, for all $\subs\models\varphi(\vec{x})$, the universe for variables in $\var(\varphi)$ is the set of factors of $\subs(\strucvar)$.
If~$\varphi(\vec{x}) \in \cpfc$ and $w \in \Sigma^{*}$, then $\fun{\varphi}\strucbra{w}$ denotes the set of all $\subs(\vec{x})$ such that $\subs \models \varphi$ and~$\subs(\strucvar) = w$. 
When determining $\fun{\varphi}\strucbra{w}$ for a given $w$, we know that $\strucvar$ represents $w$, and  hence $\strucvar$ can be treated as a constant (see~\cite{frey2019finite} for more information on the role of the universe variable).
If $\varphi \in \fc$ is \emph{Boolean} (that is, it has no free variables), $\fun{\varphi}\strucbra{w}$ is either the empty set, or the set containing the empty tuple, which we interpret as $\false$ and $\true$, respectively. 

In~\cite{frey2019finite}, $\fc$ was extended to  $\fcreg$ by adding  \emph{regular constraints}. 
This allows for atoms of the form $(x \regconst \gamma)$, where $\gamma$ is a \emph{regular expression}; 
and $\subs \models (x \regconst \gamma)$ if and only if $\subs(x) \in \lang(\gamma)$
 and $\subs(x) \sqsubseteq \subs(\strucvar)$. 
We extend $\cpfc$ to $\cpfcreg$ in the same way.

\subparagraph{Complexity} 
We now define various decision problems for $\cpfc$ and $\cpfcreg$:
The \emph{non-emptiness problem} is, given $w \in \Sigma^*$ and $\varphi$, decide whether $\fun{\varphi} \strucbra{w} \neq \emptyset$. The \emph{evaluation problem} is, given $\subs$ and $\varphi$, decide whether $\subs \models \varphi$. The \emph{model checking problem} is the special case of  non-emptiness and evaluation that only considers Boolean queries, note that for Boolean queries $\domain(\subs) = \{ \strucvar \}$. 
Given $w \in \Sigma^*$ and~$\varphi$, the \emph{enumeration problem} is outputting all $\fun{\varphi} \strucbra{w}$. 
The \emph{containment problem} is, given $\varphi$ and~$\psi$, decide whether $\fun{\varphi}\strucbra{w} \subseteq \fun{\psi}\strucbra{w}$ for all $w \in \Sigma^*$. 
Previous results on patterns and $\fc$ (see~\cite{bre:inc,ehrenfreucht1979finding,frey2019finite}) directly imply the following.
\begin{restatable}[]{proposition}{CQlowerBounds}
	\label{CQlowerBounds}
For each of $\cpfc$ and $\cpfcreg$, the evaluation problem is $\np$-complete, and the containment problem is undecidable.
\end{restatable}
As discussed in~\cite{frey2019finite}, $\fc$ and $\fcreg$ can be evaluated analogously to relational first-order logic ($\fo$), by materializing the tables that are defined by the atoms and then proceeding ``as usual''. 
Hence, bounding the width of a formula (the maximum number of free variables in a subformula) bounds the size of the intermediate tables, and thereby the complexity of evaluation.
As the complexity of evaluating $\fc$ and $\fo$ are the same ($\pspace$-complete in general, $\np$-complete for the existential-positive fragment), it is no surprise that this correspondence also translates to conjunctive queries.
From \cref{sec:acycCQFC} on, we further develop this connection by finding tractable subclasses of $\cpfcreg$.

As containment for $\cq$s is decidable (although $\np$-complete), it can be used for query minimization (see Chapter~6 of \cite{abiteboul1995foundations}). But by~\cref{CQlowerBounds}, this does not apply to $\cpfc$.

\subparagraph{Document Spanners and FC-CQs} Our next goal is to establish a connection between $\sercq$s and $\cpfcreg$s. However, first we must overcome the fact that $\fc$ reasons over strings, whereas spanners reason over intervals of positions. We deal with this by defining the notion of an $\fc$-formula \emph{realizing} a spanner, as described in~\cite{fre:doc, fre:splog, frey2019finite}.

\begin{definition}
\label{defn:realizing}
A pattern substitution $\subs$ \emph{expresses} a $(V,w)$-tuple $\mu$, if for all $x \in V$, we have that $\domain(\subs) = \{ x^P, x^C \mid x \in V \}$, and $\subs(x^P) = w_{\spn{1,i}}$ and $\subs(x^C) = w_{\spn{i,j}}$ for the span $\mu(x) = \spn{i,j}$. An $\cpfcreg$ $\varphi$ \emph{realizes} a spanner $P$ if $\fvar(\varphi) =  \{ x^P, x^C \mid x \in \SVars{P} \}$ and $\subs\models\varphi$ for all $w \in \Sigma^*$ where $\subs(\strucvar) = w$, if and only if $\subs$ expresses some $\mu \in P(w)$.
\end{definition}

Less formally, for each $\mu \in P(w)$, we have that $\mu(x) = \spn{i,j}$ is uniquely represented by the prefix, $\subs(x^P) = w_{\spn{1,i}}$, and the content, $\subs(x^C) = w_{\spn{i,j}}$.

\begin{example}
Consider the following $\cpfcreg$.
\begin{multline*} 
\varphi \df \cqhead{x_1^P, x_1^C, x_2^P, x_2^C} (\strucvar \logeq x_1^P \cdot x_1^C \cdot \mathtt{a} \cdot s_1) \land (\strucvar \logeq x_2^P \cdot x_2^C \cdot \mathtt{b} \cdot s_2) \\ \land (x_1^C \logeq x_2^C) \land (x_1^C \regconst \Sigma^+) \land (x_2^C \regconst \Sigma^+).
\end{multline*}
We can see that $\varphi$ realizes the $\sercq$ given in~\cref{example:regexCQ}.
\end{example}

Recall  that synchronized $\sercq$s consist of $\rgx$-formulas that do not have variables within sub-expressions of the form $(\gamma_1 \lor \gamma_2)$. 
As we observe in the following result, a synchronized $\sercq$ can be efficiently translated into an equivalent $\cpfcreg$. 

\begin{restatable}[]{lemma}{RGXtoPatCQ}
\label{Prop:RGXtoPatCQ}
Given a synchronized $\sercq$ $\query$, we can construct in polynomial time an $\cpfcreg$ that realizes $\query$.
\end{restatable}

The proof of~\cref{Prop:RGXtoPatCQ} follows from~\cite{frey2019finite, fre:doc, fre:splog}. 
The converse of~\cref{Prop:RGXtoPatCQ} follows directly from~\cite{frey2019finite}. However, one would need to define how $\cpfcreg$-formulas can be realized by regex formulas closed under spanner algebra (details on this can be found in~\cite{fre:splog, frey2019finite}). We omit such a result as it is not the focus on this work. 

In this section, we have introduced $\cpfcreg$s, and shown an efficient conversion from synchronized $\sercq$s to $\cpfcreg$s.
Therefore, while the present paper mainly considers a tractable fragment of $\cpfcreg$, this tractability carries over to a subclass of $\sercq$s. 
\section{Acyclic Pattern Decomposition}\label{sec:decomp}
This section examines decomposing terminal-free patterns (\ie, patterns $\alpha\in\Xi^+$) into acyclic $\conclog$s, where $\conclog$ denotes the set of $\cpfc$s where each word equation has a right-hand side of at most length two. 
Patterns are the basis for $\cpfc$ atoms, and hence, this section gives us a foundation on which to investigate the decomposition of $\cpfc$s. 
We do not consider regular constraints, or patterns with terminals.
This is because regular constraints are unary predicates, and therefore can be easily added to a join tree; and terminals can be expressed through regular constraints.
We use $\conclog$s for two reasons. 
Firstly, binary concatenation is the most elementary form of concatenation, as it cannot be decomposed into further (non-trivial) concatenations. 
Secondly, this ensures that each word equation has very low width, and therefore we can store the tables in linear space and enumerate them with constant delay -- as shown in the following.

\begin{restatable}[]{proposition}{Datastructure}\label{lemma:datastructure}
Given $w \in \Sigma^*$, we can construct a data structure in linear time that, for $x,y,z \in \Xi$, allow us to enumerate $\fun{x \logeq y \cdot z}\strucbra{w}$ with constant-delay, and to decide in constant time if $\sigma\in\fun{x \logeq y \cdot z}\strucbra{w}$ holds.
\end{restatable}

Although the cardinality of $\fun{x \logeq y \cdot z}\strucbra{w}$ is cubic in $|w|$, \cref{lemma:datastructure} allows us to represent this relation in linear space. 
As we can query such relations in constant time, they behave ``nicer'' than relations in relational algebra.
Furthermore, after materializing the relations defined by each atom of an $\conclog$,~\cref{lemma:datastructure} allows us to treat the $\conclog$ as a relational conjunctive query.
We now introduce a way to \emph{decompose} a pattern into a conjunction of word equations where the right hand side of each atom is at most length two. 
We start by looking at a canonical way to decompose terminal-free patterns.

Let $\alpha \in \Xi^+$ be a terminal-free pattern. 
To decompose $\alpha$, first we factorize $\alpha$ so that it can be written using only binary concatenation
We define $\brac$, the set of all \emph{well-bracketed patterns}, recursively as follows:
\begin{definition}
$x \in \brac$ for all $x \in \Xi$, and if $\tilde\alpha, \tilde\beta \in \brac$, then $(\tilde\alpha \cdot \tilde\beta) \in \brac$.\footnote{For convenience, we tend use $\tilde\alpha$ to denote a bracketing of the pattern $\alpha \in \Xi^+$.}
\end{definition}
We extend the notion of a factor to a \emph{sub-bracketing}. 
We write $\tilde{\alpha} \sqsubseteq \tilde{\beta}$ if $\tilde{\alpha}$ is a factor of $\tilde{\beta}$ and $\tilde\alpha, \tilde\beta \in \brac$. 
Let $\alpha \in \Xi^+$, by $\brac(\alpha)$ we denote the set of all bracketings which correspond to the pattern $\alpha$ (\ie, if we remove the brackets, then the resulting pattern is $\alpha$). 
Every $\tilde{\alpha} \in \brac(\alpha)$ can be converted into an equivalent formula $\decomp_{\tilde\alpha} \in \conclog\noconstr$ using the following. 

\begin{definition}
\label{defn:conclogConversion}
While there exists $\tilde\beta \sqsubseteq \tilde\alpha$ where $\tilde\beta = (x \cdot y)$ for some $x,y \in \Xi$, we replace every occurrence of $\tilde\beta$ in $\tilde\alpha$ with a new, unique variable $z \in \Xi \setminus \var(\alpha)$ and add the word equation $(z \logeq x \cdot y)$ to $\decomp_{\tilde\alpha}$. 
When $\tilde\alpha = \tilde\beta$, we have that $z = \strucvar$.
\end{definition}

Therefore, up to renaming of variables, every $\tilde\alpha \in \brac$ has a corresponding formula $\decomp_{\tilde\alpha} \in \conclog\noconstr$. 
We call $\decomp_{\tilde\alpha}$ the \emph{decomposition} of $\tilde\alpha$. 
The decomposition can be thought of as a logic formula expressing a \emph{straight-line program} of the pattern (see~\cite{lohrey2012algorithmics} for a survey on algorithms for SLPs). 
We now give an example of \emph{decomposing} a bracketing.

\begin{example}
\label{example:decomp}
Let $\alpha \df x_1 x_2 x_1 x_1 x_2$ and let $\tilde{\alpha} \in \brac(\alpha)$ be defined as follows: 
\[ \tilde{\alpha} \df  (((x_1 \cdot x_2)\cdot x_1) \cdot (x_1 \cdot x_2)).\] 
We now list $\tilde\alpha$ after every sub-bracketing is replaced with a variable. We also give the corresponding word equation that is added to $\decomp_{\tilde\alpha}$.
\begin{align*}
& (( \underline{(x_1 \cdot x_2)} \cdot x_1) \cdot \underline{(x_1 \cdot x_2)}) && z_1 \logeq x_1 \cdot x_2 \\
& ( \underline{(z_1 \cdot x_1)} \cdot   z_1) && z_2 \logeq z_1 \cdot x_1 \\
& \underline{(z_3 \cdot z_1)} && \strucvar \logeq z_3 \cdot z_1 
\end{align*}
Therefore, we get the decomposition $\decomp_{\tilde\alpha} \in \conclog\noconstr$, which is defined as 
\[ \decomp_{\tilde\alpha} \df \cqhead{} (z_1 \logeq x_1 \cdot  x_2) \land (z_2 \logeq z_1 \cdot x_1) \land ( \strucvar \logeq z_2 \cdot z_1). \]
Notice that every sub-bracketing of $\tilde\alpha$ has a corresponding word equation in $\decomp_{\tilde\alpha}$. 
\end{example}

The decomposition of $\tilde\alpha$ is somewhat similar to the \emph{Tseytin transformations}, see \cite{prestwich2009cnf}, which transforms a propositional logic formula into a formula in \emph{Tseytin normal form}. 

Our next focus is to study which patterns can be decomposed into an \emph{acyclic} $\conclog\noconstr$. 

\begin{definition}[Join Tree]
\label{defn:join-tree}
A \emph{join tree} for $\decomp \in \conclog$ with body $\bigwedge_{i=1}^n \chi_i$ is an undirected tree $T \df (V, E)$, where $V \df \{  \chi_i \mid i \in [n] \}$, and for all $\chi_i, \chi_j \in V$, if $x \in \var(\chi_i)$ and $x \in \var(\chi_j)$, then $x$ appears in all nodes that lie on the path between $\chi_i$ and $\chi_j$ in $T$. 
\end{definition}

Note that we use $\chi$ (with indices) to denote atoms of a $\conclog$ to distinguish them from word equations with arbitrarily large right-hand sides -- which we denote by $\eta$ (with indices). 
We call $\decomp\in\conclog$ \emph{acyclic} if there exists a join tree for $\decomp$. Otherwise, we call $\decomp$ \emph{cyclic}.

\begin{definition}[Acyclic Patterns]
If $\decomp_{\tilde\alpha} \in \conclog\noconstr$ is a decomposition of $\tilde{\alpha} \in \brac$ and $\decomp_{\tilde\alpha}$ is acyclic, then we call $\tilde{\alpha}$ \emph{acyclic}.  If $\decomp_{\tilde\alpha}$ is cyclic, then we call $\tilde{\alpha}$ \emph{cyclic}. If there exists~$\tilde{\alpha} \in \brac(\alpha)$ which is acyclic, then we say that $\alpha$ is \emph{acyclic}. Otherwise, $\alpha$ is \emph{cyclic}.
\end{definition}

When determining whether a decomposition $\decomp_{\tilde\alpha} \in \conclog\noconstr$ is acyclic, we treat each word equation (atom) of $\decomp_{\tilde\alpha}$ as a single relational symbol. 
We also consider $\strucvar$ to be a constant symbol, since $\subs(\strucvar) = w$ always holds. 
This raises the question as to whether every pattern has an acyclic decomposition. The answers is no, as the following result shows.

\begin{restatable}[]{proposition}{cycPat}
\label{cycPat}
$x_1 x_2 x_1 x_3 x_1$ is a cyclic pattern, and $x_1 x_2 x_3 x_1$ is an acyclic pattern that has a cyclic bracketing.
\end{restatable}

This leads to the following question: \emph{Can we decide whether a pattern is acyclic in polynomial time?}
Given a pattern $\alpha \in \Xi^+$, we have that $|\brac(\alpha)| = C_{|\alpha|-1}$, where $C_i$ is the $i^{th}$ \emph{Catalan number}, see~\cite{olive1985catalan}. 
As the Catalan numbers grow exponentially, a straightforward enumeration of bracketings to finding an acyclic bracketing is not enough.

If $\decomp_{\tilde\alpha} \in \conclog\noconstr$ is a decomposition of $\tilde{\alpha} \in \brac(\alpha)$, then we call the variable $x \in \Xi$ which represents the whole pattern the \emph{root variable}. If $x$ is the root variable, then the atom $(x \logeq y \cdot z)$ for some $y,z \in \Xi$, is called the \emph{root atom}. So far, the root variable has always been~$\strucvar$. 
In~\cref{sec:acycCQFC}, different root variables will be considered.

Let $\decomp_{\tilde\alpha} \in \conclog\noconstr$ be the decomposition of $\tilde{\alpha} \in \brac(\alpha)$, where $\alpha \in \Xi^+$. We define the \emph{concatenation tree} of $\decomp_{\tilde\alpha}$ as a rooted, undirected, binary tree $\mathcal{T} \df (\mathcal{V}, \mathcal{E}, <, \Gamma, \labelFunction, v_r)$, where $\mathcal{V}$ is a set of nodes and $\mathcal{E}$ is a set of undirected edges. 
If $v$ and $v'$ have a shared parent node, then we use $v<v'$ to denote that $v$ is the left child and $v'$ is the right child of their shared parent. 
We also have $\Gamma \df \var(\decomp_{\tilde\alpha})$ and the function $\labelFunction \colon \mathcal{V} \rightarrow \Gamma$ that
\emph{labels} nodes from the concatenation tree with variables from $\var(\decomp_{\tilde\alpha})$. 
We use $v_r$ to denote the root of $\mathcal{T}$. 
The \emph{concatenation tree} of $\decomp_{\tilde\alpha}$ is defined as follows.

\begin{definition}
\label{defn:concatenationTree}
Let $\decomp_{\tilde\alpha} \df \cqhead{\vec{x}} \bigwedge_{i=1}^{n} (z_i \logeq x_i \cdot x_i')$ be a decomposition of $\tilde\alpha \in \brac(\alpha)$. 
We carry out the construction of a concatenation tree in two steps. 
First, we build a tree recursively. 
If $v \in \mathcal{V}$ is labeled with $z_i$ for $i \in [n]$, then there exists a left and right child of $v$ that are labeled with $x_i$ and $x_i'$ respectively.

In the second step, we prune the result of the above construction to remove redundancies. 
For each set of \emph{non-leaf nodes} that share a common label, we define an ordering $\ll$. 
If $\tau(v_i) = \tau(v_j)$ and the distance from the root of $\mathcal{T}$ to $v_j$ is strictly less than the distance from the root to $v_i$, then $v_j \ll v_i$. If $\tau(v_i) = \tau(v_j)$ and the distance from $v_r$ to $v_i$ and $v_j$ is equal, then $v_j \ll v_i$ if and only if $v_j$ appears to the \emph{right} of $v_i$. 
For each set of non-leaf nodes that share a common label, all nodes other than the $\ll$-maximum node are called \emph{redundant}. 
All descendants of redundant nodes are removed. 
\end{definition}

Concatenation trees for $\conclog\noconstr$s can be understood as a variation of \emph{derivation trees} for straight-line programs~\cite{lohrey2012algorithmics}.
While the pruning may seem somewhat unnatural, the concatenation tree of a decomposition is a useful tool that we shall use in~\cref{lemma:cycledistance} to characterize acyclic bracketings. 

Due to the pruning procedure, every non-leaf node represents a unique sub-bracketing. 
For every node $v$ with left child $v_l$ and right child $v_r$, we define $\atom(v) \df (\labelFunction(v) \logeq \labelFunction(v_l) \cdot \labelFunction(v_r))$.
Note that for any two non-leaf nodes $v,v' \in \mathcal{V}$ where $v \neq v'$, we have that $\atom(v) \neq \atom(v')$. We call $v \in \mathcal{V}$ an $x$-parent if one of the child nodes of $v$ is labeled $x$. If $v$ is an $x$-parent, then~$\atom(v)$ must contain the variable $x$.

\begin{definition}
Let $\decomp_{\tilde\alpha} \in \conclog\noconstr$ be the decomposition of $\tilde{\alpha} \in \brac$ and let $\mathcal{T}$ be the concatenation tree for $\decomp_{\tilde\alpha}$. 
For some $x \in \var(\decomp_{\tilde\alpha})$, we say that $\decomp_{\tilde\alpha}$ is \emph{$x$-localized} if all nodes that 
exist on the path between any two $x$-parents in $\mathcal{T}$ are also $x$-parents.
\end{definition}

Since there is exactly one concatenation tree for a decomposition $\decomp_{\tilde\alpha} \in \conclog$ of $\tilde\alpha \in \brac$, we can say $\decomp_{\tilde\alpha}$ is $x$-localized without referring to the concatenation tree of $\decomp_{\tilde\alpha}$.

\begin{example}
\label{example:concatTree}
Consider the pattern $\alpha \df x_1 x_2 x_1 x_2$ and the following two bracketings: \[\tilde\alpha_1 \df ( (x_1 \cdot x_2) \cdot (x_1 \cdot x_2) ) \text{ and } \tilde\alpha_2 \df (((x_1 \cdot x_2) \cdot x_1) \cdot x_2).\] The bracketing $\tilde\alpha_1$ is decomposed into $\decomp_1 \df \cqhead{} (z_1 \logeq x_1 \cdot x_2) \land (\strucvar \logeq z_1 \cdot z_1)$ and $\tilde\alpha_2$ is decomposed into $\decomp_2 \df \cqhead{} (z_1 \logeq x_1 \cdot x_2) \land (z_2 \logeq z_1 \cdot x_1) \land (\strucvar \logeq z_2 \cdot x_2)$. The concatenation trees for $\decomp_1$ and $\decomp_2$ are given in~\cref{fig:concatTree}. The label for each node is given in parentheses next to the corresponding node. We can see that $\atom(v_2) = (z_1 \logeq x_1 \cdot x_2)$. It follows that~$\decomp_2$ is $x_1$-localized, but $\decomp_2$ is not $x_2$-localized. Observe that $v_3 \ll v_2$, since $v_2$ appears to the left of $v_3$. Therefore, $v_3$ does not have any descendants, since it is a \emph{redundant node}.

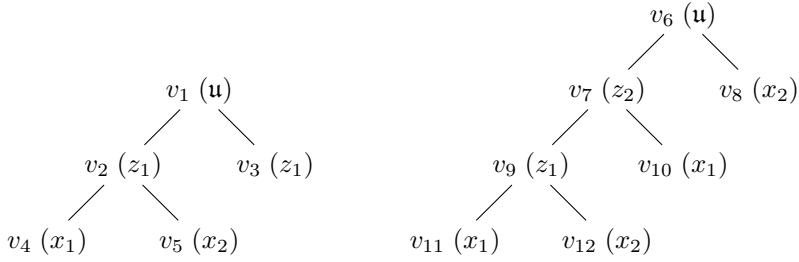
\begin{figure}
\begin{tikzpicture}[shorten >=1pt,->]
\tikzstyle{vertex}=[rectangle,fill=white!25,minimum size=12pt,inner sep=2pt]
\node[vertex] (1) at (0,0) {$v_1 \; (\strucvar)$};
\node[vertex] (2) at (1,-1)   {$v_3 \; (z_1)$};
\node[vertex] (3) at (-1,-1)  {$v_2 \; (z_1)$};
\node[vertex] (4) at (-2,-2) {$v_4 \; (x_1)$};
\node[vertex] (5) at (0,-2) {$v_5 \; (x_2)$};

\path [-](1) edge node[left] {} (2);
\path [-](1) edge node[left] {} (3);
\path [-](3) edge node[left] {} (4);
\path [-](3) edge node[left] {} (5);
\end{tikzpicture}\hspace{1cm}
\begin{tikzpicture}[shorten >=1pt,->]
\tikzstyle{vertex}=[rectangle,fill=white!25,minimum size=12pt,inner sep=2pt]
\node[vertex] (1) at (0,0) {$v_6 \; (\strucvar)$};
\node[vertex] (2) at (1,-1)   {$v_8 \; (x_2)$};
\node[vertex] (3) at (-1,-1)  {$v_7 \; (z_2)$};
\node[vertex] (4) at (-2,-2) {$v_9 \; (z_1)$};
\node[vertex] (5) at (0,-2) {$v_{10} \; (x_1)$};
\node[vertex] (6) at (-3,-3) {$v_{11} \; (x_1)$};
\node[vertex] (7) at (-1,-3) {$v_{12} \; (x_2)$};

\path [-](1) edge node[left] {} (2);
\path [-](1) edge node[left] {} (3);
\path [-](3) edge node[left] {} (4);
\path [-](3) edge node[left] {} (5);
\path [-](4) edge node[left] {} (6);
\path [-](4) edge node[left] {} (7);
\end{tikzpicture}
\caption{\label{fig:concatTree}Concatenation trees for the decompositions of $((x_1 \cdot x_2) \cdot (x_1 \cdot x_2))$ and $(((x_1 \cdot x_2) \cdot x_1) \cdot x_2)$. This figure is used to illustrate~\cref{example:concatTree}.}
\end{figure}
\end{example}

Utilizing concatenation trees for the decomposition $\decomp_{\tilde\alpha}$ of $\tilde\alpha\in\brac(\alpha)$, and the notion of $\decomp_{\tilde\alpha}$ being $x$-localized for $x \in \var(\decomp_{\tilde\alpha})$, we are now able to state sufficient and necessary conditions for $\alpha \in \Xi^+$ to be acyclic. 

\begin{restatable}[]{lemma}{cycledistance}
\label{lemma:cycledistance}
The decomposition $\decomp_{\tilde\alpha} \in \conclog\noconstr$ of $\tilde{\alpha} \in \brac(\alpha)$ is acyclic if and only if $\decomp_{\tilde\alpha}$ is $x$-localized for every $x \in \var(\decomp_{\tilde\alpha})$.
\end{restatable}

The proof of the if-direction is rather straightforward: Take the concatenation tree of $\decomp_{\tilde\alpha}$, replace each non-leaf node $v \in \mathcal{V}$ with $\atom(v)$, then remove all leaf nodes from the concatenation tree of $\decomp_{\tilde\alpha}$. 
This gives us a join tree for $\decomp_{\tilde\alpha}$. 
The only-if direction for~\cref{lemma:cycledistance} is somewhat more technical. 
This is because we need to prove this direction for the most general join tree of~$\decomp_{\tilde\alpha}$.
We prove this by contradiction, showing that there does not exist a valid label for certain non-leaf nodes of the concatenation tree if $\decomp_{\tilde\alpha}$ is not $x$-localized for some variable $x \in \var(\decomp_{\tilde\alpha})$.

Refering back to~\cref{example:concatTree}, we see that $\decomp_2$ is not $x_2$-localized and therefore $\decomp_2$ is cyclic, whereas we have that $\decomp_1$ is $x$-localized for all $x \in \var(\decomp_1)$ and hence $\decomp_1$ is acyclic. 

\begin{restatable}[]{theorem}{polytime}
\label{polytime}
Whether $\alpha \in \Xi^+$ is acyclic can be decided in time $\bigO(|\alpha|^7)$.
\end{restatable}

We prove~\cref{polytime} by giving a bottom-up algorithm that continuously adds larger acyclic subpatterns of $\alpha$ to a set. To determine whether concatenating two acyclic subpatterns results in a larger acyclic subpattern, we also keep an edge relation and check whether $x$ is localized, see~\cref{lemma:cycledistance}. We terminate the algorithm when the edge relation has reached a fixed-point. In the proof of~\cref{polytime}, we also show that if $\alpha$ is acyclic, then we can construct a concatenation tree for a decomposition for $\tilde\alpha \in \brac(\alpha)$ in $\bigO(|\alpha|^7)$ time.
\section{Acyclic FC-CQs}\label{sec:acycCQFC}
In this section, we generalize from decomposing patterns to decomposing $\cpfc$s. 
The main result of this section is a polynomial-time algorithm to determine whether an $\cpfc$ can be decomposed into an acyclic $\conclog$.
We do this to find a notion of acyclicity for $\cpfc$s such that the resulting fragment is tractable.

Decomposing a word equation $(x \logeq \alpha)$ where $x \in \Xi$ and $\alpha \in (\Xi \setminus \{x\})^+$ is analogous to decomposing $\alpha$, but whereas $\strucvar$ is the root variable when decomposing a pattern, we use $x$ as the root variable when decomposing $(x \logeq \alpha)$. 

If every atom of $\varphi \in \cpfc$ is acyclic, then $\varphi$ does not necessarily have tractable model checking.
If this were the case, then any decomposition $\decomp_{\tilde\alpha} \in \conclog$ of some $\tilde\alpha \in \brac$ would have tractable model checking (because every word equation of the form $z \logeq x \cdot y$ is acyclic).
This would imply that the membership problem for patterns can be solved in polynomial time, which contradicts~\cite{ehrenfreucht1979finding}, unless $\mathsf{P} = \mathsf{NP}$.
Furthermore, if we define $\varphi \in \cpfc$ to be acyclic if there exists a join tree for $\varphi$ where every word equation is an atom, then model checking for $\varphi$ is not tractable.
To show this, consider $\varphi \df \cqhead{} (\strucvar \logeq \alpha)$. 
Model checking for $\varphi$ is equivalent to the membership problem for $\alpha$, which is $\mathsf{NP}$-complete~\cite{ehrenfreucht1979finding}.
Therefore, we require a more refined notion of acyclicity for $\cpfc$s.

In~\cref{sec:decomp}, we studied the decomposition of terminal-free patterns. 
If $\varphi$ is an $\cpfc$ with the body $\cqhead{\vec{x}} \bigwedge_{i=1}^n \eta_i$, then the right-hand side of some $\eta_i$ may not be terminal-free. 
Therefore, before defining the decomposition of $\cpfcreg$s, we define a way to \emph{normalize} $\cpfcreg$s in order to better utilize the techniques of~\cref{sec:decomp}. 

\begin{definition}
\label{defn:normalization}
We call an $\cpfc$ with body $\bigwedge_{i=1}^n (x_i \logeq \alpha_i)$ \emph{normalized} if for all $i,j \in [n]$, we have $\alpha_i \in \Xi^+$, $x_i \notin \var(\alpha_i)$, $\strucvar \notin \var(\alpha_i)$, and $\alpha_i = \alpha_j$ if and only if $i = j$.

An $\cpfcreg$ with body $\bigwedge_{i=1}^n (x_i \logeq \alpha_i) \land \bigwedge_{j=1}^m (y_j \regconst \gamma)$ is \emph{normalized} if the subformula $\bigwedge_{i=1}^n (x_i \logeq \alpha_i)$ is normalized.
\end{definition}

Since we are interested in polynomial time algorithms, the following lemma allows us to assume that all $\cpfc\noconstr$s are normalized without affecting any claims about complexity.

\begin{restatable}[]{lemma}{normalization}
\label{lemma:normalization}
Given $\varphi \in \cpfcreg$, we can construct an equivalent, normalized $\cpfcreg$ in time $\bigO(|\varphi|^2)$.
\end{restatable}

To prove~\cref{lemma:normalization} we use a simple re-writing procedure. 
We replace every terminal factor in our formula with a new variable, and use a regular constraint to determine which terminal word that variable represents.
If $\sigma$ is a morphism that satisfies $(x \logeq \alpha)$ for some $\alpha \in \Xi$, then $|\sigma(x)| = |\sigma(\alpha)|$. 
Therefore, if $x \in \alpha$, then $|\sigma(x)| = |\sigma(\alpha_1)| + |\sigma(x)| + |\sigma(\alpha_2)|$ where $\alpha = \alpha_1 \cdot x \cdot \alpha_2$.
We can then determine that $\sigma(\alpha_1) \cdot \sigma(\alpha_2) = \emptyword$.
Hence, $x \logeq \alpha$ can be replaced with $(x \logeq y) \land \bigwedge_{z \in \var(\alpha_1 \cdot \alpha_2)}(z \regconst \emptyword)$ where $y$ is a new and unique variable.
An analogous method is used if $\strucvar \in \var(\alpha)$.

\begin{example}
We define an $\cpfcreg$ along with an equivalent normalized $\cpfcreg$:
\begin{align*}
\varphi \df& \cqhead{\vec{x}} (x_1 \logeq x_2 \cdot \strucvar \cdot x_2) \land (x_4 \logeq x_4) \land  (x_3 \logeq \mathtt{aab}), \\
\varphi' \df& \cqhead{\vec{x}} (\strucvar \logeq x_1) \land (x_2 \regconst \emptyword) \land (x_4 \logeq z_2) \land  (x_3 \logeq z_1) \land (z_1 \regconst \mathtt{aab}).
\end{align*}
\end{example}

We now generalize the process of decomposing patterns to decomposing $\cpfc\noconstr$s. 
For every $\cpfc\noconstr$ $\varphi \df \cqhead{\vec{x}} \bigwedge_{i=1}^{n} \eta_i$, we say that a   $\conclog$ $\decomp_\varphi \df \cqhead{\vec{x}} \bigwedge_{i=1}^{n}  \decomp_i$ is a \emph{decomposition} of $\varphi$ if every $ \decomp_i$ is a decomposition of $\eta_i$ and, for all $i,j\in [n]$ with $i\neq j$, the sets of introduced variables for $\decomp_i$ and $\decomp_j$ are disjoint.

\begin{example}
\label{example:LVDecomp}
Let $\varphi \in \cpfc$ be defined as follows:
\[\varphi \df \cqhead{\vec{x}} (x_1 \logeq y_1 \cdot y_2 \cdot y_3) \land (x_2 \logeq y_2 \cdot y_3 \cdot y_3 \cdot y_4).\]

We now consider the following decompositions for each word equation of $\varphi$:
\[ \decomp_1 \df (x_1 \logeq y_1 \cdot z_1) \land (z_1 \logeq y_2 \cdot y_3), \text{ and } \decomp_2 \df (x_2 \logeq z_2 \cdot y_4) \land (z_2 \logeq z_3 \cdot y_3) \land (z_3 \logeq y_2 \cdot y_3). \]

Therefore, $\decomp_\varphi \df \cqhead{\vec{x}} \decomp_1 \land \decomp_2$ is a decomposition of $\varphi$. 
\end{example}

\begin{definition}[Acyclic $\cpfc$s]
If $\decomp_\varphi \in \conclog\noconstr$ is a decomposition of $\varphi \in \cpfc$, we say that $\decomp_\varphi$ is \emph{acyclic} if there exists a join tree for $\decomp_\varphi$. Otherwise, $\decomp_\varphi$ is \emph{cyclic}. If there exists an acyclic decomposition of $\varphi$, then we say that $\varphi$ is \emph{acyclic}. Otherwise, $\varphi$ is \emph{cyclic}. 
\end{definition}

Recall that, since $\strucvar$ is always mapped to $w$, we can consider $\strucvar$  a constant symbol.
Therefore, if $T \df (V,E)$ is a join tree for some decomposition of $\varphi$, then there can exist two nodes that both contain $\strucvar$, yet it is not necessary for all nodes on the path between these two nodes to also contain $\strucvar$. 
Referring back to~\cref{example:LVDecomp}, we can see that $\varphi$ is acyclic by executing the GYO algorithm on the decomposition (see Chapter 6 of~\cite{abiteboul1995foundations} for more information on acyclic joins). 
Our next focus is to study which $\cpfc$s are acyclic, and which are not.

\begin{restatable}[]{lemma}{subtree}
\label{lemma:subtree}
If $\decomp_\varphi \in \conclog\noconstr$ is a decomposition of $\varphi \df \cqhead{\vec{x}} \bigwedge_{i=1}^n \eta_i$, and we have a join tree $T \df (V,E)$ for $\decomp_\varphi$, then we can partition $T$ into $T^1, T^2, \dots T^n$ such that for each $i \in [n]$, we have that $T^i$ is a join tree for a decomposition of $\eta_i$.
\end{restatable}

To prove~\cref{lemma:subtree}, we consider a join tree $T \df (V,E)$ for the acyclic decomposition $\decomp_\varphi \in \conclog\noconstr$ of $\varphi \in \cpfc\noconstr$, along with the induced subgraph of $T$ on the set of atoms for a decomposition of a single atom of $\varphi$. 
We show that this subgraph is connected, and since the introduced variables are disjoint for separate atoms of $\varphi$, this forms a partition on $T$. 

Let $\varphi \df \cqhead{\vec{x}} \bigwedge_{i=1}^n \eta_i$ be a normalized $\cpfc\noconstr$. A join tree $T \df (V,E)$ for $\varphi$ where $V = \{\eta_i \mid i \in [n]\}$ is called a \emph{weak join tree}. If there exists a weak join tree for $\varphi$, then we say that $\varphi$ is \emph{weakly acyclic}. Otherwise, $\varphi$ is \emph{weakly cyclic}. 
Clearly weak acyclicity is not sufficient for tractability, as discussed at the start of the current section.

\begin{example}
Consider the following normalized $\cpfc\noconstr$:
\[ \varphi \df \cqhead{\vec{x}} (\strucvar \logeq x_1 \cdot x_2 \cdot x_1 \cdot x_3 \cdot x_1) \land (x_1 \logeq x_4 \cdot x_5 \cdot x_5) \land (x_6 \logeq x_7 \cdot x_7 \cdot x_7).\]

Using the GYO algorithm, we can see that $\varphi$ is weakly acyclic. 
\end{example}

Let $\varphi \df \cqhead{\vec{x}} \bigwedge_{i=1}^n \eta_i$ be an $\cpfc\noconstr$, and let $\decomp_\varphi$ be an acyclic decomposition of $\varphi$. 
If $T \df (V,E)$ is a join tree of $\decomp_\varphi$, then for each $i \in [n]$, we use  $T^i \df (V^i, E^i)$ to denote the subtree of $T$ that is a join tree for the decomposition of $\eta_i$. 
We know that $T^i$ and $T^j$ are disjoint for all $i,j \in [n]$ where $i \neq j$, see~\cref{lemma:subtree}.

\begin{restatable}[]{lemma}{CyclicConditions}
\label{lemma:CyclicConditions}
Let $\varphi \df \cqhead{\vec{x}} \bigwedge_{i=1}^{n} \eta_i$ be a normalized $\cpfc\noconstr$. If any of the following conditions holds, then $\varphi$ is cyclic: 
\begin{enumerate}
\item $\varphi$ is weakly cyclic,
\item $\eta_i$ is cyclic for any $i \in [n]$,
\item $|\var(\eta_i) \intersect \var(\eta_j)| > 3$ for any $i, j \in [n]$ where $i \neq j$, or
\item $|\var(\eta_i) \intersect \var(\eta_j)| = 3$, and $|\eta_i| > 3$ or $|\eta_j| > 3$ for any $i, j \in [n]$ where $i \neq j$. 
\end{enumerate}\leavevmode
\end{restatable}

Condition 1 can be proven by simply replacing $T^i$ with a single node $\eta_i$ for all $i \in [n]$. 
Condition 2 follows directly from~\cref{lemma:subtree}. 
Conditions 3 and 4 can be proven by a contradiction: Consider the shortest path from any atom of the decomposition of $\eta_i$ to any atom of the decomposition of $\eta_j$. 
Since the end points of these paths cannot contain all the variables that $\eta_i$ and $\eta_j$ share, it follows that $T \df (V,E)$ is not a join tree.

While Conditions 3 and 4 might seem strict, we can pre-factor common subpatterns. 
For example, the conjunction $(x_1 \logeq \alpha_1 \cdot \alpha_2 \cdot \alpha_3) \land (x_2 \logeq \alpha_4 \cdot \alpha_2 \cdot \alpha_5)$, where $\alpha_i \in \Xi^+$ for $i \in [5]$, can be written as $(x_1 \logeq \alpha_1 \cdot z  \cdot \alpha_3) \land (x_2 \logeq \alpha_4 \cdot z \cdot \alpha_5) \land (z \logeq \alpha_2)$ where $z \in \Xi$ is a new variable. 
We illustrate this further in the following example.

\begin{example}
Consider the following $\cpfc$:
\[ \varphi \df \cqhead{} (x_1 \logeq y_1 \cdot y_2 \cdot y_3 \cdot y_4 \cdot y_5) \land (x_2 \logeq y_6 \cdot y_2 \cdot y_3 \cdot y_4 \cdot y_5). \]
Using~\cref{lemma:CyclicConditions}, we can see that $\varphi$ is cyclic. However, since the right-hand side of the two word equations share a common subpattern, we can rewrite $\varphi$ as
\[ \varphi' \df \cqhead{} (x_1 \logeq y_1 \cdot z) \land (x_2 \logeq y_6 \cdot z) \land (z \logeq y_2 \cdot y_3 \cdot y_4 \cdot y_5). \]
\end{example}

One could alter our definition of $\cpfc$ decomposition so that if two atoms share a bracketing, then the bracketing is replaced with the same variable (analogously to how decompositions are defined on patterns). 
The authors believe it is likely that such a definition of $\cpfc$ decomposition is equivalent to our definition of $\cpfc$ decomposition after ``factoring out'' common subpatterns between atoms.

Our next consideration is how the structure of a join tree for a decomposition of an acyclic query $\varphi \in \cpfcreg$ relates to the structure of a weak join tree for $\varphi$.

\begin{definition}[Skeleton Tree]
Let $\decomp_\varphi \in \conclog\noconstr$ be an acyclic decomposition of the query $\varphi \df \cqhead{\vec{x}} \bigwedge_{i=1}^n \eta_i$, and let $T \df (V,E)$ be a join tree for $\decomp_\varphi$. 
We say that a weak join tree $T_w \df (V_w, E_w)$ is the 
\emph{skeleton tree} of $T$ if there exists an edge in $E$ from a node in $V^i$ to a node in $V^j$ if and only if $\{ \eta_i, \eta_j \} \in E_w$.
\end{definition}

In the proof of~\cref{lemma:CyclicConditions} (Condition 1), we show that every join tree for a decomposition has a corresponding skeleton tree. 
We shall leverage the fact that every join tree of a decomposition of an acyclic $\cpfcreg$ has a skeleton tree in the algorithm given in the proof of~\cref{theorem:LVJoinTree}.

\begin{example}
\label{example:SkeletonTree}
We define $\varphi\in \cpfc\noconstr$ and a decomposition $\decomp_\varphi$ as follows:
\begin{align*}
\varphi &\df \cqhead{\vec{x}} (x_1 \logeq x_2 \cdot x_3 \cdot x_2) \land (x_2 \logeq x_4 \cdot x_4 \cdot x_5),\\
\decomp_\varphi &\df  \cqhead{\vec{x}} (x_1 \logeq x_2 \cdot z_1) \land (z_1 \logeq x_3 \cdot x_2) \land (x_2 \logeq z_2 \cdot x_5) \land (z_2 \logeq x_4 \cdot x_4).
\end{align*}	
The skeleton tree along with the join tree of $\decomp_\varphi$ are given in~\cref{fig:skeletonTree}.
\begin{figure}
\begin{tikzpicture}[shorten >=1pt,->]
\tikzstyle{vertex}=[rectangle,fill=red!35,minimum size=12pt,inner sep=4pt]
\tikzstyle{vertex2}=[rectangle,fill=blue!35,minimum size=12pt,inner sep=4pt]
\node[vertex] (1) at (0,0) {$x_1 \logeq x_2 \cdot z_1$};
\node[vertex] (2) at (4,0) {$z_1 \logeq x_3 \cdot x_2$};
\node[vertex2] (3) at (0, 1.2) {$x_2 \logeq z_2 \cdot x_5$};
\node[vertex2] (4) at (4, 1.2) {$z_2 \logeq x_4 \cdot x_4$};

\path [-](1) edge node[left] {} (3);
\path [-](1) edge node[left] {} (2);
\path [-](3) edge node[left] {} (4);
\end{tikzpicture} \hspace{1.5cm}
\begin{tikzpicture}[shorten >=1pt,->]
\tikzstyle{vertex}=[rectangle,fill=red!35,minimum size=12pt,inner sep=4pt]
\tikzstyle{vertex2}=[rectangle,fill=blue!35,minimum size=12pt,inner sep=4pt]
\node[vertex] (1) at (0,0) {$x_1 \logeq x_2 \cdot x_3 \cdot x_2$};
\node[vertex2] (2) at (0,1.2) {$x_2 \logeq x_4 \cdot x_4 \cdot x_5$};

\path [-](1) edge node[left] {} (2);
\end{tikzpicture}
\caption{\label{fig:skeletonTree} The join tree (left) and the skeleton tree of the join tree (right) for~\cref{example:SkeletonTree}.}
\end{figure}
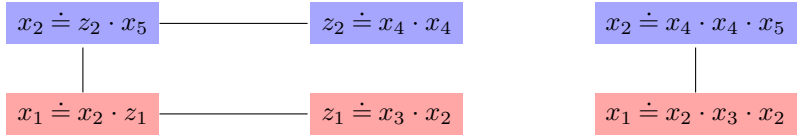
\end{example}

One might assume that some skeleton trees are more ``desirable'' than others in terms of using it for finding an acyclic decomposition of an $\cpfcreg$. 
However, as we observe next, any skeleton tree is sufficient.
\begin{restatable}[]{lemma}{skeletonTree}
\label{lemma:skeletonTree}
Let $\decomp_\varphi \in \conclog\noconstr$ be a decomposition of $\varphi \in \cpfc\noconstr$. If $\decomp_\varphi$ is acyclic, then any weak join tree can be used as the skeleton tree.
\end{restatable}

Given a weak join tree of an acyclic query $\varphi$, 
the proof of~\cref{lemma:skeletonTree} transforms the join tree of $\decomp_\varphi$ so that the resulting join tree has the given weak join tree as its skeleton tree.
Thus, we can use any weak join tree as a ``template'' for the eventual join tree of the decomposition (under the assumption that the query is acyclic).

While~\cref{lemma:CyclicConditions} and~\cref{lemma:skeletonTree} give some insights and necessary conditions for deciding whether $\varphi \in \cpfc\noconstr$ is acyclic, these conditions are not sufficient. 
We therefore give the following lemma which is needed in the proof of~\cref{theorem:LVJoinTree} to find an acyclic decomposition of $\varphi$.

\begin{restatable}[]{lemma}{atomDecomp}
\label{lemma:atomDecomp}
Given a normalized $\cpfc\noconstr$ of the form $\varphi \df \cqhead{\vec{x}} (z \logeq \alpha)$ and a set $C \subseteq \{ \{ x, y \} \mid x,y \in \var(z \logeq \alpha) \text{ and } x \neq y \}$,  we can decide whether there is an acyclic decomposition $\decomp \in \conclog\noconstr$ of $\varphi$ such that for every $\{ x,y \} \in C$, there is an atom of $\decomp$ that contains both $x$ and $y$ in time $\bigO(|\alpha|^7)$.
\end{restatable}

We prove~\cref{lemma:atomDecomp} using a variant of the algorithm given in the proof of~\cref{polytime}. 
The purposes of~\cref{lemma:atomDecomp} should become clearer after giving the following necessary and sufficient criteria for an $\cpfcreg$ to be acyclic:
Let $\varphi \df \bigwedge_{i=1}^m (x_i \logeq \alpha_i) \land \bigwedge_{j=1}^n (y_j \regconst \gamma_j)$ be a normalized $\cpfcreg$. 
Then, there exists an acyclic decomposition $\decomp \in \concreg$ of $\varphi$ if and only if the following conditions hold:
\begin{enumerate}
\item $\varphi$ is weakly acyclic,
\item for all $i \in [m]$ the pattern $\alpha_i$ is acyclic, and
\item for every $i \in [m]$, there is a decomposition $\decomp_i$ of $x_i \logeq \alpha_i$ such that for all $j \in [m] \setminus \{ i \}$ there is a decomposition $\decomp_j$ of $x_j \logeq \alpha_j$ where there exists an atom $\chi_i$ of $\decomp_i$ and an atom~$\chi_j$ of~$\decomp_j$ that satisfies $\var(\chi_i) \intersect \var(\chi_j) = \var(x_i \logeq \alpha_i) \intersect \var(x_j \logeq \alpha_j)$.
\end{enumerate}

We are now ready to give the main result of the paper. 
\begin{restatable}[]{theorem}{LVJoinTree}
\label{theorem:LVJoinTree}
Whether $\varphi \in \cpfcreg$ is acyclic can be decided in time $\bigO(|\varphi|^8)$.
\end{restatable}
To prove~\cref{theorem:LVJoinTree}, we first check whether $\varphi \in \cpfc$ has any of the conditions from~\cref{lemma:CyclicConditions}. 
If so, then we know that $\varphi$ is cyclic. 
Then, we construct a weak join tree for $\varphi$. 
If there is an edge $\{\eta_i, \eta_j\}$ of the weak join tree such that $\eta_i$ and $\eta_j$ share exactly two variables, then we use~\cref{lemma:atomDecomp} to decompose $\eta_i$ and $\eta_j$ such that there is an atom of the decomposition (of $\eta_i$ and $\eta_j$), which contains the variables that $\eta_i$ and $\eta_j$ share. 
In the full proof, we show that if such decompositions do not exist, then $\varphi$ is cyclic. 
For all other atoms of $\varphi$ we can use any decomposition. 
The resulting acyclic decomposition is the conjunction of the decompositions of each atom. 
The proof of~\cref{theorem:LVJoinTree} also shows that if $\varphi$ is acyclic, an acyclic decomposition can be constructed in polynomial time.

\begin{example}
We revisit the $\cpfcreg$ that was given in the introduction:
\[
\varphi \df \cqhead{x, y} (z \logeq z_2 \cdot x \cdot z_3 \cdot x \cdot z_4) \land (z \logeq z_5\cdot y \cdot z_6)  \land  (z \regconst \gamma_{\mathsf{sen}}) \land (x \regconst \gamma_{\mathsf{prod}}) \land (y \regconst \gamma_{\mathsf{pos}}).
\]
We can see this is acyclic by considering the following decomposition:
\begin{multline*}
\decomp \df \cqhead{x,y} (y_1 \logeq x \cdot z_3) \land (y_2 \logeq y_1 \cdot x) \land (y_3 \logeq z_2 \cdot y_2) \land (z \logeq y_3 \cdot z_4) \\ \land (y_4 \logeq z_5 \cdot y) \land (z \logeq y_4 \cdot z_6) \land  (z \regconst \gamma_{\mathsf{sen}}) \land (x \regconst \gamma_{\mathsf{prod}}) \land (y \regconst \gamma_{\mathsf{pos}}). 
\end{multline*}
\end{example}

Due to the small width of the tables that each word equation of the form $(x \logeq y \cdot z)$ produces, we  conclude the following:

\begin{restatable}[]{proposition}{enumAndEval}
\label{corollary:enumerationAndEvaluation}
If $\decomp \in \concreg$ is acyclic, then:
\begin{enumerate}
\item Given $w \in \Sigma^*$, the model checking problem can be solved in time $\bigO(|\decomp|^2 |w|^3)$.
\item Given $w \in \Sigma^*$, we can enumerate $\fun{\decomp}\strucbra{w}$ with $\bigO(|\decomp|^2|w|^3)$ delay.
\end{enumerate}\leavevmode
\end{restatable}
For $\cpfcreg$s, we first find an acyclic decomposition $\decomp_\varphi \in \concreg$ of $\varphi$ in $\bigO(|\varphi|^7)$. 
Then, the upper bound for model checking follows from~\cite{gottlob2001complexity}. 
Polynomial-delay enumeration follows from~\cite{bagan2007acyclic}, where it was proven that given an acyclic (relational) conjunctive query $\psi$ and a database $D$, we can enumerate $\psi(D)$ with $\bigO(|\psi||D|)$ delay. 
Our ``database'' is of size $\bigO(\formulaSize{\varphi} \cdot |w|^3)$ as each atom of the form $(z \logeq x \cdot y)$ defines a relation of size $\bigO(|w|^3)$.

Considering techniques from~\cite{bagan2007acyclic}, it may seem that the results of an acyclic $\cpfcreg$ without projections can be enumerated with constant-delay after polynomial time preprocessing. 
However this is not the case. 
New variables, that are not free, are introduced in the decomposition of $\varphi$ and therefore the resulting $\concreg$ may not be free-connex, which is required for the results of a $\cq$ to be enumerated with constant-delay~\cite{bagan2007acyclic}.

\subparagraph*{From FC[REG]-CQs to SERCQs} 
Combining~\cref{Prop:RGXtoPatCQ} and~\cref{corollary:enumerationAndEvaluation} gives us a class of $\sercq$s for which model checking can be solved in polynomial-time, and we can enumerate results with polynomial-delay. 
The hardness of deciding semantic acyclicity (whether a given $\sercq$ can be realized by an acyclic $\cpfcreg$) remains open. 
The authors believe that semantic acyclicity for $\sercq$s is undecidable, partly due to the fact that various minimization problems are undecidable for $\fc$~\cite{fre:doc, frey2019finite}. 
For now, all we have are sufficient critiera for a $\sercq$ to be realized by an acyclic $\cpfcreg$.

\begin{definition}\label{defn:PseudoAcyc}
We say that a query of the form $\query \df \pi_Y \bigl( \select^=_{x_1,y_1} \cdots \select^=_{x_k,y_k} ( \gamma_1 \join \cdots \join \gamma_n) \bigr)$ is \emph{pseudo-acyclic} if for every $i \in [n]$, we have that $\gamma_i \df \beta_{i_1} \cdot \bind{x_i}{\beta_{i_2}} \cdot \beta_{i_3}$ where $x_i \in \Xi$, and where $\beta_{i_1}$, $\beta_{i_2}$, and $\beta_{i_3}$ are regular expressions.
\end{definition}

We now show that~\cref{defn:PseudoAcyc} gives sufficient criteria for an $\sercq$ to be realized by an acyclic $\cpfcreg$.

\begin{restatable}[]{proposition}{QuasiAcyclicSpanners}
\label{prop:quasiAcyclic}
Given a pseudo-acyclic $\sercq$ $query$, we can construct in polynomial time an acyclic $\cpfcreg$ that realizes $\query$.
\end{restatable}

Freydenberger et al.~\cite{freydenberger2018joining} proved that fixing the number of atoms and the number of string equalities in a $\sercq$ allows for polynomial-delay enumeration of results. 
In contrast to this, \cref{prop:quasiAcyclic} allows an unbounded number of joins and string equality selection operators.
However, in order to have this tractability result, the expressive power of each regex formula is restricted to only allow one variable.
While~\cref{prop:quasiAcyclic} gives sufficient criteria for a $\sercq$ to be represented by an acyclic $\cpfcreg$, many other such classes of $\sercq$s likely exist. Research into finding large classes of $\sercq$s that map to acyclic $\cpfcreg$s seems like a promising direction for future work.

\section{A Note on k-ary Decompositions}\label{sec:kfold}
We now generalize the notion of pattern decomposition so that the length of the right-hand side of the resulting formula is less than or equal to some $k \geq 2$. While the binary decompositions might be considered the natural case, we show that generalizing to higher arities increases the expressive power of acyclic patterns. By $\kconclog{k}\noconstr$ we denote the set of $\cpfc\noconstr$ formulas that have a right-hand side of at most length $k$. We write $\brac_k$ for the set of $k$-ary bracketed patterns over $\Xi$. We define $\brac_k$ formally using the following recursive definition: For all $x \in \Xi$ we have that $x \in \brac_k$, and if $\alpha_1, \alpha_2, \dots, \alpha_i \in \brac_k$ where $i\leq k$, then $(\tilde\alpha_1 \cdot \tilde\alpha_2 \cdots \tilde\alpha_i) \in \brac_k$. We write $\tilde\alpha \in \brac_k(\alpha)$ for some $\alpha \in \Xi^+$ if the underlying, unbracketed pattern of $\tilde\alpha$ is $\alpha$. We can convert $\tilde\alpha \in \brac_k$ into an equivalent $\kconclog{k}\noconstr$ analogously to the binary case, see~\cref{defn:conclogConversion}. 

\begin{example}
\label{example:kfold}
Consider the following $4$-ary bracketing: 
\[\tilde\alpha \df ( ( ( x_1 \cdot x_2 \cdot x_3) \cdot (x_4 \cdot x_2 \cdot x_4) \cdot (x_1 \cdot x_2) \cdot (x_5 \cdot x_5) ) \cdot x_2) .\]

As with the $2$-ary case, we decompose $\tilde\alpha$ to get the following $\kconclog{4}\noconstr$:
\begin{multline*} 
\decomp_{\tilde\alpha} \df \cqhead{}  (z_1 \logeq x_1 \cdot x_2) \land (z_2 \logeq x_5 \cdot x_5) \land (z_3 \logeq x_4 \cdot x_2 \cdot x_4) \\ \land (z_4 \logeq x_1 \cdot x_2 \cdot x_3) \land (z_5 \logeq z_4 \cdot z_3 \cdot z_1 \cdot z_2) \land (\strucvar \logeq z_5 \cdot x_2). 
\end{multline*}
\end{example}

The definition of $k$-ary concatenation tree for a decomposition $\decomp_{\tilde\alpha} \in \kconclog{k}\noconstr$ of $\tilde\alpha \in \brac_k$ follows analogously to the concatenation trees for $2$-ary decompositions, see~\cref{defn:concatenationTree}. The concatenation tree of the decomposition $\decomp_{\tilde\alpha} \in \kconclog{k}\noconstr$ is a rooted, labeled, undirected tree $\mathcal{T} \df (\mathcal{V}, \mathcal{E}, <, \Gamma, \labelFunction, v_r)$, where $\mathcal{V}$ is the set of nodes, the relation $\mathcal{E}$ is the edge relation, and $<$ is used to denote the order of children of a node (from left to right). We have that $\Gamma \df \var(\decomp_{\tilde\alpha})$ is the alphabet of labels and $\tau \colon \mathcal{V} \rightarrow \Gamma$ is the labeling function. The semantics of a $k$-ary concatenation tree are defined by considering the natural generalization of~\cref{defn:concatenationTree}. 
We say that $\decomp_{\tilde\alpha}$ is \emph{$x$-localized} if all nodes which exist on a path between two $x$-parents (of $\mathcal{T}$) are also $x$-parents.

\begin{restatable}[]{proposition}{3aryLocalized}\label{prop:3aryLocalized}
There exists $\tilde\alpha \in \brac_3$ such that the decomposition $\decomp \in \kconclog{3}$ of $\tilde\alpha$ is acyclic, but there exists $x \in \var(\decomp)$ such that $\decomp$ is not $x$-localized.
\end{restatable}
\begin{proof}
Consider $\tilde\alpha \df ( (x_3 \cdot x_3) \cdot ((x_3 \cdot x_3) \cdot x_2) \cdot ( x_1 \cdot ((x_3 \cdot x_3) \cdot x_2)))$. The bracketing $\tilde\alpha$ is decomposed into $\decomp_{\tilde\alpha} \in \kconclog{3}\noconstr$, which is defined as
\[ \decomp_{\tilde\alpha} \df \cqhead{} (z_1 \logeq x_3 \cdot x_3) \land (z_2 \logeq z_1 \cdot x_2) \land (z_3 \logeq x_1 \cdot z_2) \land (\strucvar \logeq z_1 \cdot z_2 \cdot z_3) . \]

The formula $\decomp_{\tilde\alpha}$ can be verified to be acyclic. However, $\decomp_{\tilde\alpha}$ is not $z_1$-localized. 
\end{proof}

In this section, we have briefly examined $k$-ary decompositions, and have shown that there exists $\tilde\alpha \in \brac_3$ such that the decomposition $\decomp \in \kconclog{3}$ of $\tilde\alpha$ is acyclic, but $\decomp$ is not $x$-localized for some $x \in \var(\decomp)$. 
The authors note that the if-direction in the proof of~\cref{lemma:cycledistance} implies that $x$-locality for all variables is a sufficient criterion for a $k$-ary decomposition to be acyclic.
A systematic study into $k$-ary acyclic decompositions may yield more expressive spanners, and could be useful for pattern languages, which have been linked to $\fc$-formulas with bounded width~\cite{frey2019finite}. 
However, more general approaches such as bounded treewidth for binary decompositions appear to be a more promising direction for future work.
Furthermore, the membership problem for a pattern $\alpha$ parameterized by $|\alpha|$ is $\mathsf{W}[1]$-hard~\cite{fernau2016parameterised}. 
Since every pattern is trivially $|\alpha|$-ary acyclic, the authors believe it to be likely that the parameterized problem of model checking for $k$-ary acyclic decompositions is $\mathsf{W}[1]$-hard.

\section{Conclusions}\label{sec:conc}
Freydenberger and Peterfreund~\cite{frey2019finite} introduced  $\fcreg$ as a logic for querying and model checking words that behaves similar to relational $\fo$.
The present paper develops  this connection further by providing  a polynomial-time algorithm that either decomposes  an $\cpfcreg$ into an acyclic $\concreg$, or determines that this is not possible. 
These acyclic $\concreg$ formulas allow for polynomial-time model checking, and their results can be enumerated with polynomial-delay. 
Consequently, the present paper establishes a notion of tractable acyclicity for $\cpfc$s. 
Due to the close connections between $\fcreg$ and core spanners, this provides us with a large class of tractable $\sercq$s.

But this is only the first step in the study of tractable $\sercq$s and $\cpfcreg$s.
It seems likely that  more efficient algorithms for model checking and enumeration can be found by utilizing string algorithms rather than materializing the relations for each atom.
 
Another future direction for research is the consideration of other structural parameters, like treewidth.
A systematic study of the decomposition of $\cpfc$s into $\conclog$s of bounded treewidth would likely yield a large class of $\cpfc$s with polynomial-time model checking.
As a consequence, one could define a suitable notion of treewidth for core spanners. 
Determining the exact class of $\cpfc$s with polynomial-time model checking is likely a hard problem. 
This is because such a result would solve the open problem in formal languages of determining exactly what patterns have polynomial-time membership.

\bibliography{bibliography}
\newpage
\appendix
\section{Proof of~\autoref{CQlowerBounds}}\label{proof:CQlowerBounds}
\CQlowerBounds*
\begin{proof}
	The upper bound for evaluation follows immediately from the matching upper bound for the existential-positive fragment of $\fc$ with regular constraints (see~\cite{frey2019finite}).
	
	Lower bound for evaluation follows from the fact that, given $\alpha\in \Xi^*$ and $w\in \Sigma^*$,  deciding whether there is a morphism $\sigma\colon\Xi^*\to\Sigma^*$ with $\sigma(\alpha)=w$ is $\np$-complete (see Ehrenfeucht and Rozenberg~\cite{ehrenfreucht1979finding}). 
	Hence, even model-checking $\cpfc$s of the form $\cqhead{}(\strucvar\logeq \alpha)$ is $\np$-hard.
	
	The undecidability follows from the undecidability of the inclusion problem for pattern languages (see Bremer and Freydenberger~\cite{bre:inc}): Given $\alpha,\beta\in(\Xi\cup\Sigma)^*$, does every pattern substitution $\sigma$ have a pattern substitution $\tau$ with $\sigma(\alpha)=\tau(\beta)$?  
	Hence, containment is undecidable even if restricted to comparing $\cpfc$s of the form $\cqhead{}(\strucvar\logeq \alpha)$ and  $\cqhead{}(\strucvar\logeq \beta)$ with  $\alpha,\beta\in(\Xi\cup\Sigma)^*$.
\end{proof}

\section{Proof of~\autoref{Prop:RGXtoPatCQ}}

Before proving~\cref{Prop:RGXtoPatCQ}, we first define a \emph{parse trees} for $\gamma \in \synrgx$. Note that we assume $\gamma$ is well-bracketed. That is, each subexpression of $\gamma$ is of the form $a$, $\emptyset$, $\emptyword$, $(\gamma_1)^*$, $(\gamma_1 \cdot \gamma_2)$, $(\gamma_1 \lor \gamma_2)$, or $\bind{x}{\gamma_1}$ for $a \in \Sigma$ and $\gamma_1,\gamma_2, \in \synrgx$. If $\gamma \in \synrgx$ is not well-bracketed, then we can assume any valid bracketing for $\gamma$.

\begin{definition}
Let $\gamma \in \synrgx$. A \emph{parse tree} for $\gamma$ is a rooted, direct tree $T_\gamma$. Each node of $T_\gamma$ is a subexpression of $\gamma$. The root of $T_\gamma$ is $\gamma$. For each node $v$ of $T_\gamma$, the following rules must hold.
\begin{enumerate}
\item If $v$ is $(\gamma_1 \cdot \gamma_1)$ where $\SVars{\gamma_1} \neq \emptyset$ or $\SVars{\gamma_2} \neq \emptyset$, then $v$ has a left child $\gamma_1$, and a right child $\gamma_2$,
\item if $v$ is $x \{ \gamma' \}$, then $v$ has $\gamma'$ as a single child, and 
\item if $v$ is any other subexpression, then $v$ is a leaf node. 
\end{enumerate}
\end{definition}

The parse tree for $\gamma$ that we define is specific for our use, and is different to the standard definition of $\gamma$-parse trees which are used to define the semantics for regex-formulas, see~\cite{fag:spa}. The proof of the following proposition follows from~\cite{fre:doc, fre:splog, frey2019finite}, however we include this proof for completeness sake. 

\RGXtoPatCQ*
\begin{proof}
Let $\query \df \pi_Y \left( \select^=_{x_1, y_1} \cdots \select^=_{x_m, y_{m}} (\gamma_1 \join \cdots \join \gamma_k) \right)$ be a synchronized $\sercq$. We realize $\query$ using the following $\cpfcreg$:

\[ \varphi_\query \df \cqhead{\vec{x}} \bigwedge_{i=1}^m (x_i^C \logeq y_i^C) \land \bigwedge_{i=1}^k \varphi_{\gamma_i}, \]
where  $\vec{x}$ contains $x^P$ and $x^C$ for all $x \in Y$. Furthermore, for each $i \in [k]$, we define $\varphi_{\gamma_i}$  as follows: Take the parse tree $T_{\gamma_i}$ for $\gamma_i$ and associate every node $n$ of $T_{\gamma_i}$ with a variable $v_n$ as follows:
\begin{itemize}
	\item If $n$ is the root, let $v_n\df \strucvar$ and disregard the following cases.
	\item If $n$ is a variable binding  $x\{\cdot\}$, let $v_n\df x^C$.
	\item Otherwise -- that is, if $n$ is a concatenation or a regular expression -- let $v_n\df z_n$, where $z_n$ is a new variable that is unique to $n$.	
\end{itemize}
The construction shall ensure that, when matching $\gamma_i$ against a word $w$, each variable $v_n$ contains the part of $w$ that matches against the subexpression of the node $n$. 
To this end, for every node $n$, we also define an atom $A_n$ as follows:
\begin{itemize}
	\item If $n$ is a concatenation with left child $l$ and right child $r$, then $A_n$ is the word equation $(v_n\logeq v_l \cdot v_r)$.
	\item If $n$ is a variable binding, let $A_n$ be the word equation $(v_n \logeq v_c)$, where $c$ is the child of $n$.
	\item If $n$ is a regular expression $\gamma'$, then $A_n$ is the regular constraint $(v_n \regconst \gamma')$.
\end{itemize}
We add all these atoms $A_n$ to $\varphi_{\gamma_i}$. 
Up to this point, we have that every $\sigma\in\fun{\varphi_{\gamma_i}}\strucbra{w}$ encodes the contents of the spans of some  $\mu\in\fun{\gamma_i}(w)$.
The only part that is missing in the construction are the prefix variables.

\newcommand{\preffun}{\mathsf{p}}
Recall that for every node $n$ in the parse tree $T(\gamma_i)$, we defined a variable $v_n$ that represent the part of $w$ that matches against the subexpression of $n$. 
To obtain the corresponding prefix, we define a function $\preffun$ that maps each node $n$ to a pattern $\preffun(n)\in\Xi^*$ as follows.
Given a node $n$, we look for the lowest node above $n$ that is a concatenation and has $n$ as right child or descendant of its left child. 
If no such node exists -- that is, if no node above $n$ is a concatenation, or every concatenation above $n$ has $n$ as descendant on the left side -- define $\preffun(n)\df\emptyword$.
If such a node exists, we denote it by $m$ and its left child by $l$ and define $\preffun(n)\df \preffun(m)\cdot v_l$.
In other words, $\preffun(n)$ is the concatenation of all $v_l$ that belongs to nodes that refer a part of $w$ that is to the left of the part that belongs to $n$.

Hence, to get the values for prefix variables, we take each node $n$ that is a variable binding $x\{\cdot\}$ and add the word equation $(x^P\logeq \preffun(n))$ to $\varphi_{\gamma_i}$.

\subparagraph*{Complexity} First, we build the parse tree $T_{\gamma_i}$ which can be constructed in time polynomial in the size of $\gamma_i$. Then, we mark each node of $T_{\gamma_i}$ with a variable and add a word equation or regular constraint to $\varphi_{\gamma_i}$, which takes polynomial time. 
To ensure the spanner $\gamma_i$ represents is correctly realized, we add an extra word equation for the prefix variable -- this clearly takes polynomial time. 
There are linearly many regex formulas in $\query$, we can construct $\varphi_{\gamma_i}$ for all $i \in [k]$ in polynomial time. 
The final step of computing $\varphi_\query$ takes polynomial time -- we consider each string equality and add the corresponding word equation, and consider each variable in the projection and add the corresponding variables to the head of the query. 
Therefore, the overall complexity is polynomial in the size of $\query$.
\end{proof}

\section{Proof of~\autoref{lemma:datastructure}}
\Datastructure*
\newcommand{\LCP}{\mathsf{LCP}}
\begin{proof}
The two main concepts that are used for the data structures are the \emph{LCP data structure} (from ``least common prefix'', see \eg \cite{karkkainen2006linear}), and the \emph{suffix tree} (see \eg part~II of~\cite{gus:alg}), which can both be constructed from $w$ in time $\bigO(|w|)$. 

\proofsubparagraph{Evaluation} The LCP data structure (for $w$) takes two indices $1\leq i,j\leq w$ and returns in constant time $\LCP(i,j)$, the length of the longest common prefix of the two suffixes $w_{\spn{i,|w|+1}}$ and $w_{\spn{j,|w|+1}}$.
Recall that we mentioned in \cref{sec:prelim} (when clarifying the complexity assumptions) that we represent factors of $w$ as a pair of indices.
To be precise, we can express each $u\sqsubseteq w$ as a span $\spn{i,j}$ with $1 \leq i \leq j \leq |w|+1$. 
We use this to decide $\sigma\in\fun{x \logeq y \cdot z}\strucbra{w}$ in constant time as follows:
Let  $\spn{i,j}$, $\spn{i_1,j_1}$, and $\spn{i_2,j_2}$ be the representations of $\sigma(x)$, $\sigma(y)$, and $\sigma(z)$, respectively. 
In other words, $\sigma(x)=w_{\spn{i,j}}$, $\sigma(y)=w_{\spn{i_1,j_1}}$, and  $\sigma(z)=w_{\spn{i_2,j_2}}$.
For our convenience, let $\ell\df|\sigma(x)|$, $\ell_1\df|\sigma(y)|$, and $\ell_2\df|\sigma(z)|$.

We have $\sigma(x)=\sigma(y)\cdot\sigma(z)$ if and only if the following conditions are met:
\begin{itemize}
	\item $|\sigma(x)|=|\sigma(y)|+|\sigma(z)|$, that is, $\ell=\ell_1+\ell_2$,
	\item $\sigma(y)=\sigma(x)_{\spn{1,1+\ell_1}}$, and
	\item $\sigma(z)=\sigma(x)_{\spn{1+\ell_1,\ell+1}}$.
\end{itemize}
These are (respectively) equivalent to the following conditions:
\begin{itemize}
	\item $(j-i)=(j_1-i_1)+(j_2-i_2)$,
	\item $\LCP(i,i_1)\geq (j_1-i_1)$, and
	\item $\LCP(i+(j_1-i_1),i_2)\geq (j_2-i_2)$, 
\end{itemize}
due to $\ell=j-i$, $\ell_1= j_1-i_1$, and $\ell_2= j_2-i_2$.
The arithmetic operations can be performed in constant time due to our choice of computation model, and the LCP data structure can also be queried in constant time.
 
\proofsubparagraph{Enumeration of all factors}
Apart from some trivial special cases, the enumeration relies on enumerating all factors of $w$ with constant delay. 
This is a straightforward application of a \emph{suffix tree} (although the authors assume that this has been shown before, they were not able to locate a reference). 
We give a brief introduction to suffix trees, with just the level of detail that is required for our purposes. 
More information can be found (for example) in~\cite{gus:alg} (chapters~5 to~7).

The suffix tree $T(w)$ of $w$ is a rooted directed tree with $|w|$ leaves that are labeled with numbers from 1 to $n$. 
With the exception of the root, each internal node has at least two children, and each edge is labeled with a nonempty factor of $w$. 
No two edges from the same nodes are labeled with factors that start with the same letter.
Most importantly, for any leaf with label $i$, the word that is obtained by concatenating the edge labels along the path from the root to that leaf is exactly $w_{\spn{i,n+1}}$ -- that is, the suffix of $w$ that starts at position $i$. 

To ensure that a suffix tree for $w$ exists, we assume that the last letter of $w$ is a special character $\mathdollar$ that does not occur otherwise (that is, in a strict sense, we construct the suffix tree of $w\mathdollar$).
 \cref{fig:sufftreepapaya}  shows an example suffix tree, which we also use as a running example.
While storing the edge labels explicitly would take quadratic space, recall that we represent factors of $w$ as spans (this allows us to keep the size of $T(w)$ linear in $|w|$).
\begin{figure}[t]
	\centering
	\begin{tikzpicture}[->,>=stealth',shorten >=1pt,auto,on grid,,node distance=17mm]
  \tikzstyle{leaf}=[draw,shape=rectangle]
  \tikzstyle{inner}=[draw,shape=circle]	
   \tikzstyle{lab}=[pos=0.5,thin, align=center, fill=white, anchor=center]
  
		\node[leaf]  (ad) [] {6};
		\node[leaf]  (apayad) [right of=ad] {2};
		\node[leaf]  (ayad) [right of=apayad] {4};		
		\node[leaf]  (papayad) [right of=ayad,xshift=-10mm] {1};			
		\node[leaf]  (payad) [right of=papayad] {3};					
		\path (papayad) -- (payad) node[midway,anchor=center] (mid13) {};
				
		\node[inner]  (a) [above of=apayad] {};		
		\node[inner]  (pa) [above of=mid13] {};		
		
		\path (a) -- (pa) node[midway,anchor=center] (midapa) {};		
		
		\node[inner]  (root) [above of=midapa] {};		
		\node[leaf]  (yad) [right of=pa] {5};		
		\node[leaf]  (d) [left of=a] {7};		
		
        \draw[->]  (a) -- node[lab] {$\mathdollar$} (ad.north);
        \draw[->]  (a) -- node[lab] {$\mathtt{paya}\mathdollar$} (apayad.north);
        \draw[->]  (a) -- node[lab] {$\mathtt{ya}\mathdollar$} (ayad.north);

        \draw[->]  (pa) -- node[lab] {$\mathtt{paya}\mathdollar$} (papayad.north);
        \draw[->]  (pa) -- node[lab] {$\mathtt{ya}\mathdollar$} (payad.north);

        \draw[->]  (root) -- node[lab] {$\mathtt{a}$} (a);
        \draw[->]  (root) -- node[lab] {$\mathtt{pa}$} (pa);
        
        \draw[->]  (root) -- node[lab] {$\mathtt{ya}\mathdollar$} (yad.north west);
        \draw[->]  (root) -- node[lab] {$\mathdollar$} (d.north east);
	\end{tikzpicture}
	\caption{The suffix tree that we construct for the  word $w\df\mathtt{papaya}$. From left to right, the leaves correspond to the suffixes $\emptyword$, $\mathtt{a}$, $\mathtt{apaya}$, $\mathtt{aya}$, $\mathtt{papaya}$, $\mathtt{paya}$, and $\mathtt{ya}$. To enumerate the factors of $w$, we use the nodes 2, 4, 1, 3, 5.   
In the enumeration of factors, the leaf $2$ generates (in this order) $\emptyword$, $\mathtt{a}$, $\mathtt{ap}$, $\mathtt{apa}$, $\mathtt{apay}$, and $\mathtt{apaya}$; while $4$ only generates $\mathtt{ay}$ and $\mathtt{aya}$. Although this leaf corresponds to $\mathtt{aya}$, we skip the prefix $\mathtt{a}$, due to $\LCP(2,4)=1$. 
As $\LCP(4,1)=0$, we have that $1$ outputs $\mathtt{p}$, $\mathtt{pa}$,\ldots $\mathtt{papaya}$; and  $3$ only $\mathtt{pay}$ and $\mathtt{paya}$. Finally, from $5$, we get $\mathtt{y}$ and $\mathtt{ya}$.
 }\label{fig:sufftreepapaya}
\end{figure}
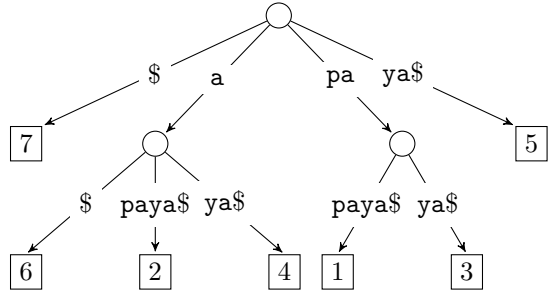

The suffix tree can be constructed in time $\bigO(|w|)$ (see \eg~\cite{gus:alg,karkkainen2006linear}). 
Note that we can ensure that the children of each node are ordered lexicographically.
To allow us enumerating all factors, we traverse the suffix tree depth-first during the preprocessing and create a list $L=i_1,\ldots,i_k$ of those leaves for which the incoming edge is labeled with more than just $\mathdollar$ (see  \cref{fig:sufftreepapaya}).
For the actual enumeration, we iterate over this list and use the leaves to output factors as follows:
\begin{itemize}
\item $i_1$ generates  $\emptyword$ and $w_{\spn{i_1,i_1+1}}$ to $w_{\spn{i_1,n+1}}$ (where we assume that $w_{\spn{n,n+1}}$ is the last letter of $w$, not $\mathdollar$), and
\item for $1\leq j < k$, every $i_{j+1}$ generates $w_{\spn{i_{j+1},i_{j+1}+\LCP(i_j,i_{j+1})+1}}$ to $w_{\spn{i_{j+1},n+1}}$.
\end{itemize}
That is, we use the suffix tree to enumerate all suffixes; and for each suffix, we enumerate all of its prefixes (apart from those that were already enumerated, which we spot skip by using $\LCP$, see ``Evaluation'' above).
As the list $L$ was derived directly from the tree, all leaves that have a common parent (which means that they longest common prefix is not $\emptyword$) are grouped together as a block. 
By using $\LCP$, we ensure that no factor is output twice. 
This is also why $L$ does not include leaves where the incoming edge is labeled~$\mathdollar$;
 factors that could be obtained from these words are handled by other leaves.
As the children of each inner node are ordered lexicographically, the construction also ensures that the factors of $w$ are output in lexicographic order. 
See \cref{fig:sufftreepapaya} for an example.

The list $L$ can be created in linear time during the preprocessing.
Each of the steps during the enumeration -- iterating over $L$, calling $\LCP$, and moving the indices -- takes only constant time. 
As the factors are returned as spans, we can conclude constant delay.

\proofsubparagraph{Enumeration of all solutions} 
To enumerate all $\sigma\in \fun{x \logeq y \cdot z}\strucbra{w}$, we need to consider various cases that depend on the three variables.
The ``standard'' case is that the variables $x,y,z$ are pairwise distinct, and none of them is~$\strucvar$. 
Then all we need to do is enumerate all $u\subword w$ (as described above). 
For each of these, we enumerate all ways of splitting $u$ into $v_1,v_2$ with $u=v_1\cdot v_2$, by enumerating the lengths of $v_1$ from 0 to $|u|$.
In each case, we define $\sigma(x)\df u$, $\sigma(y)\df u_1$, and $\sigma(z)\df u_2$ (and, of course, $\sigma(\strucvar)\df w$).

Regarding special cases, we first discuss those where at least one variable is $\strucvar$:
\begin{itemize}
	\item If $y=\strucvar$, the only solution is $\sigma(x)\df\sigma(y)=w$ and $\sigma(z)\df\emptyword$.
	This is well-defined -- unless  $z=\strucvar$ and $w\neq\emptyword$. In this case, we have  $\fun{x \logeq y \cdot z}\strucbra{w}=\emptyset$. This can be identified during the preprocessing.
	\item If $z=\strucvar$, we proceed as in the previous case.
	\item If $x=\strucvar$ and $y,z\neq \strucvar$, we distinguish two cases:
	\begin{itemize}
		\item If $y\neq z$, we set $\sigma(\strucvar)\df w$, and generate all possible $\sigma(y)$ and $\sigma(z)$ by enumerating all ways of splitting $w$ (as in the standard case).
		\item If $y=z$, we check if the first and second half of $w$ are identical (using $\LCP$ and arithmetic, we can perform this check in constant time during the preprocessing).  If this is the case, we can define the only $\sigma$ in  $\fun{\strucvar \logeq y \cdot y}\strucbra{w}$ accordingly. Otherwise, the set is empty.
	\end{itemize}
\end{itemize}
Now we can assume that none of the three variables is $\strucvar$, which leaves only cases where at least two are identical.
\begin{itemize}
	\item If $x=y=z$, the only $\sigma$ with  $\sigma\in\fun{x \logeq y \cdot z}\strucbra{w}$ has $\sigma(x)=\emptyword$.
	\item If $x=y\neq z$, we can assume $\sigma(z)=\emptyword$, and can choose any factor of $w$ for $\sigma(x)$. Hence, we enumerate all factors of $w$. The case for $x=z\neq y$ is analogous.
	\item If $x\neq y=z$, we enumerate all $u\subword w$ that are squares (\ie, that can be written as $u=vv$ for some $v\subword w$. Enumerating all these squares with constant delay is possible with additional preprocessing on the suffix tree, see Gusfield and Stoye\cite{gus:lin}.
\end{itemize}
Hence, we can set up the data structures for each of these cases during the preprocessing.
Given a word equation $x\logeq y\cdot z$, we can then pick the appropriate enumeration algorithm that allows us to enumerate $\fun{x \logeq y \cdot z}\strucbra{w}$ with constant delay.
\end{proof}
This construction also applies to equations of the form $x\logeq y_1\cdots y_k$ with $k>2$, assuming that $x$ and all $y_i$ are pairwise distinct (this proceeds as the ``standard case''). 

\section{Proof of~\autoref{cycPat}}

Before proving~\cref{cycPat}, we give the version of the GYO algorithm that we work with to decide the decomposition $\decomp_{\tilde\alpha} \df \cqhead{\vec{x}} \bigwedge_{i=1}^m \chi_i$ of $\tilde\alpha \in \brac(\alpha)$ is acyclic\footnote{We use variant of $\chi$ to denote atoms of some decomposition.} (see Chapter 6 of~\cite{abiteboul1995foundations} for more information on acyclic joins). We remind the reader that $\strucvar$ is considered a constant symbol (not a variable) since $\subs(\strucvar)$ is always our input document, $w \in \Sigma^*$.
\begin{enumerate}
\item Let $E \df \emptyset$ and $V \df \{ \chi_i \mid i \in [m]\}$.
\item Define all nodes of $V$ and all variables in $\var(\decomp_{\tilde\alpha})$ as \emph{unmarked}.
\item Repeat the following until nothing changes:
\begin{enumerate}
\item If there exists unmarked nodes $\chi_i$ and $\chi_j$ with $i \neq j$ such that $\var(\chi_i) \subseteq \var(\chi_j)$, then:
\begin{enumerate}
\item Mark $\chi_i$ and add the edge $\{ \chi_i, \chi_j \}$ to $E$.
\end{enumerate}
\item Mark all $x \in \var(\decomp_{\tilde\alpha})$ that occurs in exactly one unmarked node.
\end{enumerate}
\item If there exists exactly one unmarked node, then return $T \df (V,E)$. 
\item Otherwise, return ``$\tilde\alpha$ is cyclic''.
\end{enumerate}

\cycPat*
\begin{proof} We prove this Proposition in two parts.

\proofsubparagraph{Part 1. There exists a cyclic pattern:}
Let $\alpha \df x_1 x_2 x_1 x_3 x_1$. We prove that $\alpha$ is cyclic by enumerating every possible bracketing $\tilde\alpha \in \brac(\alpha)$, and then show that the decomposition of each bracketing is cyclic. To show a formula is cyclic, we can use the GYO algorithm. 

After the GYO algorithm has been executed on a $\conclog\noconstr$, we have a set of unmarked nodes, and each unmarked node contains unmarked variables. We represent each unmarked node as a set containing its unmarked variables. The set of unmarked nodes for $\decomp_{\tilde\alpha_i}$ after the GYO algorithm has been executed is denoted by $\mathcal{H}_i$. Therefore, the formula $\decomp_{\tilde\alpha_i}$ is acyclic if and only if $|\mathcal{H}_i| = 1$. We now consider all the bracketings, the corresponding decompositions, and the set $\mathcal{H}_i$ for each $\tilde\alpha_i \in \brac(\alpha)$:
\begin{itemize}
\item $\tilde\alpha_1 \df ((x_1 \cdot (x_2 \cdot (x_1 \cdot (x_3 \cdot x_1)))))$ which decomposes into 
\begin{align*}
\decomp_{\tilde\alpha_1} \df & \cqhead{} (z_1 \logeq x_3 \cdot x_1) \land (z_2 \logeq x_1 \cdot z_1) \land (z_3 \logeq x_2 \cdot z_2) \land (\strucvar \logeq x_1 \cdot z_3), \\
\mathcal{H}_1 \df & \{ \{ z_2, x_1\}, \{ z_3, z_2\}, \{ x_1, z_3\} \}.
\end{align*}
\item $\tilde\alpha_2 \df (x_1 \cdot ( x_2 \cdot ((x_1 \cdot x_3) \cdot x_1)))$ which decomposes into 
\begin{align*} 
\decomp_{\tilde\alpha_2} \df & \cqhead{} (z_1 \logeq x_1 \cdot x_3) \land (z_2 \logeq z_1 \cdot x_1) \land (z_3 \logeq x_2 \cdot z_2) \land (\strucvar \logeq x_1 \cdot z_3), \\
\mathcal{H}_2 \df & \{\{ z_2, x_1\}, \{ z_3, z_2\}, \{ x_1, z_3\} \}.
\end{align*}
\item $\tilde\alpha_3 \df ((x_1 \cdot x_2) \cdot (x_1 \cdot (x_3 \cdot x_1)))$ which decomposes into 
\begin{align*}
\decomp_{\tilde\alpha_3} \df & \cqhead{}(z_1 \logeq x_1 \cdot x_2) \land (z_2 \logeq x_3 \cdot x_1) \land (z_3 \logeq x_1 \cdot z_2) \land (\strucvar \logeq z_1 \cdot z_3), \\
\mathcal{H}_3 \df & \{ \{ z_1, x_1\}, \{ z_2, x_1\}, \{ z_3, z_1, z_2\}, \{ x_1, z_3\} \}.
\end{align*}
\item $\tilde\alpha_4 \df (x_1 \cdot ((x_2 \cdot x_1) \cdot (x_3 \cdot x_1)))$ which decomposes into 
\begin{align*}
\decomp_{\tilde\alpha_4} \df  &\cqhead{} (z_1 \logeq x_3 \cdot x_1) \land (z_2 \logeq x_2 \cdot x_1) \land (z_3 \logeq z_1 \cdot z_2) \land (\strucvar \logeq x_1 \cdot z_3), \\
\mathcal{H}_4 \df & \{ \{ z_1, x_1\}, \{ z_2, z_1\}, \{ z_2, x_1\} \}.
\end{align*}
\item $\tilde\alpha_5 \df (x_1 \cdot ((x_2 \cdot (x_1 \cdot x_3)) \cdot x_1 ))$ which decomposes into 
\begin{align*}
\decomp_{\tilde\alpha_5} \df  &\cqhead{} (z_1 \logeq x_1 \cdot x_3) \land (z_2 \logeq x_2 \cdot z_1) \land (z_3 \logeq z_2 \cdot x_1) \land (\strucvar \logeq x_1 \cdot z_3), \\
\mathcal{H}_5 \df & \{ \{ z_1, x_1\}, \{ z_2, z_1\}, \{ z_2, x_1\} \}.
\end{align*}
\item $\tilde\alpha_6 \df (x_1 \cdot (((x_2 \cdot x_1) \cdot x_3) \cdot x_1))$ which decomposes into 
\begin{align*} 
\decomp_{\tilde\alpha_6} \df  &\cqhead{} (z_1 \logeq x_2 \cdot x_1) \land (z_2 \logeq z_1 \cdot x_3) \land (z_3 \logeq z_2 \cdot x_1) \land (\strucvar \logeq x_1 \cdot z_3), \\
\mathcal{H}_6 \df & \{ \{ z_1, x_1\}, \{ z_3, z_2, x_1\}, \{ z_1, z_3\} \}.
\end{align*}
\item $\tilde\alpha_7 \df   ((x_1 \cdot x_2) \cdot (( x_1 \cdot x_3) \cdot x_1))$ which decomposes into 
\begin{align*} 
\decomp_{\tilde\alpha_7} \df & \cqhead{}(z_1 \logeq x_1 \cdot x_2) \land (z_2 \logeq x_1 \cdot x_3) \land (z_3 \logeq z_2 \cdot x_1) \land (\strucvar \logeq z_1 \cdot z_3), \\
\mathcal{H}_7 \df & \{ \{ z_2, x_1\}, \{ z_3, x_1\}, \{ z_3, z_2\}\}.
\end{align*}
\item $\tilde\alpha_8 \df (x_1 \cdot (x_2 \cdot x_1)) \cdot (x_3 \cdot x_1))$ which decomposes into 
\begin{align*} 
\decomp_{\tilde\alpha_8} \df & \cqhead{}(z_1 \logeq x_2 \cdot x_1) \land (z_2 \logeq x_3 \cdot x_1) \land (z_3 \logeq x_1 \cdot z_1) \land (\strucvar \logeq z_3 \cdot z_2), \\
\mathcal{H}_8 \df & \{ \{ z_1, x_1\}, \{ z_2, z_1\}, \{ z_2, x_1\} \}.
\end{align*}
\item $\tilde\alpha_9 \df (x_1 \cdot (x_2 \cdot (x_3 \cdot x_1))) \cdot x_1)$ which decomposes into 
\begin{align*} 
\decomp_{\tilde\alpha_9} \df & \cqhead{}(z_1 \logeq x_3 \cdot x_1) \land (z_2 \logeq x_2 \cdot z_1) \land (z_3 \logeq z_2 \cdot x_1) \land (\strucvar \logeq x_1 \cdot z_3), \\
\mathcal{H}_9 \df & \{ \{ z_1, x_1\}, \{ z_2, z_1\}, \{ x_1, z_2\} \}.
\end{align*}
\item $\tilde\alpha_{10} \df ((x_1 \cdot ((x_2 \cdot x_1) \cdot x_3)) \cdot x_1)$ which decomposes into 
\begin{align*} 
\decomp_{\tilde\alpha_{10}} \df & \cqhead{} (z_1 \logeq x_2 \cdot x_1) \land (z_2 \logeq z_1 \cdot x_3) \land (z_3 \logeq x_1 \cdot z_2) \land (\strucvar \logeq z_3\cdot x_1), \\
\mathcal{H}_{10} \df & \{ \{ z_1, x_1\}, \{ z_2, x_1\}, \{ z_3, z_1, z_2\}, \{ z_3, x_1\} \}.
\end{align*}
\item $\tilde\alpha_{11} \df (( (x_1 \cdot x_2) \cdot (x_1 \cdot x_3) ) \cdot x_1)$ which decomposes into 
\begin{align*} 
\decomp_{\tilde\alpha_{11}} \df & \cqhead{} (z_1 \logeq x_1 \cdot x_2) \land (z_2 \logeq x_1 \cdot x_3) \land (z_3 \logeq z_1 \cdot z_2) \land (\strucvar \logeq z_3 \cdot x_1), \\
\mathcal{H}_{11} \df & \{ \{ z_1, x_1\}, \{ z_2, x_1\}, \{ z_3, z_1, z_2\}, \{ z_3, x_1\} \}.
\end{align*}
\item $\tilde\alpha_{12} \df (((x_1 \cdot x_2) \cdot x_1) \cdot (x_3 \cdot x_1))$ which decomposes into
\begin{align*} 
\decomp_{\tilde\alpha_{12}} \df & \cqhead{} (z_1 \logeq x_1 \cdot x_2) \land (z_2 \logeq x_3 \cdot x_1) \land (z_3 \logeq z_1 \cdot x_1) \land (\strucvar \logeq z_3 \cdot z_2), \\
\mathcal{H}_{12} \df & \{ \{ z_2, x_1\}, \{ z_3, x_1\}, \{ z_3, z_2\} \}.
\end{align*}
\item $\tilde\alpha_{13} \df (((x_1 \cdot (x_2 \cdot x_1)) \cdot x_3) \cdot x_1)$ which decomposes into 
\begin{align*} 
\decomp_{\tilde\alpha_{13}} \df & \cqhead{} (z_1 \logeq x_2 \cdot x_1) \land (z_2 \logeq x_1 \cdot z_1) \land (z_3 \logeq z_2 \cdot x_3) \land (\strucvar \logeq z_3 \cdot x_1), \\
\mathcal{H}_{13} \df & \{ \{ z_2, x_1\}, \{ z_3, z_2\}, \{ z_3, x_1\} \}.
\end{align*}
\item $\tilde\alpha_{14} \df ((((x_1 \cdot x_2) \cdot x_1 ) \cdot x_3 ) \cdot x_1)$ which decomposes into
\begin{align*} 
\decomp_{\tilde\alpha_{14}} \df & \cqhead{} (z_1 \logeq x_1 \cdot x_2) \land (z_2 \logeq z_1 \cdot x_1) \land (z_3 \logeq z_2 \cdot x_3) \land (\strucvar \logeq z_3 \cdot x_1), \\
\mathcal{H}_{14} \df & \{ \{ z_2, x_1\}, \{ z_3, z_2\}, \{ z_3, x_1\} \}.
\end{align*}
\end{itemize}

For every $\tilde\alpha_i \in \brac(\alpha)$, we have that $|\mathcal{H}_i| > 1$. We can conclude that $\alpha$ is cyclic.

\proofsubparagraph{Part 2: There exists an acyclic pattern which has a cyclic bracketing.}
Let $\alpha \df x_1 x_2 x_3 x_1$, let $\tilde\alpha_1 \df ( (x_1 \cdot(x_2 \cdot x_3))\cdot x_1)$, and let $\tilde\alpha_2 \df ( (x_1 \cdot x_2) \cdot (x_3 \cdot x_1) )$. The decomposition of $\tilde\alpha_1$ is
$\decomp_{\tilde\alpha_1} \df \cqhead{} (z_1 \logeq x_2 \cdot x_3) \land  (z_2 \logeq x_1 \cdot z_1)\land (\strucvar \logeq z_2 \cdot x_1)$. Executing the GYO algorithm on $\decomp_{\tilde\alpha_1}$ shows it to be acyclic. 

The decomposition of $\tilde\alpha_2$ is $\decomp_{\tilde\alpha_2} \df \cqhead{} (\strucvar \logeq z_1 \cdot z_2) \land (z_1 \logeq x_1 \cdot x_2) \land (z_2 \logeq x_3 \cdot x_1)$. Performing the GYO algorithm on $\decomp_{\tilde\alpha_1}$ will show it to be cyclic. Therefore, we have proven that not all bracketings of an acyclic pattern is an acyclic bracketing.
\end{proof}

\section{Proof of~\autoref{lemma:cycledistance}}

We first prove a useful lemma that makes the actual proof of~\cref{lemma:cycledistance} more readable.

\begin{lemma}
\label{lemma:TreePath}
If $T \df (V, E)$ is an undirected tree where $V \df [n]$, then every node that lies on the path from $i$ to $j$, for $i,j \in [n]$ where $i < j$, must exist on a path from $k$ to $k+1$ for some $k \in \{i, i+1, \dots, j-1\}$.
\end{lemma}
\begin{proof}
Let $T \df (V,E)$ be an undirected tree where $V \df [n]$. For any $k,k' \in [n]$, let $p_{k \rightarrow k'}$ be the path from $k$ to $k'$ in $T$. The path $p_{i \rightarrow j}$ can be constructed by considering the sequence of edges $p_{i \rightarrow i+1} \cdot p_{i+1 \rightarrow i+2} \cdots p_{j-1 \rightarrow j}$, then removing all edges which appear more than once from this sequence. Since this defines a path from $i$ and $j$, and there can only be one path between any two nodes in a tree, the stated lemma holds.
\end{proof}

\cref{lemma:TreePath} can clearly be generalized to trees with any vertex set, $V$, by considering some bijection from the vertices of the tree to $[n]$ where $|V| = n$.

If $\eta \df (x \logeq y \cdot z)$ is an atom of the acyclic decomposition $\decomp_{\tilde\alpha} \in \conclog\noconstr$, then the right-hand side of $\eta$ can be reversed, \ie $\eta \df (x \logeq z \cdot y)$, and $\decomp_{\tilde\alpha}$ remains acyclic. Therefore, in the following proof, when the right-hand side of an atom is ambiguous, we can assume one without loss of generality.

\subsection*{Actual Proof of~\autoref{lemma:cycledistance}.}
\cycledistance*
\begin{proof}
Let $\decomp_{\tilde\alpha} \in \conclog\noconstr$ be a decomposition of $\tilde\alpha \in \brac$ and let $\mathcal{T} \df (\mathcal{V}, \mathcal{E}, <, \Gamma, \labelFunction, v_r)$ be the concatenation tree for $\decomp_{\tilde\alpha}$.

\proofsubparagraph{If-direction.} If $\decomp_{\tilde\alpha}$ is $x$-localized for all $x \in \var(\decomp_{\tilde\alpha})$, then we can construct a join tree for $\decomp_{\tilde\alpha}$ by augmenting the concatenation tree: First replace all non-leaf nodes $v \in \mathcal{V}$ with $\atom(v)$. Then remove all leaf nodes. By the definition of the concatenation tree, every atom of $\decomp_{\tilde\alpha}$ is a node in the supposed join tree. Also due to the definition of a concatenation tree, if $v$ is an $x$-parent, then $x$ occurs in $\atom(v)$. Because $\decomp_{\tilde\alpha}$ is $x$-localized for all $x \in \var(\decomp_{\tilde\alpha})$, it follows that if two nodes in the supposed join tree contain the variable $x$, then all nodes which exist on the path between these two nodes also contains an $x$. Hence, the resulting tree is a valid join tree for $\decomp_{\tilde\alpha}$.

\proofsubparagraph{Only if-direction.}
Let $v_0, v_n \in \mathcal{V}$ be two $x$-parents such that the distance between $v_0$ and $v_n$ in the concatenation tree $\mathcal{T}$ is $n > 1$. Let $v_1, v_2, \dots v_{n-1} \in \mathcal{V}$ be the nodes on the path between $v_0$ and $v_n$ in $\mathcal{T}$ where $v_i$ is not an $x$-parent for all $i \in [n-1]$, hence $\decomp_{\tilde\alpha}$ is not $x$-localized. For readability, we assume that $\labelFunction(v_i) = z_i$ for all $i \in \{ 0, 1, \dots, n \}$. Because the concatenation tree is pruned, $\atom(v_i) = \atom(v_j)$ if and only if $i = j$, for $i,j \in \{ 0, 1,\dots, n\}$. Furthermore, if $\labelFunction(v) = z_i$ where $v$ is a non-leaf node, then $v = v_i$ because two different non-leaf nodes cannot share a label.  \cref{fig:OnlyIf} illustrates a subtree of $\mathcal{T}$. The variable that labels each node is given next to the node in parentheses.

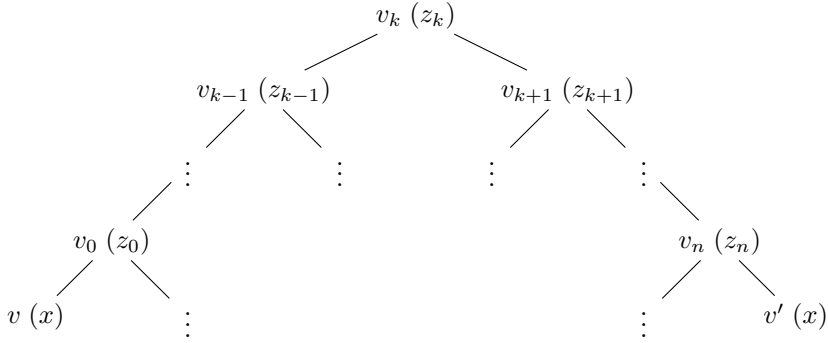
\begin{figure}
\begin{tikzpicture}[shorten >=1pt,->]
\tikzstyle{vertex}=[rectangle,fill=white!25,minimum size=12pt,inner sep=2pt]
\node[vertex] (1) at (0,0) {$v_k \; (z_k)$};
\node[vertex] (2) at (2,-1)   {$v_{k+1} \; (z_{k+1})$};
\node[vertex] (3) at (-2,-1)  {$v_{k-1} \; (z_{k-1})$};
\node[vertex] (4) at (3, -2) {$\vdots$};
\node[vertex] (5) at (-3, -2) {$\vdots$}; 
\node[vertex] (6) at (-4, -3) {$v_0 \; (z_0)$};
\node[vertex] (7) at (4, -3) {$v_n \; (z_n)$};
\node[vertex] (8) at (1, -2) {$\vdots$};
\node[vertex] (9) at (-1, -2) {$\vdots$};
\node[vertex] (10) at (-5,-4) {$v \; (x)$};
\node[vertex] (11) at (5,-4) {$v' \; (x)$};
\node[vertex] (12) at (-3,-4) {$\vdots$};
\node[vertex] (13) at (3,-4) {$\vdots$};

\path [-](1) edge node[left] {} (2);
\path [-](1) edge node[left] {} (3);
\path [-] (2) edge node[left] {} (4);
\path [-] (3) edge node[left] {} (5);
\path [-] (2) edge node[left] {} (8);
\path [-] (3) edge node[left] {} (9);
\path [-] (5) edge node [left] {} (6);
\path [-] (4) edge node [left] {} (7);
\path [-] (6) edge node [left] {} (10);
\path [-] (7) edge node [left] {} (11);
\path [-] (6) edge node [left] {} (12);
\path [-] (7) edge node [left] {} (13);
\end{tikzpicture}
\caption{\label{fig:OnlyIf} The concatenation tree, $\mathcal{T}$, we use for the only if-direction in the proof of~\cref{lemma:cycledistance}.}
\end{figure}

For sake of a contradiction, assume there exists a join tree $T \df (V, E)$ for $\decomp_{\tilde\alpha}$. Nodes in the join tree are the atoms of $\decomp_{\tilde\alpha}$ and therefore any element of $V$ can be uniquely determined by $\atom(v)$ where $v \in \mathcal{V}$ is a non-leaf node in the concatenation tree. We remind the reader that $\atom(v) = (z \logeq x \cdot x')$ if $v$ is labeled $z$ and the left and right children of $v$ are labeled $x$ and $x'$ respectively.  To improve readability, we use (variants of) $v$ for nodes of the concatenation tree, and we use $\atom(v)$ for nodes of the join tree where $v$ is some non-leaf node of the concatenation tree. 

We relax the factor notation to variables in $\var(\decomp_{\tilde\alpha})$. We write $z \sqsubset z'$, where $z,z' \in \var(\decomp_{\tilde\alpha})$, if there exists $v, v' \in \mathcal{V}$ where $v'$ which is an ancestor of $v$ in the concatenation tree, and $\labelFunction(v') = z'$ and $\labelFunction(v) = z$. We do this because the pattern that $z$ represents is a factor of the pattern $z'$ represents. 

Let $p_{i \rightarrow j}$ be the path in the join tree, $T$, from $\atom(v_i)$ to $\atom(v_j)$ for any $i,j \in \{ 0, 1, \dots, n\}$. The atom $\atom(v_1)$ cannot exist on the path $p_{0 \rightarrow n}$ because $\atom(v_0)$ and $\atom(v_n)$ contain the variable $x$, but $\atom(v_1)$ does not contain the variable $x$. We therefore consider some non-leaf node $v_1' \in \mathcal{V}$ of the concatenation tree such that $\atom(v_1')$ is the atom on the path $p_{0 \rightarrow n}$ which is closest (with regards to distance) to $\atom(v_1)$. See~\cref{fig:proofIdea} for a diagram to illustrate $\atom(v_1')$. We know that $\atom(v_1')$ has a variable $x$ since it lies on the path $p_{0 \rightarrow n}$. 

We now prove that $\atom(v_1')$ contains some variable $z_i$ where $i \in [n]$. Since $\atom(v_1')$ is the node closest to $\atom(v_1)$ on the path $p_{0 \rightarrow n}$, we have that $\atom(v_1')$ must also exist on the path $p_{1 \rightarrow n}$ (see~\cref{fig:proofIdea}). Therefore, because of~\cref{lemma:TreePath}, $\atom(v_1')$ must exist on some path $p_{j \rightarrow j+1}$ for some $j \in [n-1]$. Since $\atom(v_j)$ and $\atom(v_{j+1})$ share the variable $z_j$ or $z_{j+1}$ (depending on whether $v_j$ or $v_{j+1}$ is the parent) for all $j \in [n-1]$, it follows that $\atom(v_1')$ must contain the variable $z_i$ for some $i \in [n]$.

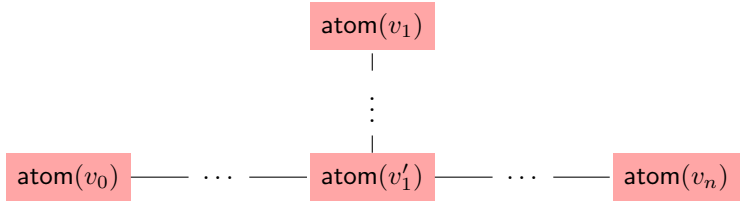
\begin{figure}
\begin{tikzpicture}[shorten >=1pt,->]
\tikzstyle{vertex}=[rectangle,fill=red!35,minimum size=12pt,inner sep=4pt]
\tikzstyle{vertex2}=[rectangle,fill=white!35,minimum size=12pt,inner sep=4pt]
\node[vertex] (1) at (0,0) {$\atom(v_0)$};
\node[vertex2] (2) at (2,0) {\dots};
\node[vertex] (3) at (4,0) {$\atom(v_1')$};
\node[vertex2] (4) at (4,1) {$\vdots$};
\node[vertex] (5) at (4,2) {$\atom(v_1)$};
\node[vertex2] (6) at (6,0) {\dots};
\node[vertex] (7) at (8,0) {$\atom(v_n)$};

\path [-](1) edge node[left] {} (2);
\path [-](2) edge node[left] {} (3);
\path [-](3) edge node[left] {} (4);
\path [-](4) edge node[left] {} (5);
\path [-](3) edge node[left] {} (6);
\path [-](6) edge node[left] {} (7);
\end{tikzpicture}\hspace{1cm}
\caption{\label{fig:proofIdea} A figure to illustrate paths $p_{0 \rightarrow 1}$ and $p_{0 \rightarrow n}$.}
\end{figure}

\proofsubparagraph{Case 1: $v_n$ is an ancestor of $v_0$ in $\mathcal{T}$.}

Since $v_n$ is an ancestor of $v_0$, we know that $v_i$ is an ancestor of $v_0$ (and hence $x \sqsubset z_0 \sqsubset z_i$) for all $i \in [n]$. Furthermore, it follows that $v_1$ is a $z_0$-parent and therefore $\atom(v_1) = (z_1 \logeq z_0 \cdot z')$ for some $z' \in \Xi$. Since $\atom(v_1')$ lies on the path $p_{0 \rightarrow 1}$, it follows that $\atom(v_1')$ contains the variable $z_0$. One of the variables of $\atom(v_1')$ must be the label of $v_1'$, we therefore consider all the possible labels for $v_1'$ and show a contradiction for each.

\begin{itemize}
\item $\labelFunction(v_1') = x$. This implies that, without loss of generality, $\atom(v_1') = (x \logeq z_0 \cdot z_i)$ and therefore $z_0 \sqsubset x$. We know that $x \sqsubset z_0$ since $v_0$ is an $x$-parent. Therefore, $z_0 \sqsubset z_0$ which we know cannot hold and hence $\labelFunction(v_1') = x$ cannot hold.
\item $\labelFunction(v_1') = z_i$ where $i \in [n]$. We split this case into two parts:
\begin{itemize}
\item $\labelFunction(v_1') = z_i$ where $i \in [n-1]$. This implies that $\atom(v_1') = \atom(v_i)$. The word equation $\atom(v_i)$ does not contain an $x$. Since we know that $\atom(v_1')$ contains the variable $x$, we can conclude that $\labelFunction(v_1') = z_i$ where $i \in [n-1]$ cannot hold.
\item $\labelFunction(v_1') = z_n$. This implies that $\atom(v_1') = \atom(v_n)$ and therefore, without loss of generality, $\atom(v_n) = (z_n \logeq x \cdot z_0)$. However, since we are in the case that $v_n$ is an ancestor of $v_0$, it follows that $v_n$ is a parent of $v_0$ (since a node labeled $x$ or $z_0$ cannot be an ancestor of $v_0$). Therefore $\decomp_{\tilde\alpha}$ is $x$-localized, and hence $\labelFunction(v_1') = z_n$ cannot hold.
\end{itemize}
\item $\labelFunction(v_1') = z_0$. This implies that $\atom(v_1') = \atom(v_0)$. Therefore, without loss of generality, $\atom(v_0) = (z_0 \logeq x \cdot z_i)$. We also know that $z_0 \sqsubset z_i$ since $v_i$ is an ancestor of $v_0$. Therefore, $z_0 \sqsubset z_i \sqsubset z_0$, which we know cannot hold. Thus, $\labelFunction(v_1') = z_0$ cannot hold.
\end{itemize}

We have proven that, for the case where $v_n$ is an ancestor of $v_0$ in the concatenation tree, there does not exist a valid label for $v_1'$. Hence we have reached a contradiction and therefore our assumption $\decomp_{\tilde\alpha}$ is acyclic cannot hold. 
  
\proofsubparagraph{Case 2: $v_0$ is an ancestor of $v_n$ in $\mathcal{T}$.}

The case where $v_0$ is an ancestor of $v_n$ is trivially identical to Case 1 by considering the closest node to $\atom(v_{n-1})$ on the path $p_{n \rightarrow 0}$. We have therefore omitted the proof. 
  
\proofsubparagraph{Case 3: $v_n$ is not an ancestor of $v_0$ in $\mathcal{T}$, and $v_0$ is not an ancestor of $v_n$ in $\mathcal{T}$.} 

Let $k \in [n-1]$ such that $v_k \in \mathcal{V}$ is the lowest common ancestor of $v_0$ and $v_n$ in $\mathcal{T}$. We remind the reader that $\atom(v_1')$ has the variables $x$, and $z_i$ for some $i \in [n]$ because $\atom(v_1')$ lies on the paths $p_{0 \rightarrow n}$ and $p_{1 \rightarrow n}$. We also have that $\atom(v_1) = (z_1 \logeq z_0 \cdot z')$ for some $z' \in \Xi$, because for this case, $v_1$ must be a parent of $v_0$, otherwise $v_0$ would be an ancestor of $v_n$. Therefore, since $\atom(v_0)$ and $\atom(v_1)$ share the variable $z_0$, we know that $\atom(v_1')$ also contains a $z_0$ -- because $\atom(v_1')$ lies on the path $p_{0 \rightarrow 1}$. We now consider each label for $v_1'$ and show a contradiction for each case.

\proofsubparagraph{Case 3.1: $\labelFunction(v_1') = x$.} 
Without loss of generality, $\atom(v_1') = (x \logeq z_0 \cdot z_i)$ which implies that $x \sqsubset z_0$ and $z_0 \sqsubset x$. This is a contradiction and hence $\labelFunction(v_1') = x$ cannot hold.

\proofsubparagraph{Case 3.2: $\labelFunction(v_1') = z_i$ where $i \in [n-1]$.}
This implies $\atom(v_1') = \atom(v_i)$, but $\atom(v_i)$ cannot have the variable $x$. This is a contradiction and hence $\labelFunction(v_1') =  z_i$ where $i \in [n-1]$ cannot hold.

\proofsubparagraph{Case 3.3: $\labelFunction(v_1') = z_n$.} 
This implies that $\atom(v_1') = \atom(v_n)$ and therefore without loss of generality, we know that $\atom(v_n) = (z_n \logeq z_0 \cdot x)$, because $\atom(v_1')$ must contain the variable $z_0$ and $x$. For this case, we first prove that $k \geq 2$ where $v_k$ is the lowest common ancestor of $v_0$ and $v_n$. For sake of contradiction, assume $k=1$. It follows that the distance from $v_k$ to $v_0$ is one and the distance from $v_k$ to $v_n$ is greater than or equal to one. Hence, the distance from $v_k$ to the children of $v_n$ is greater than or equal to two. Since $v_n$ is a $z_0$ parent, and the children of $v_n$ are further from the root than $v_0$, we know that $v_0$ must be redundant. If this is the case, $v_0$ would have no children due to the pruning procedure used when defining a concatenation tree. Therefore, $v_0$ would not be an $x$-parent which we know cannot hold (we have chosen $v_0$ \emph{because} it is an $x$-parent). Therefore, $k=1$ cannot hold and we can conclude $k \geq 2$. 

We now consider $\atom(v_k)$. We know that $\atom(v_k) = (z_k \logeq z_{k-1} \cdot z_{k+1})$ and since we have proven that $k \geq 2$, it follows that $z_{k-1} \neq z_0$. Since both $\atom(v_1)$ and $\atom(v_n)$ contain the variable $z_0$, we know that $\atom(v_k)$ cannot exist on the path $p_{1 \rightarrow n}$. Hence, we consider some non-leaf node $v_k' \in \mathcal{V}$ such that $\atom(v_k')$ lies on the path $p_{1 \rightarrow n}$ and $\atom(v_k')$ is the node on $p_{1 \rightarrow n}$ which is closest node (with regards to distance) to $\atom(v_k)$. We illustrate a subtree of such a join tree in~\cref{fig:vkProofIdea}. 

We now prove that $\atom(v_k')$ must contain some variable $z_j \in \Xi$, where $j \in [k-1]$. We know that $\atom(v_k')$ lies on the path $p_{1 \rightarrow k}$, therefore, because of~\cref{lemma:TreePath}, $\atom(v_k')$ must lies on the path $p_{i \rightarrow i+1}$ for some $i \in [k-1]$. Since each atom which lies on the path $p_{i \rightarrow i+1}$ must contain the variable $z_i$, it follows that $\atom(v_k')$ contains the variable $z_j$ for some $j \in [k-1]$.~\cref{fig:OnlyIf} illustrates why all nodes on the path $p_{i \rightarrow i+1}$ for $i \in [k-1]$ must contain the variable $z_i$ (because $v_{i+1}$ is a parent of $v_i$ for $i \in [k-1]$).

We now show that $\atom(v_k')$ must also contain the variable $z_l \in \Xi$ for some $l \in \{k+1, \dots, n\}$. We know that $\atom(v_k')$ lies on the path $p_{k \rightarrow n}$, therefore, because of~\cref{lemma:TreePath}, $\atom(v_k')$ must lies on the path $p_{i \rightarrow i+1}$ for some $i \in \{k, \dots, n-1 \}$. Since each atom which lies on the path $p_{i \rightarrow i+1}$ must contain the variable $z_{i+1}$ for $i \in \{k, \dots, n-1\}$, it follows that $\atom(v_k')$ contains the variable $z_l$ for some $l \in \{k+1, \dots, n\}$. Next, we consider the possible labels of $v_k'$.

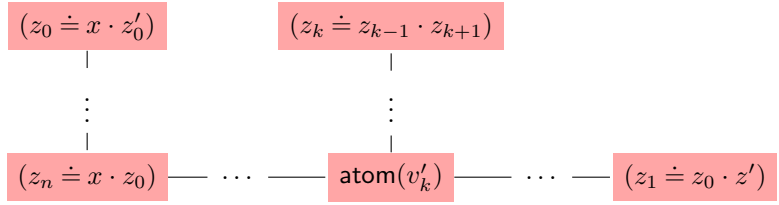
\begin{figure}
\begin{tikzpicture}[shorten >=1pt,->]
\tikzstyle{vertex}=[rectangle,fill=red!35,minimum size=12pt,inner sep=4pt]
\tikzstyle{vertex2}=[rectangle,fill=white!35,minimum size=12pt,inner sep=4pt]
\node[vertex] (1) at (4,2) {$(z_0 \logeq x \cdot z_0')$};
\node[vertex2] (2) at (4,1) {$\vdots$};
\node[vertex] (3) at (4,0) {$(z_n \logeq x \cdot z_0)$};
\node[vertex2] (4) at (6,0) {\dots};
\node[vertex] (5) at (8,0) {$\atom(v_k')$};
\node[vertex2] (6) at (10,0) {\dots};
\node[vertex] (7) at (12,0) {$(z_1 \logeq z_0 \cdot z')$};
\node[vertex2] (8) at (8,1) {$\vdots$};
\node[vertex] (9) at (8,2) {$(z_k \logeq z_{k-1} \cdot z_{k+1})$};

\path [-](1) edge node[left] {} (2);
\path [-](2) edge node[left] {} (3);
\path [-](3) edge node[left] {} (4);
\path [-](4) edge node[left] {} (5);
\path [-](5) edge node[left] {} (6);
\path [-](6) edge node[left] {} (7);
\path [-](5) edge node[left] {} (8);
\path [-](8) edge node[left] {} (9);
\end{tikzpicture}\hspace{1cm}
\caption{\label{fig:vkProofIdea} A subtree of a join tree with nodes $\atom(v_0)$, $\atom(v_n)$, $\atom(v_k)$, $\atom(v_k')$ and $\atom(v_1)$. This figure is used to illustrate Case 3.3.}
\end{figure}

\begin{itemize}
\item $\labelFunction(v_k') = z_0$. This implies $\atom(v_k') = \atom(v_0)$. We can therefore state, without loss of generality, that $\atom(v_0) = (z_0 \logeq z_j \cdot z_l)$. However, if this is the case then $x$ is not a variable of $\atom(v_0)$. Hence, $\labelFunction(v_k')=z_0$ cannot hold.
\item $\labelFunction(v_k') = z_j$ where $j \in [k-1]$. This implies that, without loss of generality, $\atom(v_k') = (z_j \logeq z_0 \cdot z_l)$. If this is the case then $j=1$ must hold, since this is the only value for $j$ such that $(z_j \logeq z_0 \cdot z_l)$ can hold. We can therefore say that $\atom(v_k') = \atom(v_1)$. For all nodes $v \in \mathcal{V}$, let $D(v)$ be the distance from the root of $\mathcal{T}$ to $v$. Since $v_l$ cannot be a redundant node, it follows that $D(v_1) + 1 \leq D(v_l)$. This implies that $D(v_k) + k - 1 + 1 \leq D(v_k) + l - k$ and hence, $k \leq \frac{l}{2}$. Because $\atom(v_n) = (z_n \logeq z_0 \cdot x)$ and $v_0$ is not redundant, we can also say that $D(v_n) + 1 \leq D(v_0)$ and hence, $D(v_k) + n - k + 1 \leq D(v_k) + k$ and therefore $n+1 \leq 2k$. Consequently, $\frac{n+1}{2} \leq k \leq \frac{l}{2}$ and hence, $n+1 \leq l$. This is a contradiction since $l \in \{k+1, \dots, n\}$. This proves that $\labelFunction(v_k) = z_j$ cannot hold.
\item $\labelFunction(v_k') = z_l$ for $l \in \{k+1, \dots, n \}$. We split this case into two parts:
\begin{itemize}
\item $\labelFunction(v_k') = z_n$. This implies that $\atom(v_k') = \atom(v_n)$. We remind the reader that $\atom(v_n) = (z_n \logeq z_0 \cdot x)$. Therefore, $\atom(v_n)$ does not contain the variable $z_j$ for $j \in [k-1]$, yet we know that $\atom(v_k')$ does contain the variable $z_j$. Consequently, $\labelFunction(v_k') = z_n$ cannot hold.
\item $\labelFunction(v_k') = z_l$ where $l \in \{k+1, \dots, n-1\}$. This implies that, without loss of generality, $\atom(v_k') = (z_l \logeq z_0 \cdot z_j)$. This cannot hold since if $k < l < n$, then $\atom(v_l)$ contains the variable $z_{l+1}$. However, $\atom(v_k')$ does not contain the variable $z_{l+1}$. 
\end{itemize}
\end{itemize}

Consequently, we have proven that if $\tau(v_1') = z_n$, then there does not exist a valid label for the non-leaf node $v_k'$, where $\atom(v_k')$ is the closest node to $\atom(v_k)$ on the path $p_{1 \rightarrow n}$. Therefore $\labelFunction(v_1') = z_n$ cannot hold.

\proofsubparagraph{Case 3.4: $\labelFunction(v_1') = z_0$.} 
This implies that $\atom(v_1') = \atom(v_0)$. Without loss of generality, $\atom(v_0) = (z_0 \logeq x \cdot z_i)$. We can see that $k < i \leq n$, since if $1 \leq i \leq k$ then $z_i \sqsubset z_0 \sqsubset z_i$ which cannot hold. We now claim that $n>2$ must hold. For sake of contradiction, assume $n=2$. Since we know that $v_k$ is the lowest common ancestor of $v_0$ and $v_n$, it follows that $k=1$. It also follows that $i=n$ since $k < i \leq n$. The distance from $v_k$ to $v_n$ is one and the distance from $v_k$ to the children of $v_0$ is two. Since $v_0$ has a child with the label $z_n$, it follows that $v_n$ is a redundant node, and hence it is not an $x$-parent. We know this cannot hold and hence $n=2$ cannot hold. Therefore, we have proven that $n>2$.

We now consider $\atom(v_{n-1}')$ which is the atom of the path $p_{n \rightarrow 0}$ which is closest to $\atom(v_{n-1})$. Since $n>2$, it follows that $v_{n-1} \neq v_1$. The nodes $v_0$ and $v_n$ are arbitrary and therefore $v_0$ and $v_n$ can be thought of as being symmetric. Thus, it must hold that $\atom(v_{n-1}') = (z_n \logeq x \cdot z_j)$ where $0 \leq j < k$ (in the same way that $\atom(v_1') = (z_0 \logeq x \cdot z_i)$ where $k<i \leq n$). We therefore have that $z_0 \sqsubset z_j$ and $z_n \sqsubset z_i$ since $0 < j < k$ and $k < i < n$. Here lies our contradiction, since $z_i \sqsubset z_0 \sqsubset z_j$ and $z_j \sqsubset z_n \sqsubset z_i$ cannot hold simultaneously.

Since we have considered all cases for $\atom(v_1')$ and have shown a contradiction for each, we know that if $\decomp_{\tilde\alpha}$ is not $x$-localized for some $x \in \var(\decomp_{\tilde\alpha})$, then $\decomp_{\tilde\alpha}$ is cyclic. 
\end{proof}

\section{Proof of~\autoref{polytime}}
\polytime*
\begin{proof}
Let $\alpha \df \alpha_1 \cdot \alpha_2 \cdots \alpha_n$ where $\alpha_i \in \Xi$ for $i \in [n]$. 
For any $i,j \in \mathbb{N}$ such that $1 \leq i \leq j \leq n$, we use $\alpha[i,j]$ to denote $\alpha_i \cdot \alpha_{i+1} \cdots \alpha_j$. 
We now give an algorithm to determine whether $\alpha$ is acyclic. 
This algorithm is essentially a bottom-up implementation of~\cref{lemma:cycledistance}. 
\cref{algorithm:acycPat} is the main algorithm and~\cref{algorithm:IsAcyclic} is a ``helper procedure''.

\begin{algorithm}
\SetAlgoLined
\SetKwInOut{Input}{Input}
\SetKwInOut{Output}{Output}
\Input{$\alpha \in \Xi^+$, where $|\alpha| = n$.}
\Output{True if $\alpha$ is acyclic, and False otherwise.}

$V \leftarrow \{ (i,i), (i+1,i+1), (i,i+1) \mid i \in [n-1] \} $;

$E' \leftarrow \{ ((i,i+1),(i,i),(i+1,i+1)) \mathrel{|} i \in [n-1] \}$;

$E \leftarrow \emptyset$;
 
\While{$E' \neq E$}{
	$E \leftarrow E'$;
 	
 	\For{$i,k \in [n]$ where $i < k$}{
 		\For{$j \in \{i, i+1, \dots, k-1\}$ where $((i,k),(i,j),(j+1,k)) \notin E'$} {\
 			\If{$(i,j),(j+1, k) \in V$ and $\mathsf{IsAcyclic}(i,j,k, \alpha, E')$}{
	        		Add $((i,k),(i,j),(j+1,k))$ to $E'$; 
	        
	        		Add $(i,k)$ to $V$;
    			}	
 		}
 	}

}

Return $\mathsf{True}$ if $(1,n) \in V$, and $\mathsf{False}$ otherwise;

\caption{Acyclic Pattern Algorithm.\label{algorithm:acycPat}}
\end{algorithm}

\begin{algorithm}
\SetAlgoLined
\SetKwInOut{Input}{Input}
\SetKwInOut{Output}{Output}
\Input{$i,j,k \in [|\alpha|]$, $\alpha \in \Xi^+$, $E'$}
\Output{True if $\alpha[i,j]$ is acyclic, and False otherwise}

\uIf{$\alpha[i,j] = \alpha[j+1,k]$}{
	Return $\mathsf{True}$;
}
 
\uElseIf{$\mathsf{var}(\alpha[i,j]) \mathrel{\cap} \mathsf{var}(\alpha[j+1,k])= \emptyset$}{
	Return $\mathsf{True}$;
}
 
\uElseIf{ $((i,j),(i,x),(x+1,j)) \in E'$ such that $\alpha[j+1,k] = \alpha[i,x]$}{
	Return $\mathsf{True}$;
} 

\uElseIf{ $((i,j),(i,x),(x+1,j)) \in E'$ such that $\alpha[j+1,k] = \alpha[x+1,j]$}{
	Return $\mathsf{True}$;
}

\uElseIf{ $((j+1,k),(j+1,x),(x+1,k)) \in E'$ such that $\alpha[i,j] = \alpha[j+1,x]$}{
	Return $\mathsf{True}$;
} 

\uElseIf{$((j+1,k),(j+1,x),(x+1,k)) \in E'$ such that $\alpha[i,j] = \alpha[x+1,k]$}{
	Return $\mathsf{True}$;
} 

\Else{  
	Return $\mathsf{False}$ \;
}
\caption{$\mathsf{IsAcyclic}$.\label{algorithm:IsAcyclic}}
\end{algorithm}

\proofsubparagraph{Correctness.} We first give a high-level overview. 
The algorithm works using a bottom-up approach, continuously adding larger acyclic subpatterns of $\alpha$ to the set $V$. 
Each subpattern is stored in $V$ as two indices for the start and end positions of the subpattern. 
To ensure that the subpatterns we are adding are acyclic, we also store an edge relation, $E$. 
The subroutine $\mathsf{IsAcyclic}$ is given two acyclic subpatterns ($\alpha[i, j]$ and $\alpha[j+1, k]$), and uses $E$ to determine whether there exists $\tilde\beta \in \brac(\alpha[i, j] \cdot \alpha[j+1, k])$ such that $\tilde\beta$ is acyclic. 
That is, the decomposition of $\tilde\beta$ is $x$-localized for all variables, see~\cref{lemma:cycledistance}. 
$\mathsf{IsAcyclic}$ is given in~\cref{algorithm:IsAcyclic} and terminates when $E$ has reached a fixed-point. 

First, assume that $\mathsf{IsAcyclic}$ returns true (given $i,j,k, \alpha$, and $E$) if and only if there exists $\tilde\alpha_1 \in \brac(\alpha[i,j])$ and $\tilde\alpha_2 \in \brac(\alpha[j+1, k])$ such that $(\tilde\alpha_1 \cdot \tilde\alpha_2)$ is acyclic. 
We now consider the while loop given on line 4. 
This loop continuously adds $(i,k)$ to $V$ if and only if there exists $(i,j), (j+1,k) \in V$ such that there exists $\tilde\alpha_1 \in \brac(\alpha[i,j])$ and $\tilde\alpha_2 \in \brac(\alpha[j+1, k])$ where $(\tilde\alpha_1 \cdot \tilde\alpha_2)$ is acyclic. 
We also add the edge $((i,k),(i,j),(j+1,k))$ to $E$ to denote that $(i,j)$ and $(j+1,k)$ are the left and right children of $(i,k)$ respectively.

This while loop terminates when $E$ reaches a fixed-point (hence, no more acyclic subpatterns of the input pattern can be derived from $E$). 
Then, either $(1, n) \in V$ and therefore $\alpha$ is acyclic, or $(1,n) \notin V$ and $\alpha$ is cyclic. 
Therefore, as long as the subroutine $\mathsf{IsAcyclic}$ is correct, our algorithm is correct.

\begin{figure}
\begin{tikzpicture}[shorten >=1pt,->]
\tikzstyle{vertex}=[rectangle,fill=white!25,minimum size=12pt,inner sep=2pt]
\node[vertex] (1) at (0,0) {$\alpha[i,k] \; (z')$};
\node[vertex] (2) at (2,-1)   {$\alpha[j+1,k] \; (x_1)$};
\node[vertex] (3) at (-2,-1)  {$\alpha[i,j] \; (z)$};
\node[vertex] (4) at (-4,-2) {$\alpha[i,x] \; (x_1)$};
\node[vertex] (5) at (0,-2) {$\alpha[x+1, j] \; (x_2)$};
\node[vertex] (6) at (-4,-2.5) {$\vdots$};
\node[vertex] (7) at (0,-2.5) {$\vdots$};

\path [-](1) edge node[left] {} (2);
\path [-](1) edge node[left] {} (3);
\path [-](3) edge node[left] {} (4);
\path [-](3) edge node[left] {} (5);
\end{tikzpicture}
\caption{\label{fig:AlgoconcatTree} Illustrating Case 3 for the correctness of the $\mathsf{IsAcyclic}$ subroutine.}
\end{figure}
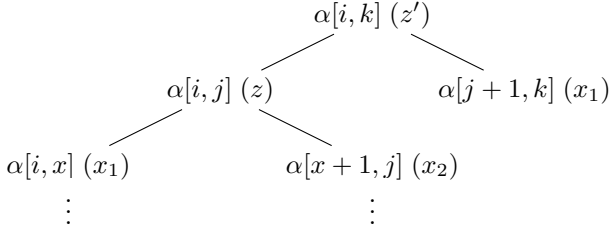

\begin{figure}
\begin{tikzpicture}[shorten >=1pt,->]
\tikzstyle{vertex}=[rectangle,fill=white!25,minimum size=12pt,inner sep=2pt]
\node[vertex] (1) at (0,0) {$v_1 \; (z)$};
\node[vertex] (2) at (-2,-1)   {$v_2 \; (z_1)$};
\node[vertex] (3) at (2,-1)  {$v_3 \; (z_2)$};

\node[vertex] (4) at (-3, -2) {$v_4 \; (x_1)$};
\node[vertex] (5) at (-1, -2) {$v_5 \; (y_1)$};

\node[vertex] (6) at (1, -2) {$v_6 \; (x_2)$};
\node[vertex] (7) at (3, -2) {$v_7 \; (y_2)$};

\node[vertex] (8) at (-3,-2.5) {$\vdots$};
\node[vertex] (9) at (-1,-2.5) {$\vdots$};
\node[vertex] (10) at (1,-2.5) {$\vdots$};
\node[vertex] (11) at (3,-2.5) {$\vdots$};

\path [-](1) edge node[left] {} (2);
\path [-](1) edge node[left] {} (3);
\path [-](2) edge node[left] {} (4);
\path [-](2) edge node[left] {} (5);
\path [-](3) edge node[left] {} (6);
\path [-](3) edge node[left] {} (7);
\end{tikzpicture}
\caption{\label{fig:AlgoconcatTree2} Illustrating the only-if direction for the correctness of the $\mathsf{IsAcyclic}$ subroutine. Note that $z_1 \neq z_2$, $z_2 \notin \{x_1, y_1\}$, and $z_1 \notin \{ x_2, y_2 \}$.}
\end{figure}
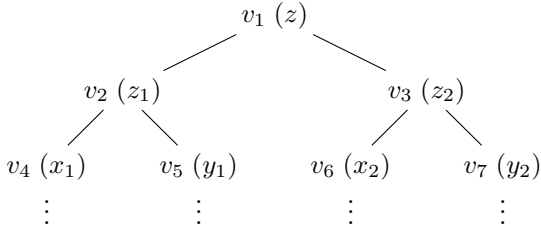

We now show that the subroutine $\mathsf{IsAcyclic}$ is correct. 
Assume $\mathsf{IsAcyclic}$ is passed $i,j,k$ (where $1 \leq i \leq j \leq k \leq n$), the pattern $\alpha \in \Xi^+$, and the edge relation $E'$. 
Since $\mathsf{IsAcyclic}$ has been passed $i$, $j$, and $k$, it follows that $(i,j), (j+1,k) \in V$ and therefore there exists $\tilde\alpha_1 \in \brac(\alpha[i,j])$ and $\tilde\alpha_2 \in \brac(\alpha[j+1,k])$ such that $\tilde\alpha_1$ and $\tilde\alpha_2$ are acyclic. 
We now prove that $\tilde\alpha \in \brac(\alpha[i,j] \cdot \alpha[j+1, k])$ is acyclic if and only if one of the following cases hold:
\begin{description}
\item[Case 1.] $\alpha[i,j] = \alpha[j+1,k]$. Since there exists an acyclic decomposition $\decomp \in\conclog$ for some $\tilde\alpha \in \brac(\alpha[i,j])$, it follows immediately from~\cref{lemma:cycledistance} that $(\tilde\alpha \cdot \tilde\alpha)$ is acyclic. Hence, $\alpha[i,k]$ is acyclic and we can add $(i,k)$ to $\mathcal{V}$.
\item[Case 2.] $\var(\alpha[i,j]) \intersect \var(\alpha[j+1,k]) = \emptyset$. Because $\alpha[i,j]$ and $\alpha[j+1,k]$ are acyclic, there exists acyclic decompositions $\decomp_1, \decomp_2 \in \conclog$ where $\decomp_1$ is the decomposition for some bracketing of $\alpha[i,j]$, $\decomp_2$ is the decomposition of some bracketing of $\alpha[j+1,k]$, and $\var(\decomp_1) \intersect \var(\decomp_2) = \emptyset$. Therefore, $\decomp \df \decomp_1 \land \decomp_2 \land (z \logeq z' \cdot z'')$ is an acyclic decomposition for some $\tilde\alpha \in \brac(\alpha[i,k])$, where $z \in \Xi$ is a new variable, and $z'$ and $z''$ are the root variables for $\decomp_1$ and $\decomp_2$ respectively. It follows from~\cref{lemma:cycledistance} that $\decomp$ is acyclic.
\item[Case 3.] $\tilde\alpha \df ((\tilde\alpha_2 \cdot \tilde\beta) \cdot \tilde\alpha_2)$ for some $\tilde \beta \in \brac$. This implies that $\tilde\alpha_1 \df (\tilde\alpha_2 \cdot \tilde\beta)$. Let $\decomp_1 \in \conclog$ be an acyclic decomposition of $\tilde\alpha_1$. Let $(z \logeq x_1 \cdot x_2)$ be the root atom of $\decomp_1$, where $x_1$ represents the bracketing $\tilde\alpha_2$. Therefore, the decomposition of $\tilde\alpha$ can be obtained from adding the atom $(z' \logeq z \cdot x_1)$ to $\decomp$ where $z' \in \Xi$ is a new variable. We illustrate a concatenation tree for this case in~\cref{fig:AlgoconcatTree} where nodes of the concatenation tree are denoted by factors of $\alpha$. Assuming $\alpha[i,x]$ and $\alpha[x+1,j]$ are acyclic, it is clear that $\decomp$ is $x$-localized for all $x \in \var(\decomp)$. Hence, $(\tilde\alpha_1 \cdot \tilde\alpha_2)$ is acyclic. 
\item[Case 4.] $\tilde\alpha \df ((\tilde\beta \cdot \tilde\alpha_2) \cdot \tilde\alpha_2)$ for some $\tilde \beta \in \brac$. Follows analogously to Case 3 because it is a simple permutation of the bracketings.
\item[Case 5.] $\tilde\alpha \df (\tilde\alpha_1 \cdot (\tilde\alpha_1 \cdot \tilde\beta))$ for some $\tilde \beta \in \brac$. Follows analogously to Case 3 because it is a simple permutation of the bracketings.
\item[Case 6.] $\tilde\alpha \df (\tilde\alpha_1 \cdot (\tilde\beta \cdot \tilde\alpha_1))$ for some $\tilde \beta \in \brac$. Follows analogously to Case 3 because it is a simple permutation of the bracketings.
\end{description}

Each condition has a corresponding if-condition in the subroutine $\mathsf{IsAcyclic}$. 
Therefore, we know that if $\mathsf{IsAcyclic}$ returns true, given $i,j,k$ (where $1 \leq i \leq j \leq k \leq n$), the pattern $\alpha \in \Xi^+$, and the relation $E'$, then $\alpha[i,k]$ is acyclic.

Now assume non of the above conditions hold. 
Let $\decomp_1$ be the acyclic decomposition of $\tilde\alpha_1$ and let $\decomp_2$ be the acyclic decomposition of $\tilde\alpha_2$. 
Let $(z_1 \logeq x_1 \cdot y_1)$ be the root atom of $\decomp_1$, and let $(z_2 \logeq x_2 \cdot y_2)$ be the root atom of $\decomp_2$. 
The decomposition of $\tilde\alpha \df (\tilde\alpha_1 \cdot \tilde\alpha_2)$ would be $\decomp \df \decomp_1 \land \decomp_2 \land (z \logeq z_1 \cdot z_2)$, where $z \in \Xi$ is a new variable. 
We illustrate part of the concatenation tree for $\decomp$ in~\cref{fig:AlgoconcatTree2}. 
Due to the fact that $\tilde\alpha_1 \neq \tilde\alpha_2$ it follows that $z_1 \neq z_2$. 
Furthermore, because Cases 3 to 6 do not hold, we know that $z_1 \notin \{x_2,y_2\}$ and $z_2 \notin \{x_1,y_1 \}$. 
However, since $\var(\decomp_1) \intersect \var(\decomp_2) \neq \emptyset$ it follows that there exists some $x \in \var(\decomp_1) \intersect \var(\decomp_2)$ such that $\decomp$ is not $x$-localized. 
Hence $\decomp$ is cyclic. 
Notice that $x_1$ (or $y_1$) \emph{could} be in the set $\{x_2, y_2 \}$. 
But if this is the case, then $z_1$ and $z_2$ are both $x_1$-parents (or $y_1$-parents) and $z$ is not an $x_1$-parent ($y_1$-parent), hence $\decomp$ is not $x_1$-localized.

\proofsubparagraph{Deriving the concatenation tree.} 
If $(1,n) \in \mathcal{V}$, then we know that $\alpha$ is acyclic. 
We can then use $V$ and $E$ to derive a concatenation tree, $\mathcal{T}\df (\mathcal{V}, \mathcal{E}, <, \Gamma, \labelFunction, v_r)$, for some acyclic decomposition $\decomp_{\tilde\alpha} \in \conclog$ of $\tilde\alpha \in \brac(\alpha)$. 
This procedure is given in the following construction: 
\begin{enumerate}
\item Let $v_r = (1,n)$. 
\item While there exists some leaf node of $\mathcal{T}$ of the form $(i,k)$ where $i \neq k$, do:
\begin{enumerate}
\item Find some $j \in \{ i, i+1, \dots, k\}$ such that one of the following conditions holds:
\begin{enumerate}
\item $\alpha[i,j] = \alpha[j+1,k]$, or $\var(\alpha[i,j]) \intersect \var(\alpha[j+1,k]) = \emptyset$, then 
\begin{enumerate}
\item Add $\{(i,k), (i,j) \}$ and $\{(i,j), (j+1, k) \}$ to $\mathcal{E}$, and let $(i,j) < (j+1,k)$.
\end{enumerate}
\item There exists $x$ such that $((i,j),(i,x),(x+1,j)) \in E$ and, $\alpha[i,x]= \alpha[j+1,k]$ or $\alpha[x+1,j] = \alpha[j+1,k]$, then:
\begin{enumerate}
\item Add $\{(i,k), (i,j) \}$ and $\{(i,j), (j+1, k) \}$ to $\mathcal{E}$, and let $(i,j) < (j+1,k)$. 
\item Add $\{(i,j), (i,x) \}$ and $\{(i,j), (x+1, j) \}$ to $\mathcal{E}$, and let $(i,x) < (x+1,j)$.
\end{enumerate}
\item There exists $x$ such that $((j+1,k),(j+1,x),(x+1,k)) \in E$, and $\alpha[i,j] = \alpha[j+1,x]$ or $\alpha[i,j] = \alpha[x+1,k]$, then:
\begin{enumerate}
\item Add $\{(i,k), (i,j) \}$ and $\{(i,j), (j+1, k) \}$ to $\mathcal{E}$, and let $(i,j) < (j+1,k)$. 
\item Add $\{(j+1,k), (j+1,x) \}$ and $\{(j+1,k), (x+1, j) \}$ to $\mathcal{E}$, and let $(j+1,x) < (x+1,j)$.
\end{enumerate}
\end{enumerate} 
\end{enumerate}
\end{enumerate}

During the construction, we assume that $\mathcal{V}$ is always updated to be the set of nodes that the edge relation $\mathcal{E}$ uses. 
For intuition, we are essentially taking the relation $E$, which have been computed by~\cref{algorithm:acycPat}, and choosing one binary tree from this set of edges. 
Some care is needed to ensure that the binary tree we choose will result in a concatenation tree for an acyclic decomposition. 
This is why we cannot choose any edge from $E$ recursively.

Once the tree has been computed, we mark each node with a variable, such that: $(1,n)$ is marked with $\strucvar$, $(i,i)$ is marked with $x$ where $\alpha[i,i] = x$, and each $(i,j)$, where $i \neq j$ and either $i \neq 1$ or $j \neq n$, is marked with $x_\beta$ where $\beta = \alpha[i,j]$. 
We then prune the tree, as defined in~\cref{defn:concatenationTree}. 
The resulting tree is the concatenation tree for some decomposition of some acyclic $\tilde\alpha \in \brac(\alpha)$.

\proofsubparagraph{Complexity.} 
We first consider the subroutine $\mathsf{IsAcyclic}$. 
The first two if-statements (lines 11 and 13), run in $\bigO(n)$ time. 
The if-statements on lines 15, 17, 19, and 21 run in time $\bigO(n)$ due to the fact that there are $\bigO(n)$ such values for $x$, and it times $\bigO(1)$ time to check whether the two factors of $\alpha$ are equal (after linear time preprocessing, see~\cref{compModel}). 
Therefore, $\mathsf{IsAcyclic}$ runs in time $\bigO(n)$.

The set $\mathcal{V}$ holds substrings of $\alpha$, and therefore $|\mathcal{E}| \leq n^3$, since each $(i,k) \in \mathcal{V}$ has $\bigO(n)$ outgoing edges. 
It follows that the while loop from line 4 to line 14 is iterated $\bigO(n^3)$ times. 
The for loop on line 6 is iterated $\bigO(n^2)$ times. 
The for loop on line 7 is clearly iterated $\bigO(n)$ times. 
Therefore, the whole algorithm runs in time $\bigO(n^7)$.

We now consider the complexity of deriving the concatenation tree. 
There are $\bigO(n)$ nodes in a concatenation tree, and given a node $(i,j)$, where $i \neq j$, finding an edge $((i,k), (i,j), (j+1,j)) \in E$ takes at most $\bigO(n^3)$ time, since there are at most $n$ such values for $j$ and making sure the relative conditions hold (in the above construction) takes $\bigO(n^2)$ time, as we have previously discussed when discussing the time complexity for~\cref{algorithm:acycPat}.
Therefore, deriving the concatenation tree, without pruning, takes $\bigO(n^4)$ time. 
Finally, pruning the concatenation tree takes $\bigO(n^2)$ time, since we consider each variable that labels a node, traverse the tree to find the $\ll$-maximum (see~\cref{defn:concatenationTree}), and prune accordingly. 
Therefore, we can derive the concatenation tree from $V$ and $E$ in time $\bigO(n^4)$.
\end{proof}

\section{Proof of~\autoref{lemma:normalization}}

Before proving~\cref{lemma:normalization}, we restate the definition of normalized $\cpfcreg$s. We call an $\cpfc$ with body $\bigwedge_{i=1}^n (x_i \logeq \alpha_i)$ \emph{normalized} if for all $i,j \in [n]$, the following conditions hold:
\begin{description}
\item[Condition 1.] $\alpha_i \in \Xi^+$,
\item[Condition 2.] $x_i \notin \var(\alpha_i)$ and $\strucvar \notin \var(\alpha_i)$, and 
\item[Condition 3.] $\alpha_i = \alpha_j$ if and only if $i = j$.
\end{description}

If an $\cpfcreg$ has body $\bigwedge_{i=1}^n (x_i \logeq \alpha_i) \land \bigwedge_{j=1}^m (y_j \regconst \gamma)$, then it is \emph{normalized} if the subformula $\bigwedge_{i=1}^n (x_i \logeq \alpha_i)$ is normalized. 

\normalization*
\begin{proof}
Let $\varphi \df \cqhead{\vec{x}} \bigwedge_{i=1}^n \eta_i$ be an $\cpfc\noconstr$. We give a way to construct a normalized formula $\varphi' \in \cpfcreg$ where $\varphi'$ is equivalent to $\varphi$. 

\proofsubparagraph{Condition 1.} For all $i \in [n]$ assume that $\eta_i = (x \logeq \alpha)$ where $\alpha \in (\Sigma \union \Xi)^*$. We now consider the unique factorization for $\alpha \df \beta_1 \cdot \beta_2 \cdots \beta_k$ for some $k \in \mathbb{N}$, where for all $\beta_j$ where $j \in [k]$, either $\beta_j \in \Xi^+$ or $\beta_j \in \Sigma^+$. Furthermore, if $\beta_{j} \in \Xi^+$ then $\beta_{j+1} \in \Sigma^+$, and if $\beta_j \in \Sigma^+$ then $\beta_{j+1} \in \Xi^+$ for all $j \in [k-1]$. We then replace each $\beta_i$ where $\beta_i \in \Sigma^+$ with a new variable $z_i \in \Xi$ and add the regular constraint $(x_i \regconst \beta_i)$ to $\varphi$. This takes linear time by scanning each $\eta_i$ from left to right, and replacing each $\beta_i \in \Sigma^+$ with a new variable.

\proofsubparagraph{Condition 2.} While there exists an atom of $\varphi$ of the form $(x \logeq \alpha_1 \cdot x \cdot \alpha_2)$, we define an $\cpfc\noconstr$-formula $\psi$ with the following body:
\[(x \logeq z) \land \bigwedge_{y \in \var(\alpha_1 \cdot \alpha_2)} (y \logeq \emptyword),\]
where $z \in \Xi$ is a new variable. We then replace $(x \logeq \alpha_1 \cdot x \cdot \alpha_2)$ in $\varphi$ with $\psi$. We can show the $\psi$ is equivalent to $(x \logeq \alpha_1 \cdot x \cdot \alpha_2)$ by a simply counting argument. Given any $\sigma$ which satisfies $(x \logeq \alpha_1 \cdot x \cdot \alpha_2)$, we have that $|\sigma(x)| \mathrel{=} |\sigma(\alpha_1)| \mathrel{+} |\sigma(x)| \mathrel{+} |\sigma(\alpha_2)|$ and hence, $|\sigma(\alpha_1)| + |\sigma(\alpha_2)| = 0$, which implies that $\sigma(\alpha_1) = \sigma(\alpha_2) = \emptyword$. 

While there exists an atom of $\varphi$ of the form $\eta_i = (x_i \logeq \alpha_1 \cdot \strucvar \cdot \alpha_2)$, we can replace $\eta_i$ with the subformula $\psi$ with body:
\[ (\strucvar \logeq x_i) \land \bigwedge_{y \in \var(\alpha_1 \cdot \alpha_2)} (y \logeq \emptyword). \]

We show that replacing $\eta_i$ with $\psi$ results in an equivalent formula using a counting argument. It follows that $|\sigma(x_i)| \mathrel{=} |\sigma(\alpha_1)| \mathrel{+} |\sigma(\strucvar)| \mathrel{+} |\sigma(\alpha_2)|$. Furthermore, we know that $|\sigma(x_i)| \leq |\sigma(\strucvar)|$ and therefore it must hold that $|\sigma(x_i)| = |\sigma(\strucvar)|$, which implies that $\sigma(x_i) = \sigma(\strucvar)$. Therefore, $|\sigma(\alpha_1)| \mathrel{+} |\sigma(\alpha_2)| = 0$ which can only hold if $\sigma(\alpha_1) \cdot \sigma(\alpha_2) = \emptyword$.

The process defined takes polynomial time, since for each atom, we linearly scan the right-hand side. If it does, then we replace a word equation with $\psi$, as described above. Since we perform a linear scan, this takes $\bigO(|\varphi|)$ time.

\proofsubparagraph{Condition 3.} If two atoms are identical, then one can be removed. If $\eta_i = (x_i \logeq \alpha)$ and $\eta_j = (x_j \logeq \alpha)$ where $x_i \neq x_j$, then we can replace $\eta_j$ in $\varphi$ with $(x_j \logeq x_i)$. This takes $\bigO(|\varphi|^2)$ time by considering every pair of atoms.

Since we are always replacing a subformula of $\varphi$ with an equivalent subformula, it follows that the result of the above construction is equivalent and it is normalized. Furthermore, we have shown that the re-writing procedure defined takes $\bigO(|\varphi|^2)$ time.
\end{proof}

\section{Proof of~\autoref{lemma:subtree}}

If $T \df (V,E)$ is a tree and $V' \subset V$, then the induced subgraph of $T$ on $V'$ is the graph $G \df (V', E')$ where $e \in E'$ if and only if $e \in E$ and the two endpoints of $e$ are in the set $V'$. Notice that $G$ is not necessarily a tree, because $G$ may not be connected.

\subtree*
\begin{proof}
Let $\varphi \in \cpfc\noconstr$ be an acyclic formula defined as $\varphi \df \cqhead{\vec{x}} \bigwedge_{i=1}^{n} \eta_i$. 
Let $\decomp_\varphi \in \conclog\noconstr$ be an acyclic decomposition of $\varphi$ and let $T \df (V, E)$ be a join tree of $\decomp_\varphi$. 
By definition, $\decomp_\varphi \df \cqhead{\vec{x}} \bigwedge_{i=1}^{n} \decomp_i $ where $\decomp_i$ is a decomposition of $\eta_i$ for each $i \in [n]$. 
Since $V$ contains all atoms of $\decomp$, it follows that all atoms of $\decomp_i$ are in $V$.

Let $T^i \df (V^i, E^i)$ be the induced subgraph of $T$ on the atoms of $\decomp_i$. 
We now prove that $T^i$ is a join tree for $\decomp_i$. 
By definition, we know that all atom of $\decomp_i$ are present in $T^i$ and that no cycles exist in $T^i$ (since it is a subgraph of $T$). 
Therefore, to show that the resulting structure is a join tree, it is sufficient to show that this structure is connected. 

We prove that $T^i$ is connected by a contradiction. 
Let $(z_1 \logeq z_2 \cdot z_3), (z_4 \logeq z_5 \cdot z_6) \in V^i$ be two nodes of $T^i$ which we assume are not connected. 
Let $\mathcal{T} \df (\mathcal{V}, \mathcal{E}, <, \Gamma, \labelFunction, v_r)$ be the concatenation tree for $\decomp_i$. 
Let $v_1,v_n \in \mathcal{V}$ be non-leaf nodes of $\mathcal{T}$ such that $\atom(v_1) = (z_1 \logeq z_2 \cdot z_3)$ and $\atom(v_n) = (z_4 \logeq z_5 \cdot z_6)$. Let $(v_1, v_2, \dots, v_n)$ be the sequence of nodes which exist on the path in the concatenation tree from $v_1$ to $v_n$. 
Let $k \in [n]$ such that $v_k \in \mathcal{V}$ is the lowest common ancestor of $v_1$ and $v_n$. 
Notice that $\atom(v_i)$ and $\atom(v_{i+1})$ for all $i \in [k-1]$ share the variable that labels $v_i$. 
Therefore, since $T$ is a join tree, these nodes are connected via a path where each node that lies on that path contains the variable that labels $v_i$. 
We know that no nodes removed in the manipulation contain the variable that labels $v_i$ since this is an introduced variable for $\decomp_i$ and therefore the variable that labels $v_i$ is not present in any atom of $\decomp_j$ for any $j \in [n] \setminus \{ i \}$. 
Hence, $\atom(v_i)$ and $\atom(v_{i+1})$ must be connected for all $i \in [k-1]$ in the structure resulting from the above manipulating the join tree. 
Thus, $\atom(v_1)$ and $\atom(v_k)$ are connected in this structure, by transitivity. 
The analogous reason means that $\atom(v_n)$ and $\atom(v_k)$ are connected in $T^i$. 
Hence, $\atom(v_1)$ and $\atom(v_n)$ is connected in the resulting structure and we have reached the desired contradiction. 
If $v_1$ is an ancestor of $v_n$ (or $v_n$ is an ancestor of $v_1$), then connectivity follows trivially. 

Therefore, there is a subtree of $T \df (V,E)$ that is a join tree for the decomposition of $\eta_i$. 
Due to the fact that the body of $\decomp_\varphi$ is $\bigwedge_{i=1}^n \decomp_i$ where $\decomp_i$ is a decomposition of $\eta_i$ such that the set of introduced variables for $\decomp_i$ is disjoint from the introduced variables for $\decomp_j$, where $i \neq j$, it follows that $V^i \intersect V^j = \emptyset$ for $T^i \df (V^i, E^i)$ and $T^j \df (V^j, E^j)$. 
\end{proof}

\section{Proof of~\autoref{lemma:CyclicConditions}}
\CyclicConditions*
\begin{proof}
Let $\varphi \df \cqhead{\vec{x}} \bigwedge_{i=1}^{n} \eta_i$ be a normalized $\cpfc\noconstr$, and let $\decomp_\varphi \df \cqhead{\vec{x}} \bigwedge_{i=1}^n \decomp_i$ be the decomposition of $\varphi$ where $\decomp_i$ is a decomposition of $\eta_i$ for all $i \in [n]$. We will now prove that if any of the conditions hold, then $\varphi$ is cyclic.

\proofsubparagraph{Condition 1.} For sake of a contradiction, assume $\varphi$ is an acyclic, normalized $\cpfc\noconstr$ which is weakly cyclic. Let $T \df (V, E)$ be a join tree for $\decomp_\varphi$. From~\cref{lemma:subtree}, it follows that for each $i \in [n]$ there exists a subtree $T^i$ of $T_\varphi$ which is a join tree for a decomposition of $\eta_i$. We now construct a weak join tree for $\varphi$. Let $T_w \df (V_w,E_w)$ where $V_w \df \{ \eta_i \mid i \in [n] \}$, and $\{ \eta_i, \eta_j \} \in E_w$ if and only if there is an edge $\{v_i, v_j\} \in E$ where $v_i \in V^i$ and $v_j \in V^j$ for each $i,j \in [n]$ where $i \neq j$. We now prove that this is a weak join tree for $\varphi$. 

For sake of contradiction, assume that $T_w$ is not a weak join tree for $\varphi$. By the procedure used to compute $T_w$ we know that $V_w = \{ \eta_i \mid i \in [n] \}$, and we know that this structure is a tree (we know this because if $T_w$ is not a tree, then $T$ is not a tree). Therefore, if $T_w$ is not a join tree, it follows that there exists $\eta_i \in V_w$ and $\eta_j \in V_w$ such that there is some variable $x \in \var(\eta_i) \intersect \var(\eta_j)$ where some node $\eta_k \in V_w$ exists on the path between $\eta_i$ and $\eta_j$ in $T_w$, and $x \notin \var(\eta_k)$. If this is the case, then $x \in \var(\decomp_i) \intersect \var(\decomp_j)$, and $x \notin \var(\decomp_k)$. Hence there is a path between two nodes in $T$ which contain the variable $x \in \var(\decomp_\varphi)$, which are atoms of $\decomp_i$ and $\decomp_j$, yet there is a node on the path between these nodes which does not contain the variable $x$, which is some atom of $\decomp_k$. Therefore, $T$ is not a join tree and we have reached a contradiction. Since we have reached a contradiction, $T_w \df (V_w, E_w)$ is a weak join tree for $\varphi$ and hence if $\varphi$ is weakly cyclic, we can conclude that $\varphi$ is cyclic.

\proofsubparagraph{Condition 2.} This follows directly from~\cref{lemma:subtree}. Since for any join tree $T \df (V,E)$ of a decomposition of $\varphi$, there exists a subtree which is a join tree for some decomposition of $\eta_i$, we can conclude that if $\eta_i$ is cyclic, then $\varphi$ is cyclic.

\proofsubparagraph{Condition 3.} For sake of contradiction, assume that $\varphi$ is acyclic, and assume that $|\var(\eta_i) \intersect \var(\eta_j)| > 3$ for some $i, j \in [n]$ where $i \neq j$. Let $T \df (V, E)$ be a join tree for $\decomp_\varphi$. Let $T^i$ and $T^j$ be subtrees of $T$ which are join trees for the decompositions of $\eta_i$ and $\eta_j$ respectively. Note that these trees are disjoint. Let $(z_1 \logeq x_1 \cdot y_1)$ and $(z_2 \logeq x_2 \cdot y_2)$ be nodes of $T^i$ and $T^j$ respectively, such that  $(z_1 \logeq x_1 \cdot y_1)$ is the closest node (with regards to distance) to any node in $T^j$, and $(z_2 \logeq x_2 \cdot y_2)$ is the closest node to any node in $T^i$, these atoms are well defined because $T$ is a tree. Notice that $|\var(z_1 \logeq x_1 \cdot y_1) \intersect \var(z_2 \logeq x_2 \cdot y_2)| \leq 3$. Therefore, there is a node of $T^i$ which shares a variable with some node of $T^j$, yet this variable does not exist on the path between these nodes, since $(z_1 \logeq x_1 \cdot y_1)$ must exist on such a path.

\proofsubparagraph{Condition 4.} Towards a contradiction. Assume that $\varphi$ is acyclic and there exists $i, j \in [n]$, where $i \neq j$, such that $|\var(\eta_i ) \intersect \var(\eta_j)| = 3$ and $|\eta_i|>3$ (the other case is symmetric). Let $T \df (V, E)$ be a join tree for $\decomp_\varphi \in \conclog\noconstr$. Let $T^i$ be the subtree of $T$ which is a join tree for $\eta_i$ and let $T^j$ be the subtree of $T$ which is a join tree for $\eta_j$. Since we have that $|\eta_i| > 3$, we decompose $\eta_i$ into $\decomp_i \in \conclog\noconstr$. Note that for each atom of $\decomp_i$, there is a variable $z \in \var(\decomp_i) \setminus \var(\decomp_j)$. This holds due to the fact that the set of introduced variables for $\decomp_i$ is disjoint from the set of introduced variables for $\decomp_j$ where $i,j \in [n]$ and $i \neq j$. Therefore the maximum shared variable between an atom of $\decomp_i$ and an atom of $\decomp_j$ is $2$. Using the same argument used in Condition 3, this results in a contradiction and therefore our assumption that $\varphi$ is acyclic cannot hold. 
\end{proof}

\section{Proof of~\autoref{lemma:skeletonTree}}
\skeletonTree*
\begin{proof}
Let $\varphi \df \cqhead{\vec{x}} \bigwedge_{i=1}^{n} \eta_i$ be a normalized $\cpfc\noconstr$ and let $\decomp_\varphi \df \cqhead{\vec{x}} \bigwedge_{i=1}^{n} \decomp_i$ be an acyclic decomposition of $\varphi$ such that $\decomp_i \in \conclog\noconstr$ is the decomposition of $\eta_i$ for each $i \in [n]$. Let $T\df (V, E)$ be a join tree of $\decomp_\varphi$ and let $T_s \df (V_s, E_s)$ be the skeleton tree of $T$. We work towards a contradiction, assume $T_w \df (V_w, E_w)$ is a weak join tree for $\varphi$, but there does not exist a join tree $T' \df (V', E')$ of $\decomp_\varphi$ such that $T_w$ is the skeleton tree of $T'$. We now transform $T$ to obtain the join tree $T'$, and thus reach our contradiction. 

For each $i \in [n]$, let $T^i \df (V^i, E^i)$ be the subtree of $T$ such that $T^i$ is a join tree for $\decomp_i$. We know that these subtrees are disjoint. Let $F \df (V_f, E_f)$ be a forest where $V_f \df \bigcup_{i=1}^{n} V^i$ and $E_f \df \bigcup_{i=1}^{n} E^i$. Then, for each edge $\{ \eta_i, \eta_j \} \in E_w$, let $\chi_{i,j}$ be the atom of $\decomp_i$ and $\chi_{j,i}$ the atom of $\decomp_j$ such that these are the end nodes in the shortest path from any atom of $\decomp_i$ to any atom of $\decomp_j$ in $T$. Then, add the edge $\{ \chi_{i,j}, \chi_{j,i} \}$ to $E_f$ for each $\{ \eta_i, \eta_j \} \in E_w$. Let $T' \df (V', E')$ be the result of the above augmentation of $T$.

We now prove that $T' \df (V', E')$ is a join tree for $\decomp_\varphi$. We can see that $T'$ is a tree, every atom of $\decomp_\varphi$ is a node of $T'$, and that $\var(\chi_{i,j}) \intersect \var(\chi_{j,i}) = \var(\eta_i) \intersect \var(\eta_j)$ which holds because otherwise $T$ would not be a join tree (see Conditions 3 and 4 of~\cref{lemma:CyclicConditions}). We use this last fact to show that every node that lies on the path between any $\chi, \chi' \in V'$ where $x \in \var(\chi) \intersect \var(\chi')$, also contains the variable $x$. Without loss of generality, assume that $\chi \in V^1$ and $\chi' \in V^k$ where $V^1$ and $V^k$ are the set of vertices for the join tree for the decomposition of $\eta_1$ and $\eta_k$ respectively. Further assume that the path from $\eta_1$ to $\eta_k$ in $T_w$ consists of $\{\eta_i, \eta_{i+1} \}$ for $i \in [k-1]$. Since $T_w$ is a weak join tree, and that $\eta_1$ and $\eta_k$ both contain the variable $x$, it follows that for all $i \in [k-1]$, the word equation $\eta_i$ contains the variable $x$. Furthermore, we know that for any any edge $\{ \chi_i, \chi_{i+1} \} \in E'$, where $\chi_i \in V^i$ and $\chi_{i+1} \in V^{i+1}$, that $\var(\chi_i) \intersect \var(\chi_{i+1}) =  \var(\eta_i) \intersect \var(\eta_{i+1})$, therefore $x \in \var(\chi_i) \intersect \var(\chi_{i+1})$. Because $T^i \df (V^i, E^i)$ is a join tree for $\decomp_i$, every node that lies on the path between two nodes of $V^i$ which have the variable $x$, also has the variable $x$. Furthermore, for any edge $\{ \chi_i, \chi_{i+1} \} \in E'$, where $\chi_i \in V^i$ and $\chi_{i+1} \in V^{i+1}$, we know that $x \in \var(\chi_i) \intersect \var(\chi_{i+1})$. Hence, all nodes on the path between $\chi$ and $\chi'$ contain~$x$. 
\end{proof}

\section{Proof of~\autoref{lemma:atomDecomp}}

The following is a lemma for the ``main case'' of~\cref{lemma:atomDecomp}.

\begin{lemma}
\label{constrainedBracketings}
Given a pattern $\alpha \in \Xi^+$ and a set $ C \subseteq \{ \{ x, y \} \mid x,y \in \var(\alpha) \text{ and } x \neq y \}$. We can decide in polynomial time whether there exists an acyclic $\tilde\alpha \in \brac(\alpha)$ such that for each $\{x,y \} \in C$, either $(x \cdot y) \sqsubseteq \tilde\alpha$ or $(y \cdot x) \sqsubseteq \tilde\alpha$.
\end{lemma}
\begin{proof}
We assume that every variable that appears in $C$ also appears in the input pattern, since if this does not hold, we can immediately return $\false$.
This initial check can clearly be done in polynomial time.
The algorithm used to solve the problem stated in the lemma is given in~\cref{algorithm:acycPatTwo}.
This is a variation of the algorithm given in~\cref{polytime}, but $V$ and $E'$ are initialized differently.
There is also an extra subroutine given in~\cref{algorithm:extraCheck} to deal with a special case. 
It follows from the proof of~\cref{polytime} that if then~\cref{algorithm:acycPatTwo} returns $\true$, then an acyclic concatenation tree can be derived from $E$ and $V$ in polynomial time. 
We first look at the correctness.

\begin{algorithm}
\SetAlgoLined
\SetKwInOut{Input}{Input}
\SetKwInOut{Output}{Output}
\Input{$\alpha \in \Xi^+$, where $|\alpha| = n$.}
\Output{True if $\alpha$ is acyclic, and False otherwise.}

$E' \leftarrow \{ ((i,i+1),(i,i),(i+1,i+1)) \mathrel{|} \text{ for all } c \in C \text{ we have } (i,i), (i+1,i+1) \notin c \}$;

$E' \leftarrow E' \union \{ ((i,i+1),(i,i),(i+1,i+1)) \mathrel{|} \{ (i,i), (i+1,i+1)\} \in C \}$;

$V$ is the set of nodes in $E'$;

Add $(i,i)$ to $V$ for all $i \in [n]$;

$E \leftarrow \emptyset$;
 
\While{$E' \neq E$}{
	$E \leftarrow E'$;
 	
 	\For{$i,k \in [n]$ where $i < k-1$}{
 		\For{$j \in \{i, i+1, \dots, k-1\}$ where $((i,k),(i,j),(j+1,k)) \notin E'$} {\
 			\If{$(i,j),(j+1, k) \in V$ and $\mathsf{IsAcyclic}(i,j,k, \alpha, E')$ and $\mathsf{extraCheck}(i,j,k, \alpha, C)$ }{
	        		Add $((i,k),(i,j),(j+1,k))$ to $E'$; 
	        
	        		Add $(i,k)$ to $V$;
    			}	
 		}
 	}

}

Return $\mathsf{True}$ if $(1,n) \in V$, and $\mathsf{False}$ otherwise;

\caption{A variant of the Acyclic Pattern Algorithm. The subroutine $\mathsf{IsAcyclic}$ is identical to how it was given in the proof of~\cref{polytime}. \label{algorithm:acycPatTwo}}
\end{algorithm}

\begin{algorithm}
\SetAlgoLined
\SetKwInOut{Input}{Input}
\SetKwInOut{Output}{Output}
\Input{$i,j,k,\alpha, C$}
\Output{False, if $\{x,y \} \in C$ and $x$ is concatenated to $\tilde\beta$ where $\var(\tilde\beta) \neq \{x,y\}$. \\ True, otherwise.}

\uIf{$i = j$ and there exists $\{x,y\} \in C$ where $\alpha[i,j] \in \{x,y\}$}{
	\uIf{$\var(\alpha[j+1,k]) = \{ x, y\}$ }{
		Return $\mathsf{True}$;	
	}
	\Else{
		Return $\mathsf{False}$;
	}
	
}
 
\uElseIf{$j = k$ and there exists $\{x,y\} \in C$ where $\alpha[j+1,k] \in \{x,y\}$}{
	\uIf{$\var(\alpha[i,j]) = \{ x, y\}$ }{
		Return $\mathsf{True}$;	
	}
	\Else{
		Return $\mathsf{False}$;
	}
}

\Else{  
	Return $\mathsf{True}$ \;
}
\caption{$\mathsf{extraCheck}$.\label{algorithm:extraCheck}}
\end{algorithm}

\proofsubparagraph{Correctness.}
\cref{algorithm:acycPatTwo} initializes $E'$ such that one of the following conditions must hold:
\begin{enumerate}
\item $\{((i,i+1),(i,i),(i+1,i+1))\} \in E'$ where $\{ (i,i), (i+1,i+1)\} \in C$, or
\item $\{((i,i+1),(i,i),(i+1,i+1))\} \in E'$ where for all $c \in C$ we have that $(i,i) \notin c$ and $(i+1,i+1) \notin c$.
\end{enumerate}

Furthermore, line 37 now ensures that $i < k-1$. This avoids the case where $(i,i+1)$ is added to $V$ where $(i,i+1)$ does not satisfy one of the above conditions. 

The subroutine $\mathsf{extraCase}$ ensures that if some $x \in \Xi$, where $\{x,y\} \in C$ for some $y \in \Xi$, is concatenated to some $\tilde\beta \in \brac$, then the set of variables in $\tilde\beta$ is $\{x , y \}$. 
We now consider two cases. We note that we use the shorthand $\var(\tilde\alpha)$ for any $\tilde\alpha \in \brac$ to denote the set variables that appears in $\tilde\alpha$.

\proofsubparagraph{Case 1: If $\tilde\alpha$ exists, then $\mathsf{IsAcyclic}$ returns true.}
This direction follows from the proof of~\cref{polytime}. 
However, we need to prove that the new restrictions added to the $\mathsf{IsAcyclic}$ ensures that if such an $\tilde\alpha$ (that satisfies the conditions given in the lemma statement) exists, then $\mathsf{IsAcyclic}$ still returns true.
 
Let $\alpha \in \Xi^+$ and $ C \subseteq \{ \{ x, y \} \mid x,y \in \var(\alpha) \text{ and } x \neq y \}$. Let $\tilde\alpha \in \brac(\alpha)$ such that for each $\{x,y \} \in C$, either $(x \cdot y) \sqsubseteq \tilde\alpha$ or $(y \cdot x) \sqsubseteq \tilde\alpha$. 
 
Due to the initialization of $E'$ and $V$, we know that if $(x \cdot y) \sqsubseteq \tilde\alpha$, then either there exist some $\{x,y\} \in C$, or for all $\{x',y' \} \in C$ we have that $x \notin \{x',y'\}$ and $y \notin \{x',y'\}$. 
To show that this is the correct behavior, we prove the claim that if, without loss of generality, $(x \cdot y) \sqsubseteq \tilde\alpha$ for all $\{x,y\} \in C$, and $(x \cdot z) \sqsubseteq \tilde\alpha$ where $z$ is not an element of any $\{x',y'\} \in C$, then $\tilde\alpha$ is cyclic.
To prove this claim, we work towards a contradiction.
Let $\alpha \in \Xi^+$ and assume that $\tilde\alpha \in \brac(\alpha)$ is acyclic where, without loss of generality, $(x \cdot y), (x \cdot z) \sqsubset \tilde\alpha$ and $z$ is not an element of any $\{x',y'\} \in C$ (it follows that $x \neq y$).
Let $\decomp_{\tilde\alpha} \in \conclog\noconstr$ be the decomposition of $\tilde\alpha$.
We can see that both $(z \logeq x \cdot y)$ and $(z' \logeq x \cdot z)$ are atoms of $\decomp_{\tilde\alpha}$ where $z \neq z'$.
Let $\mathcal{T} \df (\mathcal{V}, \mathcal{E}, <, \Gamma, \labelFunction, v_r)$ be the concatenation tree for $\decomp_{\tilde\alpha}$.
It follows that, there exists two nodes $v,v' \in \mathcal{V}$ where $\tau(v) = z$ and $\tau(v') = z'$ where $z$ and $z'$ are $x$-parents.
Consider the lowest common ancestor of $z$ and $z'$.
This lowest common ancestor is not an $x$-parent, since it must be a parent of two nodes labeled with an introduced variable, yet it lies on the path between $z$ and $z'$.
Hence, $\decomp_{\tilde\alpha}$ is not $x$-localized and hence $\tilde\alpha$ is cyclic.
Therefore, the initialization of $E'$ and $V$ is the correct behavior.	

Next, we look at the $\mathsf{extraCheck}$ subroutine. 
Assume that without loss of generality $(x \cdot y) \sqsubseteq \tilde\alpha$ for all $\{x, y\}$, and $\tilde\alpha$ is acyclic. 
It follows that there exists a node $v_1$ with two children $v_2$ and $v_3$ such that $\labelFunction(v_2) = x$ and $\labelFunction(v_3) = y$.
Let $\decomp_{\tilde\alpha} \in \conclog$ be the decomposition of $\tilde\alpha$, and let $\mathcal{T}$ be the concatenation tree for $\decomp_{\tilde\alpha}$.
Since $\decomp_{\tilde\alpha}$ is acyclic, it must be both $x$ and $y$ localized. 
Therefore, since $v_1$ is itself an $x$-parent, all $x$ parents form a subtree of $\mathcal{T}$ which is connected to $v_1$.
Hence, if $x$ is concatenated to $\tilde\beta$ in $\tilde\alpha$, it follows that $\var(\tilde\beta) = \{x, y\}$ must hold.  
This concludes the correctness proof for this direction.

\proofsubparagraph{Case 2: If $\mathsf{IsAcyclic}$ returns true, then $\tilde\alpha$ exists.}
If~\cref{algorithm:acycPatTwo} terminate and $(1,n) \in V$, then $\alpha$ is acyclic and we can derive a concatenation tree for some acyclic decomposition $\decomp_{\tilde\alpha}$ of $\tilde\alpha \in \brac(\alpha)$, see the proof of~\cref{polytime}. 
The derivation procedure adds edges from $E$ to the concatenation tree until the leaf nodes are all $(i,i)$ for $i \in [n]$. 
Hence, if a node has the children $(i,i)$ and $(i+1,i+1)$, it follows that these nodes must satisfy the conditions defined in the initialization of $E$.
We now show that $\{x,y\} \in C$, either $(x \cdot y) \sqsubseteq \tilde\alpha$ or $(y \cdot x) \sqsubseteq \tilde\alpha$. 
For sake of a contradiction, assume that there exists some $\{ x, y \} \in C$ such that, without loss of generality, $(x \cdot y) \sqsubseteq \tilde\alpha$ does not hold.
Due to the initialization of $E'$, it follows that there cannot exist some $(x \cdot z) \sqsubseteq \tilde\alpha$ such that $z \neq y$. 
Furthermore, if $x \in \Xi$ is concatenated to some $\tilde\beta \sqsubset \tilde\alpha$, then it follows that $\var(\tilde\beta) = \{x, y \}$. 
Hence, without loss of generality, $(x \cdot y) \sqsubseteq \tilde\beta$ holds.
We also do a preprocessing step to make sure that all the variables that appear in $C$, also appear in $\alpha$.
Therefore, the resulting concatenation tree represents an acyclic bracketing $\tilde\alpha$ of the input pattern $\alpha$, where $(x \cdot y)$ or $(y \cdot x)$ is a subbracketing of $\tilde\alpha$ for all $\{x, y\}\in C$.

\proofsubparagraph{Complexity.}
Due to the fact that~\cref{algorithm:acycPatTwo} is almost identical to the algorithm given in the proof of~\cref{polytime}, it is sufficient to prove that it takes polynomial time to initialize $V$ and $E$, and that the subroutine $\mathsf{extraCheck}$ can be executed in polynomial time. We can assume that we precompute the set $\bar{C} \df \bigcup_{s \in C} s$.

We first consider the initialization of $V$ and $E'$.
For each $i \in [n-1]$, we check whether $\{ \alpha[i], \alpha[i+1] \} \in C$, and if that is false, we check whether $\alpha[i], \alpha[i+1] \notin \bar{C}$. 
Therefore, the initialization of $E'$ takes $\bigO(n)$, since the checks for each $i \in [n-1]$ takes constant time, and adding to $E'$ takes constant time. 
Furthermore, adding all nodes of $E'$ to $V$ takes $\bigO(|E'|)$ time, and since $|E'| \in \bigO(n)$, this also takes $\bigO(n)$ time. 
Now, we consider the time complexity of the $\mathsf{extraCheck}$ subroutine. 
Deciding whether $i=j$ and $\alpha[i] \in \bar{C}$ takes constant time (line 47), and deciding whether $\var(\alpha[j+1,k]) =  \{x, y\}$ takes $\bigO(n)$ time. 
Since the other case is symmetric, the total running time of $\mathsf{extraCheck}$ is $\bigO(n)$.
Therefore, it follows form the proof of~\cref{polytime} that~\cref{algorithm:acycPatTwo} runs in time $\bigO(n^7)$.
\end{proof}

\subsection*{Actual proof of~\autoref{lemma:atomDecomp}.}
\atomDecomp*
\begin{proof}
If for all $\{ x,y \} \in C$, we have that $x,y \in \var(\alpha)$, then we know that this problem can be decided in time $\bigO(n^7)$. 
We use~\cref{constrainedBracketings} and since we can decide whether there exists an acyclic bracketing $\tilde\alpha \in \brac(\alpha)$ such that $(x \cdot y) \sqsubset \tilde\alpha$ or $(y \cdot x) \sqsubseteq \tilde\alpha$.
If such a decomposition exists, it follows that $(z_1 \logeq x \cdot y)$ or $(z_1  \logeq y \cdot x)$, for some $z_1 \in \Xi$, is an atom in the decomposition of $(z \logeq \alpha)$, where $\tilde\alpha \in \brac(\alpha)$ is the bracketing used for the decomposition. 

If for some $\{x,y\} \in C$, we have that $x = z$, then we know $y \in \var(\alpha) \setminus \{ z \}$ since $x \neq y$.  
We now claim that the acyclic decomposition $\decomp \in \conclog\noconstr$ exists, in the case where $x=z$, if and only if there exists $i,j \in \mathbb{N}$ such that $\alpha = y^i \cdot \beta \cdot  y^j$ where $\beta$ is acyclic and $|\beta|_y = 0$.

For the if direction, we give the following bracketing of $\alpha$:
\[  \tilde\alpha \df ( ( ( y \cdot (\cdots ( y \cdot (y \cdot \tilde\beta ) ) ) \cdot y ) \cdots ) \cdot y),  \]

where $\tilde\beta \in \brac(\beta)$ and the decomposition, $\decomp_{\tilde\beta}$, of $\tilde\beta$ is acyclic. We can see that $\tilde\alpha$ is decomposed $\decomp_{\tilde\alpha} \in \conclog\noconstr$ which is acyclic since $\tilde\beta$ is acyclic, and we are repeatedly prepending $y$ symbols, before repeatedly appending $y$ symbols. Therefore, $\decomp_{\tilde\alpha}$ is $y$-localized and $x'$-localized for all $x' \in \var(\decomp_{\tilde\beta})$. Furthermore, we have that $(z \logeq z' \cdot y)$, for some $z' \in \Xi$, is an atom of the decomposition. 

We now prove the only if direction. Let $\decomp_{\tilde\alpha} \in \conclog\noconstr$ be an acyclic decomposition of $(z \logeq \alpha)$ such that some atom of $\decomp_{\tilde\alpha}$ contains the variables $z$ and $y$. Let  $\mathcal{T} \df (\mathcal{V}, \mathcal{E}, <, \Gamma, \labelFunction, v_r)$ be the concatenation tree for $\tilde\alpha \in \brac(\alpha)$, where $\decomp_{\tilde\alpha}$ is the decomposition of $\tilde\alpha$. Since $z$ only appears in the root atom of $\decomp_{\tilde\alpha}$, we know that for $y$ and $z$ to appear in the same atom, the root atom of $\decomp_{\tilde\alpha}$ must contain the variable $y$ (\ie, the root atom is either $(z \logeq y \cdot z')$ or $(z \logeq z' \cdot y)$ for some $z' \in \var(\decomp_{\tilde\alpha})$). It therefore follows that there exists $\{v_1, v_2\}, \{v_1, v_3\} \in \mathcal{E}$, where $v_2 < v_3$, such that $\tau(v_1) = z$ and either $\tau(v_2) = y$ or $\tau(v_3) = y$ and where $v_1 \in \mathcal{V}$ is the root of the concatenation tree. Let $\mathcal{T}_y$ be the induced sub-tree of $\mathcal{T}$ which contains only $y$-parents along with their children. We know that $\mathcal{T}_y$ is a connected since $\decomp$ is $y$-localized since $\decomp_{\tilde\alpha}$ to be acyclic. We also know that the root of the tree is a $y$-parent. Thus, each $y$ can only contribute to the prefix or suffix of $\alpha$ and hence $\alpha = y^i \cdot \beta \cdot y^j$ where $|\beta|_y = 0$ must hold. See~\cref{fig:subConcTree} for an example of $\mathcal{T}_y$. 

\begin{figure}
\begin{tikzpicture}[shorten >=1pt,->]
\tikzstyle{vertex}=[rectangle,fill=white!25,minimum size=12pt,inner sep=2pt]
\node[vertex] (1) at (0,0) {$v_1 \; (z)$};
\node[vertex] (2) at (2,-1)   {$v_2 \; (y)$};
\node[vertex] (3) at (-2,-1)  {$v_3 \; (z_1)$};
\node[vertex] (4) at (-3,-2) {$v_4 \; (z_2)$};
\node[vertex] (5) at (-1,-2) {$v_5 \; (y)$};
\node[vertex] (6) at (-4,-3) {$v_6 \; (y)$};
\node[vertex] (7) at (-2,-3) {$v_7 \; (z_2)$};
\node[vertex] (8) at (-1,-4){$\vdots$};
\node[vertex] (9) at (-3,-4){$\vdots$};

\path [-](1) edge node[left] {} (2);
\path [-](1) edge node[left] {} (3);
\path [-](3) edge node[left] {} (4);
\path [-](3) edge node[left] {} (5);
\path [-](4) edge node[left] {} (6);
\path [-](4) edge node[left] {} (7);
\path [-](7) edge node[left] {} (8);
\path [-](7) edge node[left] {} (9);
\end{tikzpicture}
\caption{\label{fig:subConcTree}A diagram of $\mathcal{T}_y$ used to illustrate the proof of~\cref{lemma:atomDecomp}.}
\end{figure}
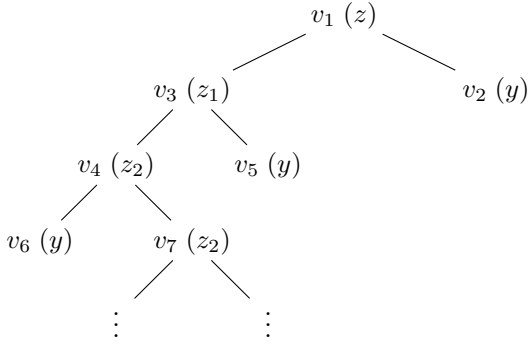

Therefore, to decide whether $(z \logeq \alpha)$ can be decomposed into an acyclic formula $\decomp \in \conclog\noconstr$ such that there exists an atom of $\decomp$ which has the variables $z$ and $y$, it is sufficient to decide whether $\alpha = y^i \cdot \beta \cdot y^j$ where $\beta$ is acyclic and $|\beta|_y =0$. 
This can obviously be decided in $\bigO(n^7)$ time by removing the prefix $y^i$ and the suffix $y^j$ in linear time, then checking whether $\beta$ is acyclic. 
Note that there can exist exactly one element of $C$ which contains the variable $z$, due to the fact that if two elements of $C$ are not disjoint, then we can decide that $\decomp$ does not exist. 
Therefore, after we have dealt with this case, we can continue with the procedure defined in~\cref{constrainedBracketings} to determine whether whether there is an acyclic decomposition $\decomp \in \conclog\noconstr$ of $\varphi$ such that for every $\{ x,y \} \in C$, there exists an atom $(z_1 \logeq z_2 \cdot z_3)$ of $\decomp$ where $\{ x,y \} \subseteq \var(z_1 \logeq z_2 \cdot z_3)$ in $\bigO(n^7)$.
\end{proof}

\section{Proof of~\autoref{theorem:LVJoinTree}}
\LVJoinTree*
\begin{proof}
Let $\varphi \df \cqhead{\vec{x}} \bigwedge_{i=1}^{m} \eta_i$ be a normalized $\cpfc\noconstr$, where $\eta_i \df (x_i \logeq \alpha_i)$ for all $i \in [m]$. We first rule out some cases where $\varphi$ must be cyclic (see~\cref{lemma:CyclicConditions}):
\begin{enumerate}
\item If $\varphi$ is weakly cyclic, then return ``$\varphi$ is cyclic'', otherwise let $T_w \df (V_w, E_w)$ be a weak join tree for $\varphi$.
\item If there exists $\{ \eta_i , \eta_j \} \in E_w$ such that $|\var(\eta_i) \intersect \var(\eta_j)| > 3$ then return ``$\varphi$ is cyclic''.
\item If there exists an edge $\{ \eta_i, \eta_j\} \in E_w$ where $|\var(\eta_i) \intersect \var(\eta_j)| = 3$ and $|\eta_i| > 3$ or $|\eta_j| > 3$, then return ``$\varphi$ is cyclic''.
\end{enumerate}

We then label every edge, $e \in E_w$, with the set of variables that the two endpoints share. 
For every atom $\eta_i$ of $\varphi$, we create the set $C_i \in \mathcal{P}(\Xi)$. 
We define $C_i$ by considering every outgoing edge of $\eta_i$ in $T_w$, and taking a union of the sets that label of those edges. We now give a construction to find an acyclic decomposition, $\decomp_\varphi \in \conclog\noconstr$, of $\varphi$, if one exists.

If $|C_i|=0$, then let $\decomp_i$ be any acyclic decomposition of $\eta_i$. 
If $\mathsf{max}_{k \in C_i} (|k|) = 1$ then let $\decomp_i$ be any acyclic decomposition of $\eta_i$. 
If $\mathsf{max}_{k \in C_i} (|k|) = 2$ then we can use~\cref{lemma:atomDecomp} to obtain the acyclic decomposition $\decomp_i$ of $\eta_i$ such that for all $k \in C_i$ where $|k| = 2$, there is an atom of $\decomp_i$ which contains the variables of $k$. 
If $\mathsf{max}_{k \in C_i} (|k|) = 3$ then we know that $|\eta_i| \leq 3$, and therefore $\decomp_i = \eta_i$ (see~\cref{lemma:CyclicConditions}). 
\begin{claim}
If there does not exist an acyclic decomposition $\decomp_i \in \conclog\noconstr$ of $\eta_i$ such that for all $k \in C_i$ where $|k| = 2$, there is an atom of $\decomp_i$ which contains all the variables of $k$, then $\varphi$ is cyclic. 
\end{claim}
\begin{claimproof}
We prove this claim by working towards a contradiction. Assume that there exists $\decomp_\varphi \in \conclog\noconstr$ which is an acyclic decomposition of $\varphi$, and that there exists two atoms $\eta_i$ and $\eta_j$ such that there does not exist an acyclic decomposition $\decomp_i$ of $\eta_i$ where some atom of $\decomp_i$ is of the form $(z \logeq x \cdot y)$, where $\var(\eta_i) \intersect \var(\eta_j)  = \var(z \logeq x \cdot y) \intersect \var(\eta_j) $.

Let $T \df (V,E)$ be the join-tree for $\decomp_\varphi$. We know from~\cref{lemma:subtree} that there exists a sub-tree of $T$ which is a join tree for the decompositions of $\eta_i$ and $\eta_j$. Let $T^i$ be the sub-tree of $T$ which represents a join-tree for $\decomp_i$ (the decomposition of $\eta_i$), and let $T^j$ be the sub-tree of $T$ which is a join-tree for $\decomp_j$ (the decomposition of $\eta_j$). Let $p$ be the shortest in path in $T$ from some node in $T^i$ to some node in $T^j$. Because $T$ is a tree, this path is uniquely defined. However, there does not exist a node $(z \logeq x \cdot y)$ of $T^j$ such that $\var(\eta_i) \intersect \var(\eta_j) = \var(z \logeq x \cdot y) \intersect \var(\eta_j) $.
Therefore, there is some variable $z' \in \var(\eta_i) \intersect \var(\eta_j)$ where $z'$ is not a variable of every atom on the path $p$. Therefore $T$ is not a join tree.
\end{claimproof}	

Once we have an acyclic formula $\decomp_i \in \conclog\noconstr$ for all $i \in [m]$, we can define $\decomp_\varphi \in \conclog\noconstr$ as an acyclic decomposition of $\varphi$ as $\decomp_\varphi \df \cqhead{\vec{x}} \bigwedge_{i=1}^{m}  \decomp_i $. 

\proofsubparagraph{Complexity.} We now prove that given the normalized formula, $\varphi \in \cpfc\noconstr$, we can decide whether $\varphi$ is acyclic in polynomial time. 
First, construct a weak join tree for $\varphi$, which takes polynomial time using the GYO algorithm, and we label each edge with the variables that the two end points of that edge share (which takes $\bigO(|\varphi|^2)$ time). 
We then find an acyclic decomposition of each $\eta_i$ in polynomial time using~\cref{polytime}, and if $\eta_i$ shares two variables with another atom we use~\cref{lemma:atomDecomp} to find an acyclic decomposition in polynomial time. 
Since there are $\bigO(\formulaSize{\varphi})$ atoms of $\varphi$, constructing the decomposition, $\decomp_i$, for all atoms, $\eta_i$, of $\varphi$ takes $\bigO(|\varphi||\eta_{\mathsf{max}}|^7)$ time, where $\eta_{\mathsf{max}}$ is the largest $|\eta_i|$ of any $i \in [m]$. Then, let $\decomp_\varphi$ have the body $\bigwedge_{i=1}^{m} \decomp_i$. This last step can be done in time $\bigO(|\varphi|)$. Therefore, in time $\bigO(|\varphi| |\eta_{\mathsf{max}}|^7)$, we can construct the acyclic formula $\decomp_\varphi$.
Since $|\eta_{\mathsf{max}}| = |\varphi|$ when $m=1$, we get the final running time of $\bigO(|\varphi|^8)$.
While $\varphi$ is not necessarily normalized, we know from~\cref{lemma:normalization} that normalizing $\varphi$ can be done in $\bigO(|\varphi|^2)$. 
Therefore, this does not affect the complexity claims of this lemma.

\proofsubparagraph{Correctness.} To prove that $\decomp_\varphi$ is acyclic, we construct a join tree for $\decomp_\varphi$ using the a weak join tree $T_w \df (V_w, E_w)$ as the skeleton tree. Let $T^i \df (V^i, E^i)$ be a join tree for $\decomp_i$ for each $i \in [m]$. We now construct a join tree for $\decomp_\varphi$. Let $T \df (V,E)$ be a forest where $V \df \bigcup_{i=1}^{n} V^i$ and let $E \df \bigcup_{i=1}^{n} E^i$. We add an edge $\{ \chi_i, \chi_j \} \in E$ between some $\chi_i \in V^i$ and $\chi_j \in V^j$ if and only if $\{ \eta_i, \eta_j \} \in E_w$ and $\var(\chi_i) \intersect\var(\chi_j) = \var(\eta_i) \intersect \var(\eta_j)$. Since all atoms of $\decomp_\varphi$ are nodes of $V$, and $T$ is a tree, to show that $T \df (V, E)$ is a join tree, it is sufficient to prove that for any $\chi, \chi' \in V$ where there exists some $x \in \var(\chi) \intersect \var(\chi')$, every node that lies on the path between $\chi$ and $\chi'$ in $T$, contains the variable $x$. The proof of this is analogous to the proof of~\cref{lemma:skeletonTree}, however we include the proof here for completeness sake.

Without loss of generality, assume that $\chi \in V^1$ and $\chi' \in V^k$ where $V^1$ and $V^k$ are the set of vertices for the join tree for the decompositions of $\eta_1$ and $\eta_k$ respectively. Further assume that the path from $\eta_1$ to $\eta_k$ in $T_w$ consists of $\{\eta_i, \eta_{i+1} \}$ for all $i \in [k-1]$. Since we know that $T_w$ is a weak join tree, and that $\eta_1$ and $\eta_k$ both contain the variable $x$, it follows that for all $i \in [k]$, the word equation $\eta_i$ contains the variable $x$. We know that for any any edge $\{ \chi_i, \chi_{i+1} \} \in E$, where $\chi_i \in V^i$ and $\chi_{i+1} \in V^{i+1}$, that $\var(\chi_i) \intersect \var(\chi_{i+1}) = \var(\eta_i) \intersect \var(\eta_{i+1})$, and hence $x \in \var(\chi_i) \intersect \var(\chi_{i+1})$. Since any path between any two nodes of $V^i$ which share the variable $x$, for some $i \in [m]$, all the nodes on the path between them also contain the variable $x$ (due to the fact that $T^i \df (V^i, E^i)$ is a join tree for $\decomp_i$), and that for any edge $\{ \chi_i, \chi_{i+1} \} \in E$, where $\chi_i \in V^i$ and $\chi_{i+1} \in V^{i+1}$, we know that $x \in \var(\chi_i) \intersect \var(\chi_{i+1})$, it follows that all nodes on the path between $\chi$ and $\chi'$ contain the variable $x$. Therefore, $T \df (V, E)$ is a join tree for the decomposition $\decomp_\varphi$ of $\varphi$.
\end{proof}

\section{Proof of~\autoref{corollary:enumerationAndEvaluation}}
\enumAndEval*
\begin{proof}
For each word equation $\chi$ of $\decomp_{\varphi'}$, we can enumerate $\fun{\chi}\strucbra{w}$ in time $\bigO(|w|^3)$, since $\chi = (x_1 \logeq x_2 \cdot x_3)$, or $\chi = (x_1 \logeq x_2)$ for some $x_1, x_2, x_3 \in \Xi$. For every regular constraint $(x \regconst \gamma)$ of $\decomp_{\varphi'}$, we can enumerate $\fun{(x \regconst \gamma)}\strucbra{w}$ in polynomial time, since there are $\bigO(|w|^2)$
factors of $w$, and for each factor, the membership problem for regular expressions can be solved in polynomial time~(see Theorem 2.2 of~\cite{jiang1991note}). Since there are $\bigO(|\varphi|)$ atoms of $\decomp_{\varphi'}$, computing $\fun{\chi}\strucbra{w}$ for each atom of $\decomp_{\varphi'}$ takes time $\bigO(|\varphi| \cdot |w|^3)$. Then, we can proceed with the model checking problem and enumeration of results identically to relational $\cq$s.

The problem of model checking and enumeration reduces, in polynomial time, to the equivalent problems for standard relational acyclic conjunctive queries using the procedure we have just described. Therefore, since the model checking problem for relational acyclic conjunctive queries can be solved in polynomial time~\cite{gottlob2001complexity}, we can decide the model checking problem for acyclic $\cpfcreg$s in polynomial time. Furthermore, because we can enumerate the results of relational acyclic $\cq$s with polynomial delay, see~\cite{bagan2007acyclic}, we can enumerate $\fun{\varphi}\strucbra{w}$ with polynomial delay. We note that our database is of size $\bigO(\formulaSize{\varphi} \cdot |w|^3)$. 
\end{proof}

Note that this approach to model checking leaves room for a small optimization: Assume we are dealing with a word equation $\chi$ and regular constraint $(x \regconst \gamma)$ for some variable $x\in\var(\chi)$. 
Instead of computing $\fun{\chi}\strucbra{w}$ and $\fun{(x \regconst \gamma)}\strucbra{w}$ separately and then joining them, we include the check if $\subs(x)\in \lang(\gamma)$ in the enumeration of $\fun{\chi}\strucbra{w}$. 

That is, instead of constructing a relation with $\bigO(|w|^3)$ and a relation with $\bigO(|w|^2)$ elements and then combining them, we construct $\fun{\chi\land (x \regconst \gamma)}\strucbra{w}$ directly.
While this does not lower the time complexity -- as we still need to iterate over $\bigO(|w|^3)$ factors of~$w$ -- we can avoid construction unnecessary intermediate tables.

\section{Proof of~\autoref{prop:quasiAcyclic}}
\QuasiAcyclicSpanners*
\begin{proof}
Let $\query \df \pi_Y \left( \select^=_{x_1,y_1} \select^=_{x_2,y_2} \cdots \select^=_{x_m, y_m} \left( \gamma_1 \join \gamma_2 \cdots \join \gamma_k \right) \right)$ where each $i \in [k]$, we have that $\gamma_i \df \beta_{i_1} \cdot \bind{x_i}{\beta_{i_2}} \cdot \beta_{i_3}$ for $x_i \in \Xi$ and where $\beta_{i_1}$, $\beta_{i_2}$, and $\beta_{i_3}$ are regular expressions. 
We know define $\varphi_\query \in \cpfcreg$ such that $\varphi_\query$ is acyclic. 
\begin{itemize}
\item For every variable $x_i \in \SVars{\query}$, we add $(\strucvar \logeq x_i^P \cdot z_i)$ and $(z_i \logeq x_i^C \cdot x_i^S)$ to $\varphi_\query$.
\item For every $\gamma_i$ for $i \in [k]$, we add $(x_i^P \regconst \beta_1)$, $(x_i^C \regconst \beta_2)$ and $(x_i^S \regconst \beta_3)$ to $\varphi_\query$.
\end{itemize}

Since for any $\gamma_i$ and $\gamma_j$ for $1 \leq i, j \leq k$ where $i \neq j$ the word equations we add to $\varphi_\query$ are disjoint, it follows that $\varphi_\query$ is (so far) acyclic. 
Furthermore, $\varphi_\query$ remains acyclic after adding the regular constraints since they are unary. Next, we deal with string equality.

Let $G_\select \df (V_\select, E_\select)$ be a graph where $V_\select \df \{ x_i, y_i \mid i \in [m] \}$ and $E_\select \df \{ \{x_i, y_i \} \mid i \in [m] \}$.
Let $F_s \df (V_s, E_s)$ be a spanning forest of $G_\select$. 
For every $\{ x_i, y_i \} \in E_s$, we consider the directed edge $(x_i, y_i)$, and add the word equation $(x_i^C \logeq y_i^C)$ to $\varphi_\query$. 
Finally, for every $x \in Y$, where $Y$ is the set of variable in the projection, we add $x^P$ and $x^C$ to the head of $\varphi_\query$.

\subparagraph*{Complexity} 
First, we add two word equations to $\varphi_\query$ for every $x \in \SVars{\query}$, and for each $i \in [k]$, we add three regular constraints to $\varphi_\query$. 
Then, we create a string equality graph, and find a spanning forest of this graph. 
Finally, for every edge we add a word equation to $\varphi_\query$. 
Since it takes polynomial time to execute each of these steps, it follows that we can construct $\varphi_\query$ in polynomial time.

\subparagraph*{Correctness} 
To show that $\varphi_\query$ is acyclic, we construct a join tree. 
For each tree of $F_\select$, let an arbitrary node be the root and assume all edges are directed away from the root. 
Then, for each node $n$ we create an undirected line graph $L_n$ containing nodes $(n \logeq n')$ for all $n' \in F_\select$ where $( n, n' ) \in E_\select$, where $E_\select$ is the set of edges of $F_\select$. 
If $(n, n') \in E_\select$, then we find a node of $L_n$ containing the variable $n'$ and a node in $L_{n'}$ containing $n$ and add an edge between them -- since $(n,n') \in E_\select$, such an edge must exist. 
This results is a new forest, $F \df (G, E)$. 
Pick one node in each tree in $F$, and add edges between these nodes so that no cycles are introduced. 
It follows that $F$ is now a join tree for $\bigwedge_{i=1}^k (x_i^C \logeq y_i^C)$. 
For each variable $x_i \in \SVars{\query}$, we add the nodes $(\strucvar \logeq x_i^P \cdot z_i)$ and $(z_i \logeq x_i^C \cdot x_i^S)$ to $F$, and add an edge between $(\strucvar \logeq x_i^P \cdot z_i)$ and $(z_i \logeq x_i^C \cdot x_i^S)$, and an edge between any node of some $L_n$ that contains $x_i^C$ and $(z_i \logeq x_i^C \cdot x_i^S)$. 
Finally, we incorporate every regular constraint into the tree -- which can easily be done since a regular constraint is unary. 
Therefore, we have a join tree for $\varphi_\query$, and hence $\varphi_\query$ is acyclic.
\end{proof}

\end{document}